\documentclass{article}

\newcommand{\dref}{\mathbin{{\sf dref}}}
\newcommand{\linkto}{\mathbin{{\sf link}}}

\newcommand{\Guard}[1]{{[}#1{]}}

\def \cdota{\!\cdot\!}
\def \hasgn{\asgn}

\def \deref{\mathop{*}} 
\def \derefe{\mathop{*\!}}
\newcommand{\prev}{{\bf prev}}
\usepackage{amsthm}
\theoremstyle{plain}
\newcounter{thm}
\newtheorem{theorem}{Theorem}[section]
\newtheorem{lemma}[thm]{Lemma}

\theoremstyle{definition}
\newtheorem{definition}{Definition}[section]

\newtheorem{example}{Example}[section]
\usepackage{amsmath}
\usepackage{mine}
\usepackage{timestamp}
\usepackage{oz}
\usepackage{zed-csp}
\usepackage{rcp}
\usepackage{graphicx}
\usepackage{color}
\usepackage{ra4}
\usepackage{stmaryrd}

\def \Eval{{\sf eval}} 
 
\def \Update{{\sf update}}
\def \stable{{\sf stable}}

\newcommand{\ceil}[1]{\lceil #1 \rceil}

\newcommand{\Idle}{{\sf Idle}}

\newcommand{\False}{{\sf False}}
\newcommand{\Fin}{{\sf Fin}}
\newcommand{\Inf}{{\sf Inf}}

\renewcommand{\qed}{\ensuremath{{}_\Box}}
\newcommand{\lub}{{\sf lub}}
\newcommand{\glb}{{\sf glb}}

\newcommand{\True}{{\sf True}}
\renewcommand{\Chaos}{{\sf Chaos}}

\def\Return{\mathop{\textbf{return}}}
\newcommand{\st}{~{\scriptscriptstyle ^\bullet}~}

\def\figrule{\rule{\columnwidth}{0.5pt}}
\def\uop{\ominus}
\def\bop{\oplus}
\def \rely {\mathop{\textsc{Rely}}}

\def \enf {\mathop{\textsc{Enf}}}

\def\Init{\mathop{\textsc{Init}}}

\newcommand{\NoteEnv}[3]{\newenvironment{#1}{\par\color{#3}#2: }{}}
\definecolor{brijeshcolor}{rgb}{0,0,1}
\definecolor{iancolor}{rgb}{1,0,0}
\definecolor{lindsaycolor}{cmyk}{0.2,1,0,0}
\definecolor{johncolor}{rgb}{1,0,0}

\NoteEnv{brijesh}{Brijesh}{brijeshcolor}
\NoteEnv{ian}{Ian}{iancolor}
\NoteEnv{lindsay}{Lindsay}{lindsaycolor}
\NoteEnv{john}{John}{johncolor}

\def \llb {\llbracket}
\def \rrb {\rrbracket}
\def\abssynt{\mathop{:\joinrel:\joinrel=}}

\newcommand{\Context}[2]
{\!\left\llb #1 \!\left/
    \begin{array}[c]{@{}l@{}}
      #2
    \end{array}\right.
  \right\rrb }

\newcommand{\Enf}[2]{
  \enf
  \left(\begin{array}[c]{@{}l@{}}
    #1
  \end{array}\right)
  \st 
  \begin{array}[c]{@{}l@{}}
    #2
  \end{array}}

\newcommand{\Rely}[2]
{\rely #1 \st
    \begin{array}[c]{@{}l@{}}
      #2
    \end{array}}

\newcommand{\Code}[1]
{\begin{array}[c]{@{}l@{}}
      #1
    \end{array}}

\newcommand{\Par}{\textstyle\mathop{\|}}

\newcommand{\Empty}{{\sf Empty}}

\newcommand{\Always}{\textstyle\mathord{\boxdot}}
\newcommand{\Eventually}{\rotatebox[origin=c]{45}{$\textstyle\boxdot$}}

\def\adjoins{\mathbin{\varpropto}}
\newcommand{\adj}[1]{\adjoins_{_{\!#1}}}
\newcommand{\adjd}[1]{\adj{\Delta}}

\def\ch{\mathbin{;}}

\def \kif{\mathop{\mathsf{if}}}
\def \kthen{\mathbin{\mathsf{then}}}
\def \kelse{\mathbin{\mathsf{else}}}

\def \klet {\mathop{{\sf let}}}
\def \kin {\mathbin{{\sf in}}}

\def \bs {\backslash}

\def \dom {\mathrm{dom}}

\def \seq {\mathrm{seq}}

\DeclareMathSymbol{\Box}{\mathord}{lasy}{"32}
\DeclareMathSymbol{\Diamond}{\mathord}{lasy}{"33}
\def \prev {{\sf prev}}
\def \next {{\sf next}}

\begin{document}

\title{Proving linearisability via coarse-grained abstraction}

\author{Brijesh Dongol \qquad John Derrick\\ \\
  Department of Computer Science, \\
  The University of
  Sheffield, S1 4DP, \\
  United Kingdom \\
  \texttt{B.Dongol@sheffield.ac.uk,
    J.Derrick@dcs.shef.ac.uk} 
  % \timestamp
}

% \author{Brijesh Dongol \and John Derrick}

% \institute{Brijesh Dongol \and John Derrick \at Department of Computer Science, The University of
%   Sheffield, S1 4DP, United Kingdom \\
%   \email{B.Dongol@sheffield.ac.uk,
%     J.Derrick@dcs.shef.ac.uk} 
%   \timestamp
% }

\maketitle

\begin{abstract}
  Linearisability has become the standard safety criterion for
  concurrent data structures ensuring that the effect of a concrete
  operation takes place after the execution some atomic statement
  (often referred to as the linearisation point). Identification of
  linearisation points is a non-trivial task and it is even possible
  for an operation to be linearised by the execution of other
  concurrent operations. This paper presents a method for verifying
  linearisability that does not require identification of
  linearisation points in the concrete code. Instead, we show that the
  concrete program is a refinement of some coarse-grained
  abstraction. The linearisation points in the abstraction are
  straightforward to identify and the linearisability proof itself is
  simpler due to the coarse granularity of its atomic statements. The
  concrete fine-grained program is a refinement of the coarse-grained
  program, and hence is also linearisable because every behaviour of
  the concrete program is a possible behaviour its abstraction.
\end{abstract}

\section{Introduction}

With the increasing prevalence of concurrent computation in modern
systems, development of concurrent data structures that enable a
greater degree of parallelism have become increasingly important. To
improve efficiency, programs that implement concurrent data structures
often exhibit fine-grained atomicity and use atomic non-blocking
compare-and-swap operations as their main synchronisation primitive. A
consequence of these features is the increase in complexity of the
programs, making their correctness harder to judge. Hence, formal
verification of programs for concurrent data structures is known to be
a necessity. There are even examples of errors being uncovered by
formal verification in published algorithm that were previously
believed to be correct \cite{Doherty03}.

The main correctness criterion for programs that implement concurrent
data structures is linearisability \cite{Herlihy90}, which allows one
to view operations on concurrent objects as though they occur in some
sequential order. Over the years, numerous approaches to verifying
linearisability have been developed using a variety of different
frameworks, and several of these approaches are partially/fully
mechanised.  Herlihy and Wing's original paper use the notion of a
possibilities mapping, which defines the set of possible abstract data
structures that corresponds to each point of interleaving. Doherty et
al \cite{CDG05,DGLM04} use a simulation-based method using
input/output automata with proofs mechanised using PVS. Vafeiadis et
al use a framework that combines separation logic and rely/guarantee
reasoning \cite{Vaf07,VHHS06}. An automated method based on this
theory has been developed, but the method is known not to apply to a
more complex programs \cite{Vaf10}. Derrick et al have developed
refinement-based methods that have been mechanised using the theorem
prove KIV \cite{DSW07,DSW11TOPLAS,DSW11}. Turon and Wand propose a
compositional rely/guarantee framework with separation logic to show
that the concrete programs implement another so-called ``obviously
correct'' program, which may or may not be linearisable \cite{TW11}.
Several other verification methods have been proposed which we discuss
in more detail in \refsec{sec:concl-relat-work}.  % The methods in this
% paper differ in that our framework models the true concurrency between
% the parallel processes. This decision was taken because of the
% inherent true parallelism in modern many/multicore systems.

Each of the existing methods above build on the fact that
linearisability guarantees the existence of a so-called
\emph{linearisation point}, which is an atomic statement whose
execution causes the effect of an operation to be
felt.
\begin{quote}
  ``Linearisability provides the illusion that each operation applied
  by concurrent processes takes effect instantaneously at some point
  between its invocation and its response.'' \cite{Herlihy90}
\end{quote}
Hence, the methods in
\cite{CDG05,DSW07,DSW11TOPLAS,DSW11,DGLM04,Vaf07,Vaf10,VHHS06} involve
identification of linearisation points in the concrete code, and a
proof that execution of a linearisation point does indeed correspond
to an execution of the corresponding abstract operation. However, in
many sophisticated programs, the linearisation points are not always
immediately identifiable, and often require a high degree of expertise
on the proof techniques as well as the program at hand. In more
complex algorithms, it is even possible for some operations to be
linearised by the execution of other concurrent operations and hence
require the use backwards reasoning techniques
\cite{CDG05,DSW07,DSW11,Sha11,Vaf07,VHHS06}. Algorithms that require
backwards reasoning are precisely those that cause difficulties for
the method described in \cite{Vaf10}.

We present a method for verifying linearisability where we aim to
establish a relationship between a fine-grained concrete program and a
program in which the operations execute with coarse-grained
atomicity. In particular, the fine-grained program is one that
implements the one with coarse-grained atomicity.  Groves \cite{Gro08}
and separately Elmas et al \cite{EQSST10} start with a coarse-grained
program, which is incrementally refined to an implementation with
finer-grained atomicity. Splitting the atomicity of a statement is
justified using \emph{reduction} \cite{Lip75}, which ensures that the
operations are immune to interleavings with other concurrent
statements.  However, because one must consider each pair of
interleavings, reduction-based methods are not compositional, and
hence do not scale well as the complexity of an operation increases.

We develop a framework that enables reasoning over the interval time
in which a program executes, presenting an alternative to traditional
reasoning over the pre/post states of a program and captures the
possible interference that may occur during a program's execution
\cite{DDH12}. Our model incorporates fractional permissions
\cite{Boy03} to simplify reasoning about conflicting accesses to
shared variables \cite{DDH12}. Permissions are also used to model
properties such as interference freedom and locality. Our framework
also incorporates reasoning about pointer-based programs by allowing
the domains of each state is assumed to consist of variables and
addresses, and take extra care when updating or evaluating pointers
because the values at their addresses may be dynamically changing. The
behaviour of a command is defined over an interval, and hence, for
example, expression evaluation is assumed to take a number of
steps. Note though that we do not take into account all the
complexities of non-deterministic expression evaluation \cite{HBDJ11}.

We develop interval-based theories refinement, namely \emph{behaviour
  refinement} (which is akin to operation refinement) and \emph{data
  refinement} (which allows the state spaces of two programs to be
linked using a simulation predicate). As far as we are aware, our
formulation of data refinement over intervals is novel to this
paper. Data refinement is used to link the abstract representation of
the data structure (in which each operation takes place in a single
atomic step) and the coarse-grained abstraction (in which operations
execute on the same data representation as the implementation, but
with a coarse-grained atomicity). Like Derrick et al
\cite{SWD12,DSW11,DSW11TOPLAS}, the data refinement proof encapsulates
a proof that the coarse-grained program is linearisable with respect
to the abstract representation. Then to show that the final
implementation is linearisable, we show that the it is a behaviour
refinement of the coarse-grained abstraction. We use Treiber's Stack
\cite{Tre86} to illustrate our approach.

Interval-based methods for proving linearisability have also been
proposed by Baumler et al \cite{BSTR11}. However, their model assumes
that a program executes with its environment by interleaving the
statements of a program with those of its environment, as opposed to
our model, that allows true concurrency, which allows one to model the
inherent true parallelism in modern many/multicore systems.
Furthermore, Baumler et al prove linearisability of the concrete
program directly unlike our method in which we first show that the
fine-grained (concrete) implementation refines a coarse-grained
abstraction. Due to the coarse-granularity of the statements in the
second program, its linearisation points are straightforward to
identify, and the linearisation proof itself is simpler.

This paper is structured as follows. In
\refsec{sec:linearisability-1}, we present a formalisation of
linearisability and an alternative definition that simplifies its
proof. In \refsec{sec:exampl-treib-stack}, we present the Treiber
stack, which we use as a running example throughout the paper. We
present our interval-based framework in \refsec{sec:an-interval-based}
and use it to define a semantics for a language that allows explicit
control of a program's atomicity
(\refsec{sec:interv-based-semant}). We also develop a theory for
refining the behaviour of commands within this framework. In
\refsec{sec:coarse-grain-line}, we develop the coarse-grained
abstraction of the Treiber stack, present methods of verifying its
linearisability using data refinement, and present the actual proof of
its correctness. As part of this, we develop data refinement rules for
our interval-based framework. We then develop methods for showing that
a fine-grained program implements a coarse-grained abstraction. To
this end, we develop compositional rely/guarantee-style rules in
\refsec{sec:compositional-proofs} and some high-level transformation
rules specific to CAS-based implementations in
\refsec{sec:transformation-rules}. We apply these to prove that the
Treiber stack implements the coarse-grained abstraction in
\refsec{sec:verification}.

\section{Linearisability}
\label{sec:linearisability-1}

In this section, we present the original definition of linearisability
\cite{Herlihy90} and an alternative definition that simplifies proofs
of linearisability. 

Two operations $opi$ and $opj$ of a concurrent program are said to
execute \emph{concurrently} iff the invocation of $opi$ occurs after
the invocation but before the response of
$opj$. % Operations $opi$ and $opj$ are said to execute
During a concurrent execution of two or more operations, the atomic
statements of the operations may be arbitrarily interleaved. As a
result, the effect of two concurrent operations may take place in any
order and does not correspond to the ordering of invocations and
responses. For example, \reffig{fig:stack-hist-ex} depicts a scenario
where process $r$ linearises before process $p$ even though the
invocation of $p$ occurs before the invocation of $r$. Similarly,
process $q$ linearises after process $q$ even though the response of
process $q$ occurs before the response of process $p$. 

Not every ordering of invocations and responses is linearisable. In
particular, linearisability requires that a chosen ordering of effects
of the concurrent operations corresponds to a valid sequential
history. For example, assuming that we start with an empty stack,
history in \reffig{fig:stack-hist-ex} is linearisable whereas
\reffig{fig:stack-hist-ex2} is not \cite{DSW11TOPLAS}. In particular,
history \reffig{fig:stack-hist-ex} can be linearised by selecting
linearisation points marked by the crosses. In
contrast, there is no possible selection of linearisation points for
\reffig{fig:stack-hist-ex2} that results in a valid sequential history because
processes $q$ and $r$ both return $x$, even though there is only one
concurrent $push(x)$ operation and execution started with an empty
stack.  We give a more formal presentation of these concepts in
\refex{ex:linlin}.

\subsection{Herlihy and Wing's definition}
\label{sec:linearisability}

To formalise linearisability using the nomenclature of Herlihy and
Wing \cite{Herlihy90}, we let $\seq.X$ denote sequences of type
$X$. We assume sequences start with index $0$.  An event is a tuple
$Event \sdef Op \times Proc \times \{invoke,return\} \times \seq.Val$,
respectively corresponding to an operation identifier, a process
identifier, the type of the event (invoke or return), and a sequence
corresponding to the input/output parameters of the event. We use
$op_p^I(k_1, k_2, \dots k_n)$ and $op_p^R(k_1, k_2, \dots k_n)$ to
denote operations $(op, p, invoke, \aang{k_1, \dots, k_n})$ and $(op,
p, response, \aang{k_1, \dots, k_n})$, respectively. Notations
$op_p^I$ and $op_p^R$ denote an invoke and response events with no
parameters, respectively. A \emph{history} is a
sequence of invocation of response events. 

\begin{example}
  For a stack data structure, sequences \refeq{eq:23} and
  \refeq{eq:32} below are possible histories of invocation and
  response events, where $p$, $q$ and $r$ are pairwise distinct
  processes.
  \begin{eqnarray}
    \label{eq:23}
    & \langle push_p^I(x), push_q^I(y), pop_r^I, push_q^R, push_p^R,
    pop_r^R(Empty) \rangle
    \\
    \label{eq:32}
    &  \langle push_p^I(x), pop_q^I, pop_r^I, pop_q^R(x), push_r^R,
    pop_p^R(x)  \rangle
  \end{eqnarray}
  A visualisation of histories \refeq{eq:23} and \refeq{eq:32} are
  given in \reffig{fig:stack-hist-ex} and \reffig{fig:stack-hist-ex2},
  respectively.  \hfill ${}_\clubsuit$
\end{example}

For $H \in \seq.Event$ of
invocations and responses, $H\!|\!p$ denotes the subsequence of $H$
consisting of all invocations and responses of process $p$.  Two
histories $H_1$, $H_2$ are \emph{equivalent} if for all processes $p$,
$H_1\!|\!  p = H_2\!|\!p$.  An invocation $opi^I_p(x)$ \emph{matches}
a response $opj^R_q(y)$ iff $opi = opj$ and $p = q$. An invocation is
\emph{pending} in a history $H$ iff there is no matching response to
the invocation in $H$. We let $complete(H)$ denote the maximal
subsequence of history $H$ consisting of all invocations and matching
responses in $H$, i.e., the history obtained by removing all pending
invocations of $H$.

\begin{figure}[t]
  
  \hfill
  \begin{minipage}[t]{0.46\columnwidth}
    \scalebox{0.8}{\input{stack-hist.pspdftex}}
    \caption{History corresponding to \refeq{eq:23}}
    \label{fig:stack-hist-ex}
  \end{minipage}
  \begin{minipage}[t]{0.46\columnwidth}
    \scalebox{0.8}{\input{stack-hist2.pspdftex}}
    \caption{History corresponding to \refeq{eq:32}}
    \label{fig:stack-hist-ex2}
  \end{minipage}
  \hfill
\end{figure}

An operation $op$ in a history is defined by an invocation
$invocation.op$ followed by the next matching response $response.op$.
For a history $H$, $<_H$ is an irreflexive partial order on
operations where $opi <_H opj$ iff $response.opi$ occurs before
$invocation.opj$ in $H$, i.e., $opi$ and $opj$ do not execute
concurrently and $opi$ occurs before $opj$. 

\begin{definition}[Sequential history]
  A history $H$ is \emph{sequential} iff the first element of $H$ is an
  invocation and each invocation (except possibly the last) is
  immediately followed by its matching response.
\end{definition}

\begin{definition}[Linearisability \cite{Herlihy90}]
  \label{def:linearisability}
  A (concurrent) history $HC$ is \emph{linearisable} iff $HC$ can be
  extended to a (concurrent) history $HC'$ by adding zero or more
  matching responses to pending invocations such that $complete(HC')$
  is equivalent to some sequential history $HS$ and ${}<_{HC}
  {}\subseteq {}<_{HS}{}$.
\end{definition}

\begin{example}
  \label{ex:linlin}
  Using \refdef{def:linearisability}, assuming that the stack is
  initially empty, history \refeq{eq:23} is may be linearised by the
  following sequential history:
  \begin{eqnarray}
    \label{eq:49}
    \langle pop_r^I, pop_r^R(Empty), push_q^I(y),
    push_q^R , push_p^I(x), push_p^R \rangle 
  \end{eqnarray}
  Note that a single concurrent history may be linearised by several
  sequential histories. For example \refeq{eq:23} can also be
  linearised by the following sequential history, in which case the
  order of the linearisation points of process $p$ and $q$ shown in
  \reffig{fig:stack-hist-ex} would be swapped.
  \begin{eqnarray}
    \label{eq:50}
    \langle pop_r^I, pop_r^R(Empty), push_p^I(x), push_p^R, push_q^I(y),
    push_q^R \rangle
  \end{eqnarray}
  Linearising history \refeq{eq:23} using \refeq{eq:49} results in a
  final stack $\aang{y,x}$ with element $y$ at the top, whereas
  \refeq{eq:50} results in a final stack $\aang{x,y}$ with element $x$
  at the top.

  Unlike \refeq{eq:23}, there is no valid linearisation of
  \refeq{eq:32}.  \hfill ${}_\clubsuit$
\end{example}
The definition of linearisability allows histories to be extended with
matching responses to pending invocations. This is necesssary because
some operations may be past their linearisation point, but not
yet responded. For example consider the following stack history, where
the stack is initially empty.
\begin{eqnarray}
  \label{eq:87}
  \langle push_p^I(x), pop_q^I, pop_q^R(x) \rangle
\end{eqnarray}
The effect of the invocation $push_p^I(x)$ has clearly been felt in
\refeq{eq:87} because the $pop_q$ returns $x$, i.e., the linearisation
point of $push_p^I(x)$ occurs before that of $pop_q^I$. We can
validate this formally because \refeq{eq:87} can be extended with a
matching response to $push_p^I(x)$, then linearised by the following
sequential history
\begin{eqnarray*}
  \langle push_p^I(x), push_p^R, pop_q^I, pop_q^R(x) \rangle
\end{eqnarray*}

\subsection{An alternative definition of linearisability}
\label{sec:an-altern-defin}

Verifying linearisability (\refdef{def:linearisability}) by reasoning
at the level of histories of invocations and responses directly is
clearly infeasible. Hence, we follow the methods of Derrick et al who
link the concurrent and sequential histories using a matching function
\cite{SWD12,DSW11TOPLAS,DSW11}. This then allows one to prove
linearisability via data refinement. The key idea here is to
distinguish between invocations that have and have not linearised but
not yet returned. This allows one to determine whether or not a
concrete program contributes to the abstract history. 

One must first define matching pairs of events and pending invocation
events within a sequence of events. 
\begin{definition}[Matching pair, Pending invocation]
  \label{def-matching}
  For a sequence of events $H$, we say $i, j \in \dom.H$ forms a
  \emph{matching pair} in $H$ iff $mp_H(i,j)$ holds and $i$ is a
  \emph{pending invocation} in $H$ iff 
  $pending_H.i$ holds, where
  \begin{eqnarray*}
    mp_H.(i,j) & \sdef &
    \begin{array}[t]{@{}l@{}}
      (i < j) \land (H.i.opid = H.j.opid) \land
      (H.i.proc = H.j.proc)  \land \\
      (H.i.type = invoke) \land (H.j.type
      = response) \\
      (\all k \st i < k < j \imp H.k.proc
      \neq H.i.proc)
    \end{array}
    \\
    pending_H.i & \sdef & (H.i.type = invoke) \land (\all j : \dom.H
    \st j \geq i
    \imp \neg mp_H.(i, j))
  \end{eqnarray*}
\end{definition}  
Hence, two indices $i$ and $j$ form a matching pair in a sequence of
events $H$ iff $j$ follows $i$, events $H.i$, $H.j$ are invocations
and responses of the same operation by the same process and there are
not other invocations/responses by process $H.i.proc$ between $i$ and
$j$. Index $i$ of sequence $H$ is pending if there is no matching
index $j > i$.
\begin{example}
  Given that $H$ is the history corresponding to \refeq{eq:23},
  $mp_H.(0,4)$, $mp_H.(1,3)$ and $mp_H.(2,5)$ hold. If $h$ is the
  history corresponding to \refeq{eq:87}, $mp_{h}.(1,2)$ and
  $pending_{h}.0$ hold.  \hfill ${}_\clubsuit$
\end{example}
Using these, one may now define the set of legal histories. 
\begin{definition}[Legal history]
  \label{def-legalhist}
  A history $H$ is \emph{legal} iff $legal.H$ holds,
  where
  \begin{eqnarray*}
    legal.H & \sdef & \all i : \dom.H \spot
    \begin{array}[t]{@{}l@{}}
      \kif (H.i.type = invoke)
      \kthen (pending_H.i \lor \exists j : \dom.H \spot mp_H.(i, j))
      \\
      \kelse \ 
      (\exists j : \dom.H \spot mp_H.(j, i))
    \end{array}
  \end{eqnarray*}
\end{definition}
Hence, $H$ is legal iff for every index corresponding to an invocation
in $H$ is either pending or has a matching index, and indices
corresponding to responses have an earlier matching invocation. Using
this, one may now define a notion of a lin-relation between two
histories.  

\begin{definition}[Lin-relation]
  \label{def-lin-rel}
  A history $HC$ is said to be in a \emph{lin-relation} with history
  $HS$ with respect to a matching function $f$ iff $linrel(HC, f, HS)$
  holds, where
  \begin{eqnarray}
    \label{eq:1}
    linrel(HC, f, HS) & \sdef & f \in  \dom.HC \surj \dom.HS \land \\
    \label{eq:4}
    && (\all i, j: \dom.HC \spot mp_{HC}.(i,j)
    \imp \{i,j\} \subseteq \dom.f) \land \\
    \label{eq:55}
    && (\all i: \dom.HS \spot HC.i =
    HS.(f.i)) \land \\
    \label{eq:73}
    && (\all i, j : \dom.HC \spot i < j \land mp_{HC}.(i, j) % \land \neg pending_{HC}.i
    \imp f.j =
    f.i + 1) \land \\
    \label{eq:74}
    && (\all i, j, k, l : \dom.HC \spot j < k \land mp_{HC}.(i, j) \land
    mp_{HC}.(k, l) \imp f.j < f.k)
  \end{eqnarray}
\end{definition}
Hence, a concrete history $HC$ is in a lin-relation with sequential
history $HS$ with respect to linearising function $f$ iff $f$ is a
surjection (an onto function) mapping the indices of $HC$ to the
indices of $HS$ (\ref{eq:1}), every pair of indices of $HC$ that forms
a matching pair in $HC$ is in the domain of $f$ (\ref{eq:4}), for
every index $i$ of the sequential history $HC$, element $HC.i$ is the
same as element $HS.(f.i)$ \refeq{eq:55}, for every matching pair of
indices $i$ and $j$, $f.j$ is one greater than $f.i$ \refeq{eq:73},
and matching indices $i$, $j$ must occur before matching indices $k$,
$l$ if $j$ occurs before $k$ \refeq{eq:74}. 

We let $prefix.tt$ denote the set of all prefixes of a sequence $tt$
and hence $ss \in prefix.tt$ denotes that sequence $ss$ is a prefix of
$tt$. Using the definitions above, we now present Derrick et al's
alternative definition of linearisability \cite{DSW11TOPLAS}.
\begin{definition}[Linearisability]
  \label{def-linearisable}
  A (concurrent) history $HC$ is \emph{linearisable} with respect to a
  sequential history $HS$ iff $linearisable(HC, HS)$ holds, where
  \begin{eqnarray*}
    linearisable(HC, HS) & \sdef & \exists HE : \seq.Event
    \spot HC \in prefix.HE \land  
    legal.HE \land \exists f \st linrel(HE, f,  HS)
  \end{eqnarray*}
\end{definition}
Hence, concrete history $HC$ is linearisable with respect to
sequential history $HS$ iff $HC$ can be extended to a sequence of
(return) events $HE$ such that the exended history $HE$ is legal and
there exists a linearising function between $HE$ and $HS$. 

\section{Example: The Treiber Stack}
\label{sec:exampl-treib-stack}
We present our methods via a verification of the Treiber Stack as a
running example (see \reffig{fig:con-stack}), which is a well-known
program that implements a list-based concurrent stack
\cite{Tre86}. Verification of this stack has become a standard
exercise in the literature. The program uses a pointer ${\tt Top}$ of
type ${\tt ptr\_ctr}$ within which the pointer field ${\tt ptr}$
stores a pointer to the top of the stack and counter field ${\tt ctr}$
stores the number of times ${\tt Top}$ has been modified. Stack nodes
have a field ${\tt key}$ for the value of the node and a next field
${\tt nxt}$, which is a pointer to the next node of the stack.

\begin{figure}[t]
  \centering

  $\begin{array}[t]{@{}l@{\qquad\qquad\qquad}l}
    \begin{array}[t]{l}
      {\tt \mathbf{data}\ node\ \{ key: Val; nxt: *node \}} \\
      {\tt \mathbf{struct}\ ptr\_ctr\ \{ ptr: *node; ctr: nat \}} \\
      {\tt \mathbf{var}\ Top : *ptr\_ctr} \\
      {\tt \mathbf{initially}\ Top = (null, 0)}
    \end{array} 
    \\
    \\
    \begin{array}[t]{l@{~}l}
      \multicolumn{2}{l}{{\tt push(x)} \sdef } 
      % \\
      % \multicolumn{2}{l}{{\bf var}\ {\tt n, t} }
      \\  
      h_1: & {\tt n \asgn {\bf new}(Node)\ ;}% \Frame{n_p}{\Always(n_p
      % \in Fr) \ch \Always(n_p \notin Fr)}\ ;
      \\
      h_2: & {\tt n . key \asgn x}\ ;
      \\
      & {\bf repeat}
      \\
      h_3: & \quad {\tt t \asgn *Top\ ;}
      \\
      h_4: & \quad {\tt n. nxt \asgn t.ptr}  
      \\
      h_5: & {\bf until}\ {\tt CAS(Top, t, (n, t.ctr+1))} 
    \end{array}
    & 
    \begin{array}[t]{l@{~}l}
      \multicolumn{2}{l}{{\tt pop} \sdef } 
      \\
      % \multicolumn{2}{l}{{\bf var}\ {\tt t, tn, rv} }
      % \\
      & {\bf repeat} \\
      l_1: & \quad {\tt t \asgn *Top\ ;}
      \\
      l_2: & \quad \If {\tt t.ptr  = null}\ {\bf then}\\ 
      l_3: & \qquad {\tt rv \asgn Empty}\  ; % {\bf exit}
      \\
      l_4: & \qquad \Return % {\tt Empty}\ % ; {\bf exit}
      \\
      & \quad \Else  
      \\
      l_5:   & \qquad {\tt tn \asgn t.ptr \rightarrow nxt\ ; } \\
      l_6: & \qquad {\tt rv \asgn\ t.ptr \rightarrow key}   \\
      & \quad \Fi\ ;
      \\
      l_7: & {\bf until}\ {\tt CAS(Top, t, (tn, t.ctr +1))} \ch {} \\
      l_8: & \Return {\tt rv}
      % \\
      % 7. & {\bf free}\ {\tt t}\ ;
      % \\
      % l_7: & \Return {\tt
      % rv} % \Frame{t_p}{true \wchl \Always(t_p \in
      %                        % \Free)}
    \end{array}
  \end{array}$

  \figrule
  
  \caption{Push and pop operations of Treiber's lock-free stack}
  \label{fig:con-stack}
\end{figure}

Like many lock-free algorithms, the stack is implemented using an
atomic non-blocking compare-and-swap $CAS(ae, \alpha, \beta)$
primitive, which takes an address-valued expression $ae$ and variables
$\alpha$ and $\beta$ as input. If the value at $ae$ is equal to
$\alpha$, the CAS updates the value at $ae$ to $\beta$ and returns
$true$, otherwise (the value at $ae$ is not equal to $\alpha$) the CAS
does not modify anything and returns $false$.

The stack processes may perform either execute a push ${\tt push(x)}$
or a pop operation ${\tt pop}$. Within operation ${\tt push(x)}$, the
executing process uses local variables ${\tt n}$ and ${\tt t}$. The
process executing ${\tt push(x)}$ sets up a new node with value $x$
(lines $h_1$-$h_2$), then executes a try-CAS loop (lines $h_3$-$h_5$),
where the value of global variable ${\tt Top}$ is read and stored in
local variable ${\tt t}$ (line $h_3$), ${\tt n.nxt}$ is set to be the
local ${\tt t}$ (line $h_4$) and a ${\tt CAS(Top, t, (n, t.ctr + 1))}$
is executed (line $h_5$). The loop terminates if the CAS is
successful, otherwise lines $h_3$-$h_5$ are re-executed.  Operation
${\tt pop}$ stores ${\tt Top}$ in local variable ${\tt t}$ (line
$l_1$). If ${\tt t}$ is null (line $l_2$) it returns empty (line
$l_3$), otherwise it stores the value of ${\tt t.ptr \rightarrow nxt}$
in local variable ${\tt tn}$ (line $l_4$) and the value of ${\tt t.ptr
  \rightarrow key}$ in local variable ${\tt rv}$ (line $l_5$). The
loop terminates if the CAS is successful (line $l_6$), otherwise it
re-executes the loop body (i.e., from line $l_1$).

We verify this algorithm for an arbitrarily chosen set of processes
$P$ as follows, assuming that $AS(P)$, $LS(P)$, and $TS(P)$ denote,
respectively, the abstract stack (which is implemented as a sequence),
the coarse-grained abstraction (which is implemented as a linked list)
and the fine-grained concrete implementation (i.e., the Treiber
Stack). The approach we propose is to show that $TS(P)$ implements
(i.e., is a refinement of) $LS(P)$, which itself implements $AS(P)$.
Furthermore, we show that $LS(P)$ is \emph{linearisable} with respect
to $AS(P)$ \cite{Herlihy90,DSW11TOPLAS,VHHS06,Vaf07}.

The idea is that the coarse-grained abstraction $LS(P)$ allows
concurrency, but large parts of the code are sequential. This
simplifies the proof of linearisability because a data refinement is
performed from an abstract data structure representation in $AS(P)$ to
$LS(P)$, as opposed to a data refinement from $AS(P)$ directly to the
fine-grained program $TS(P)$. The data refinement between $AS(P)$ and
$LS(P)$ is performed on an extended specification that includes
histories of invocations and responses \cite{DSW11TOPLAS}, which
allows one to relate the refinement to Herlihy and Wing's original
definition. Verification that $TS(P)$ is a refinement of $LS(P)$ is
performed using a number of decomposition and transformation theorems.

\section{Interval-based framework}
\label{sec:an-interval-based}

We start by presenting our programming syntax. We then present our
interval-based framework, which has similarities to Interval Temporal
Logic \cite{Mos00}. However, the underlying semantics consists of
\emph{complete streams} that map each time to a state, and adjoining
intervals are assumed to be disjoint (adjoining intervals share a
boundary in Interval Temporal Logic) \cite{DDH12}.  To formalise the
behaviour over an interval, we restrict one's view of a stream to the
interval under consideration. However, because streams encode the
complete behaviour over all time, our framework allows properties
outside the given interval to be considered in a straightforward
manner \cite{DH11AVOCS,DH12,DDH12}.

\subsection{Syntax}

There are a number of well-established approaches to modelling program
behaviour, e.g., $Z$, $B$, I/O automata. However, the use of these
formalisms involve a non-trivial translation from the program to the
model. We use a framework in which commands closely resemble program
code, which simplifies the translation. The programs we consider often
have a pointer-based structure and hence, as in separation logic, we
distinguish between variables and addresses in the domains of the
program states \cite{Rey02}.

We assume variable names are taken from the set $Var$, values have
type $Val$, addresses have type $Addr \sdef \nat$, $Addr \cap Var =
\emptyset$ and $Addr \subseteq Val$. A \emph{state} over a set of
locations $VA \in Var \cup Addr$ is a member of $State_{VA} \sdef VA
\fun Val$ (i.e., a total function from $VA$ to $Val$). A \emph{state
  predicate} over $VA$ is a function of type $State_{VA} \fun \bool$.

A location may correspond to a data type with a field identifiers of
type $Field$.  We assume that every data type with $m$ number of
fields is assigned $m$ contiguous blocks of memory \cite{Vaf07} and
use $offset.f \in \nat$ to return the offset of the field $f$.  For
example, for any node of the Treiber Stack, we have $offset.key = 0$
and $offset.nxt = 1$. We assume that expressions have the following
syntax, where $k \in Val$ is a constant, $v \in Var$, $ae$ is an
address-valued expression, $f$ is a field, $e$, $e_1$ and $e_2$ are
expressions, $\uop$ is a unary operators and $\bop$ is a binary
operator, respectively. Both $\uop$ and $\bop$ are abstractions of the
possible unary and binary operators on expressions.
\begin{eqnarray*}
  e & \mathbin{:\!:\!=}  & k  ~\mid ~ v ~\mid ~ %  e  ~\mid ~ 
  \derefe ae
  ~\mid~ ae
  \cdota  f ~\mid ~   
  ae \mapsto f  ~\mid~ \uop e
  ~\mid~ e_1 \bop e_2 
\end{eqnarray*}
The semantics of expressions is given by function $eval$ which is
defined below for a state $\sigma$, where `.' denotes function
application. 
\begin{eqnarray*}
  \begin{array}[t]{rcl}
    eval.k.\sigma & \sdef & k 
    \\
    eval.v.\sigma & \sdef & \sigma.v
    \\
    eval.(\derefe ae).\sigma & \sdef & \sigma.(eval.ae.\sigma) 
    \label{eq:53}
  \end{array}
  & \qquad \qquad \qquad &   
  \begin{array}[t]{rcl}
    eval.(ae \cdota  f).\sigma & \sdef & eval.ae.\sigma +
    offset.f  \\
    \label{eq:57}
    eval.(\uop e).\sigma & \sdef & \uop eval.(e.\sigma )
    \\
    \label{eq:58}
    eval.(e_1 \bop e_2).\sigma & \sdef & eval.(e_1.\sigma )
    \bop eval.(e_2.\sigma)
  \end{array}
\end{eqnarray*}
Hence, the evaluation of a constant (including an address) in any
state is the constant itself and an evaluation of a variable is the
value of the variable in the given state. Evaluation of $\derefe ae$
returns the value of the state at the address that $ae$ evaluates to,
and $ae \cdota f$ returns the address that $ae$ evaluates to plus the
offset of the field $f$. The interpretations over unary and binary
operators are lifted in the normal manner.  We assume that expressions
are well-formed so that their evaluation in every state is always
possible. We define a shorthand
\begin{eqnarray*}
ae \mapsto f & \sdef & 
\deref (ae \cdota f)
\end{eqnarray*}
which returns the value of the state at address $ae \cdota f$.

\begin{definition}[Command syntax]
  \label{def:synt}
  For a state predicate $c$, variable $v$, expression $e$,
  address-valued expression $ae$, set of processes $P \subseteq Proc$,
  set of variables $Z$, and label $l$, the abstract syntax of a
  command is given by $Cmd$ below, where $C, C_1, C_2, C_p \in Cmd$.
  \begin{eqnarray*}
    Cmd & \abssynt & 
    \begin{array}[t]{@{}l@{}}
      \Chaos~~\mid~~ \Idle ~~\mid ~~  [c] ~~\mid~~ v
      \asgn e ~~ \mid~~ 
      ae \hasgn e ~~ \mid~~
      C_1 \ch C_2 ~~\mid~~ 
      C_1 \sqcap C_2
      ~~\mid~~  \\
      C^\omega 
      ~~\mid~~ \Par_{p:P}\ C_p 
      ~~\mid~~ \Context{Z}{C} ~~\mid~~ l : C ~~\mid~~ \Init c \st C 
    \end{array}
    % \\
    % Stmt & \abssynt &     
    % ~~\mid~~ \Grd{b}% ~~\mid ~~ x \asgn e
  \end{eqnarray*}
\end{definition}
Thus, a command may be $\Chaos$ (which is chaotic and allows any
behaviour), an $\Idle$ command, a guard $[c]$, assignments $v \asgn e$
and $ae \hasgn e$, a sequential composition $(C_1 \ch C_2)$, a
non-deterministic choice $C_1 \sqcap C_2$, an iteration $C^\omega$, a
parallel composition $\Par_{p:P}\ C_p$, a command $C$ executing in a
context $Z$, a command with a label $l$, or a command that executes
from an initial state that satisfies $c$.  

As an example, consider the simple sequential program in
\reffig{fig:abs-stack-1}, which demonstrates use of the syntax in
\refdef{def:synt}. For the program in \reffig{fig:abs-stack-1}, we
assume that $S$ is a sequence of values. 
\begin{example}
  \label{ex:abs-stack}
  Operation $SPush(x)$ updates the sequence $S$ by appending $x$ to
  the start of $S$, $SEmpty(arv)$ models a pop operation that returns
  $Empty$, and $SDoPop(arv)$ models a pop operation on a non-empty
  stack, which checks to see if $S$ is non-empty, sets the return
  value $arv$ to be the value of $S.0$, and removes the first element
  by setting the new value of $S$ to be the tail of $S$. An single
  stack operation is modelled by $SPP$, which consists of a
  non-deterministic choice between push operations that
  non-deterministically insert one of the possible key values onto the
  stack, or a pop operation. A possibly infinite number of stack
  operations starting from an empty stack is modelled by $SS$, which
  ensures that variable $S$ is in the context of the program, that the
  initial value of $S$ is $\aang{}$, and iterates over $SPP$. \hfill
  ${}_\clubsuit$
\end{example}

\begin{figure}[t]
  \centering
  $\begin{array}[t]{@{}r@{~~}c@{~~}l@{}}
    
    SPush(x)% (p)
    & \sdef &
    S \asgn
    \aang{x} \cat S% \Enf{OnlyAccessedBy.S.p}{} 
    \\
    \\
    SEmpty(arv)% (p)
    & \sdef &
    \begin{array}[c]{@{}l}
      % \enf OnlyAccessedBy.S.p \st 
      \begin{array}[c]{@{}l@{~}l@{}}
        & \Guard{S = \aang{}} \ch {}
        arv%_p% , hs
        \asgn Empty% , hs \cat \aang{pop_p^I, pop^R(p, Empty)}
      \end{array}
    \end{array}
    \smallskip \\
    SDoPop(arv)% (p)
    & \sdef &
    \begin{array}[c]{@{}l@{}}
      \begin{array}[c]{@{}l@{~}l@{}}
        & \Guard{S \neq \aang{}} \ch 
        arv% _p
        \asgn S.0 \ch 
        S% , hs
        \asgn tail.S%, hs \cat \aang{pop_p^I, pop^R(p, arv_p)}
      \end{array}
    \end{array}
    \smallskip \\
    SPop(arv)% (p)
    & \sdef & SEmpty(arv)% (p)
    \sqcap SDoPop(arv)% (p)
    \\
    \\
    SPP% (p)
    & \sdef & % \Idle \ch 
    \left(\textstyle\bigsqcap_{x : Val}\ SPush(x)\right)% (p)
    \sqcap
    SPop(arv)% (p)
    \\
    SS% (P)
    & \sdef &
    \Context{S}{\begin{array}[c]{@{}l@{}}
        \Init\ S = \aang{} \st  
        % \Par_{p: P}\ 
        SPP% (p)
        ^\omega
      \end{array}}
  \end{array}
  $
  \figrule

  \caption{An abstract (sequential) stack specification}
\label{fig:abs-stack-1}
\end{figure}

As a more complicated example, we now consider the program in
\reffig{fig:treibermodel}, which is a formal model for the Treiber
Stack from \reffig{fig:con-stack}.  We model a ${\tt ptr\_ctr}$ a
structure by type $Ptr\_Ctr \sdef Addr \times \nat$. For $(pp, cc) \in
Ptr\_Ctr$, we define functions $ptr.(pp, cc) \sdef pp$ and $ctr.(pp,
cc) \sdef cc$. The modification counter $ctr.(\derefe Top)$ is used to avoid the
ABA problem (see \refex{ex:ABA}). Unlike a ${\tt ptr\_ctr}$ structure,
which is assumed to be accessed atomically, list nodes are objects
whose fields may be accessed independently, and hence the $key$ and
$nxt$ fields are modelled as having different addresses.
\begin{example}
  The (interval-based) semantics of our language is given in
  \refsec{sec:behaviour}. The behaviours of both $CASOK_p$ and
  $CASFail_p$ are given in \refex{ex:CAS-beh} and the behaviour of
  $newNode(p, n_p)$ is formalised in \refex{ex:newnode}. 
  Both commands require the use of permissions to control the
  atomicity (see \refsec{sec:readwr-perm-interf}). 
  
  The Treiber Stack consists global address $Top \in Addr$ and a set
  of free addresses $FAddr \subseteq Addr$. The initial value of $Top$
  is $(null, 0)$ and all addresses different from $Top$ are free. We
  assume the existence of a garbage collector that gathers free
  pointers and returns them to $FAddr$. % , which is modelled by the rely
  % condition \cite{Jon83} $FreeAddr$. 

  Execution of ${\tt push(x)}$ by process $p$ is modelled by $Push(p,
  x)$. Within the push operation, we split the label $h_5$ into $hf_5$
  and $ht_5$ to distinguish between execution of the failed and
  successful branches of the CAS, respectively.  Commands $TryPush(p)$
  and $DoPush(p)$ model executions of the loop body that fail and
  succeed in performing the CAS at $h_5$, respectively.  Within
  $Push(p, x)$, because we use an ${}^\omega$ iteration, command
  $TryPush(p)$ may be executed a finite (including zero) number of
  times, after which $DoPush(p)$ is executed. However, it is also
  possible for $TryPush(p)$ to be executed an infinite number of times
  in which case $DoPush(p)$ never executes. Such behaviour is
  allowable for the Treiber Stack, which only guarantees lock-freedom
  of its concurrent processes \cite{CD09,CD07,Don06ICFEM}.

  For the pop operation, we use $ToCAS(p)$ to model the statements
  executed by process $p$ from the beginning of the loop up to the CAS
  at $l_7$ (via a failed test at $l_2$).  As in $Push(p,x)$, label
  $l_2$ is split into $lf_2$ and $lt_2$, and $l_7$ is split into
  $lf_2$ and $lt_5$ in the pop operation. The $TryPop(p)$ and
  $DoPop(p)$ commands model executions of the loop body that fail and
  succeed in executing the compare and swap at $l_7$,
  respectively. Command $Empty(p)$ models an execution of the pop
  operation that returns empty. The $Pop(p)$ operation consists of a
  finite or infinite iteration of $TryPop(p)$ followed by an execution
  of either $Empty(p)$ or $DoPop(p)$.

  The program is modelled as a
  parallel composition of processes, where each process repeatedly
  chooses either a pop or push operation non-deterministically, then
  executes the operation. \hfill ${}_\clubsuit$
\end{example}

  \begin{figure}[t]
    \centering

    $\begin{array}[t]{@{}r@{~~}c@{~~}l@{}} Setup(p,x) &
      \sdef &
      \begin{array}[c]{@{}l@{}}
        h_1 : newNode(p, n_p) % \asgn {\bf new}\ Node
        % \nasgn FAddr \ch FAddr \asgn FAddr \bs \{
        \ch h_2 : (n_p \cdota key) \hasgn x 
      \end{array}
      \smallskip \\ 
      TryPush(p) & \sdef &
      \begin{array}[c]{@{}l@{}}
        h_3: t_p \asgn \derefe Top \ch h_4: (n_p \cdota nxt) \hasgn ptr.t_p
        \ch 
        hf_5: CASFail_p(Top, t_p)
      \end{array}
      \smallskip \\ 
      DoPush(p) & \sdef &
      \begin{array}[c]{@{}l@{}}
        h_3: t_p \asgn \derefe Top \ch 
        h_4: (n_p \cdota nxt) \hasgn ptr.t_p \ch 
        ht_5: CASOK_p(Top, t_p, (n_p, ctr.t_p + 1)) 
      \end{array}
      \smallskip \\
      Push(p,x) & \sdef & Setup(p,x) \ch TryPush(p)^\omega \ch DoPush(p)
      \\
      \\
      ToCAS(p) & \sdef &
      \begin{array}[c]{@{}l@{}}
        l_1: t_p \asgn \derefe Top \ch lf_2: \Guard{ptr.t_p \neq null} \ch 
        l_5: tn_p \asgn ptr.t_p\mapsto nxt \ch 
        l_6: rv_p \asgn ptr.t_p\mapsto key
      \end{array}
      \smallskip \\
      TryPop(p) & \sdef &  ToCAS(p) \ch lf_7 : CASFail_p(Top, t_p)  
      \smallskip \\ 
      Empty(p, rv_p) & \sdef & l_1: t_p \asgn  Top \ch lt_2: \Guard{ptr.t_p = null} \ch
      l_5: rv_p \asgn Empty
      \smallskip \\
      DoPop(p, rv_p) & \sdef &  ToCAS(p) \ch lt_7 : CASOK_p(Top, t_p, (tn_p,
      ctr.t_p + 1)) %\ch l_8: freeNode(p,t_p)
      \smallskip \\
      Pop(p, rv_p)  & \sdef & TryPop(p)^\omega \ch (Empty(p, rv_p) \sqcap
      DoPop(p, rv_p))
      \\
      \\
      TPP(p) & \sdef & pidle: \Idle \ch \left( \bigsqcap_{x : Val}\
        Push(p,x)\right) \sqcap 
      Pop(p, rv_p) \smallskip\\
      TInit & \sdef & \derefe Top = (null, 0) \land FAddr \subseteq Addr  \bs
      \{Top\} \smallskip \\
      % FreeAddr & \sdef & \all a : FAddr, v \in Var, b \in Addr \st
      % \Always ((v 
      % \neq a) \land (\derefe b \neq a) \land (\derefe a = null))  \smallskip \\
      TS(P) & \sdef &  \Context{Top, FAddr}{
        \Init  TInit \st % \rely FreeAddr \st
        % \st
        \Par_{p : P}\ 
        \Context{t_p,n_p, tn_p, rv_p}{TPP(p)^\omega}}\label{eq:TS}
    \end{array}$

    % \begin{brijesh}
    %   Should initialisation be $FAddr \subseteq Addr$?
    % \end{brijesh}
    \figrule
    \caption{Formal model of the Treiber stack}
    \label{fig:treibermodel}
  \end{figure}

\subsection{Intervals}

A (discrete) {\em interval} (of type $Intv$) is a contiguous set of
integers (of type $Time \sdef \integer$), i.e., we define
\begin{eqnarray*}
Intv & \sdef & \{\Delta \subseteq Time \mid \all
t, t' : \Delta \st \all u : Time @ t \leq u \leq t' \imp u \in
\Delta\}
\end{eqnarray*}
We let $\lub.\Delta$ and $\glb.\Delta$ denote the \emph{least upper}
and \emph{greatest lower} bounds of an interval $\Delta$,
respectively, and define $\lub.\emptyset \sdef -\infty$ and
$\glb.\emptyset \sdef \infty$. If the size of $\Delta$ is infinite and
$\glb.\Delta \in \integer$, then $\lub.\Delta = \infty$ (i.e,. is not
a member of $\integer$) and if $\Delta$ is infinite and $\lub.\Delta
\in \integer$ then $\glb.\Delta = - \infty$. The \emph{length} of a
non-empty interval $\Delta$ is given by $\ell.\Delta \sdef \lub.\Delta
- \glb.\Delta$, and we define the length of an empty interval
$\emptyset$ to be $\ell.\emptyset \sdef 0$. We define the following
predicates on intervals. 
\begin{eqnarray*}
  \Inf.\Delta & \sdef & \mathsf{lub}.\Delta = \infty \\
  \Fin.\Delta & \sdef & \neg \Inf.\Delta \\
  \Empty.\Delta & \sdef & \Delta = \emptyset
\end{eqnarray*}
Hence, $\Inf.\Delta$, $\Fin.\Delta$ hold iff $\Delta$ has an infinite
and finite least upper bound, respectively, and $\Empty.\Delta$ holds
iff $\Delta$ is empty.

We must often reason about two \emph{adjoining} intervals, i.e.,
intervals that immediately precede or follow a given interval.  For
$\Delta, \Delta' \in Intv$, we define
\begin{eqnarray*}
  \Delta \adjoins \Delta' & \sdef &
  \begin{array}[t]{@{}l}
  \Delta \neq \emptyset \land \Delta' \neq \emptyset \imp
    (\lub.\Delta < \glb.\Delta') % \land  (\Delta \cap \Delta' =
    % \emptyset)
    \land (\Delta \cup \Delta' \in Intv)
  \end{array}
\end{eqnarray*}
Thus, $\Delta \adjoins \Delta'$ holds if and only if $\Delta'$
immediately follows $\Delta$% . By conjunct $\Delta \cap \Delta' =
and % adjoining intervals
$\Delta$ and $\Delta'$ are disjoint. Furthermore, by conjunct $\Delta
\cup \Delta' \in Intv$, the union of $\Delta$ and $\Delta'$ must be
contiguous. Note that both $\Delta \adjoins \emptyset$ and $\emptyset
\adjoins \Delta$ hold trivially. % , and furthermore,

\subsection{Interval predicates}

We aim to reason about the behaviours of a program over its interval
of execution and hence define an interval-based semantics for the
language in \refdef{def:synt}. In particular, we define interval
predicates \cite{DH12MPC,DH12,DH12iFM,DDH12}, which map an interval
and a stream of states to a boolean. The stream describes the
behaviour of a program over all time and an interval predicate
describes the behaviour over the given interval.

A \emph{stream} of behaviours over $VA \subseteq Var \cup Addr$ is
given by the total function $Stream_{VA} \sdef Time \fun State_{VA}$,
which maps each possible time to a state over $V$ and $A$. To reason
about specific portions of a stream, we use \emph{interval
  predicates}, which have type $IntvPred_{VA} \sdef Intv \fun \mcP
Stream_{VA}$. % We leave out subscripts
As with state expressions, we assume pointwise lifting of operators on
stream and interval predicates. % in the normal manner,
We assume pointwise lifting of operators on stream and interval
predicates in the normal manner, e.g., if $g_1$ and $g_2$ are interval
predicates, $\Delta$ is an interval and $s$ is a stream, we have $(g_1
\land g_2).\Delta.s = (g_1.\Delta.s \land g_2.\Delta.s)$.  When
reasoning about properties of programs, we would like to state that
whenever a property $g_1$ holds over any interval $\Delta$ and stream
$s$, a property $g_2$ also holds over $\Delta$ and $s$. Hence, we
define universal implication for $g_1, g_2 \in IntvPred$ as
$$
g_1 \entails g_2 \sdef \all \Delta : Intv, s : Stream \st g_1.\Delta.s
\imp g_2.\Delta.s
$$ 
We say $g_1\equiv g_2$ holds iff both $g_1 \entails g_2$ and $g_2
\entails g_1$ hold.

We define two trivial interval predicates 
\begin{eqnarray*}
\True.\Delta.s & \sdef &  true \\ 
\False.\Delta.s &  \sdef & false
\end{eqnarray*}
Like Interval Temporal Logic \cite{Mos00}, for an interval predicate
$g$, we say $(\Box g).\Delta.s$ holds iff $g$ holds in each
subinterval of $\Delta$ in stream $s$, say $(\Diamond g).\Delta.s$
holds iff $g$ holds in some subinterval of $\Delta$, and say
$\prev.g.\Delta.s$ holds iff there is some immediately preceding
interval of $\Delta$ within which $g$ holds in $s$. More formally, we
define:
\begin{eqnarray*}
  (\Box
  g).\Delta.s &\sdef& \all \Delta' : Intv \st \Delta' \subseteq \Delta \imp
  g.\Delta'.s \\
  (\Diamond g).\Delta.s & \sdef & \exists \Delta' : Intv \st \Delta'
  \subseteq \Delta \land g.\Delta'.s\\
  \prev.g.\Delta.s &\sdef& \exists \Delta' \st \Delta'
  \adjoins \Delta \land g.\Delta'.s 
\end{eqnarray*}

The \emph{chop} operator `;' is a basic operator, where $(g_1 \ch
g_2).\Delta$ holds iff either interval $\Delta$ may be split into two parts
so that $g_1$ holds in the first and $g_2$ holds in the second, or the
least upper bound of $\Delta$ is $\infty$ and $g_1$ holds in
$\Delta$. Thus, we define
\begin{eqnarray*}
  (g_1 \ch g_2).\Delta.s & \sdef &
  \begin{array}[t]{@{}l@{}}
    % (\lub.\Delta = \infty \land g_1.\Delta) \lor \\
    \left(\begin{array}[c]{@{}l@{}}
        \exists \Delta_1, \Delta_2 : Intv  \st
        (\Delta = \Delta_1
        \cup \Delta_2) \land (\Delta_1 \adjoins \Delta_2) 
        \land g_1.\Delta_1.s \land g_2.\Delta_2.s
      \end{array}\right) \lor \\
    (\lub.\Delta = \infty \land g_1.\Delta.s)
  \end{array}
\end{eqnarray*}
Note that $\Delta_1$ may be empty in which case $\Delta_2 = \Delta$,
and similarly if $\Delta_2$ is empty then $\Delta_1 = \Delta$. In the
definition of chop, we allow the second disjunct $\lub.\Delta = \infty
\land g_1.\Delta$ to allow for $g_1$ to model an infinite (divergent
or non-terminating) program. 

We define the possibly infinite iteration of an interval predicate $p$
as follows.  We assume that interval predicates are ordered using
universal implication `$\entails$' and that the stream and interval
are implicit on both sides of the definition.
\begin{eqnarray*}
  g^\omega & \sdef & \nu z \st (g \ch z) \lor
  \Empty 
\end{eqnarray*}
Thus, $g^\omega$ is a greatest fixed point that defines either finite
or infinite iteration of $g$ \cite{DHMS12}. % Note

Properties that hold over a larger interval may be decomposed into
properties of the subintervals if the interval predicate that
formalises the property splits and/or joins \cite{Hay08,DH12}. We also
find it useful to reason about properties that widen, where a property
holds over a larger interval if it holds over any subinterval.
\begin{definition}[Splits, Joins, Widens]
  Suppose $g$ is an interval predicate. We say
  \begin{itemize}
  \item $g$ \emph{splits} iff $g \entails \Box g$, i.e., if $g$ holds
    over an interval $\Delta$, then $g$ must hold over all
    subintervals of $\Delta$,
  \item $g$ \emph{joins} iff $(g \ch g^\omega) \entails g$, i.e., for any
    interval $\Delta$ if $g$ iterates over $\Delta$, then $g$ must
    hold in $\Delta$, and
  \item $g$ \emph{widens} iff $\Diamond g \entails g$, i.e., if $g$
    holds over some subinterval of $\Delta$, then $g$ holds over the
    interval $\Delta$.
  \end{itemize}
\end{definition}

\subsection{Evaluating state predicates over intervals}
\label{sec:expr-eval-over}

The values of an expression $e$ at the left and right ends of an
interval $\Delta$ with respect to a stream $s$ are given by
$\ola{e}.\Delta.s$ and $\ora{e}.\Delta.s$, respectively, which are
defined as
\begin{eqnarray*}
  \ola{e}.\Delta.s & \sdef & eval.e.(s.(\glb.\Delta))\\
  \ora{e}.\Delta.s & \sdef & eval.e.(s.(\lub.\Delta))
\end{eqnarray*}
Note that $\ola{e}.\Delta.s$ is undefined if $\glb.\Delta = - \infty$
and $\ora{e}.\Delta.s$ is undefined if $\inf.\Delta$. Similarly, we
use the following notation to denote that $c$ holds at the beginning
and end of the given interval $\Delta$ with respect to a stream $s$,
respectively.
\begin{eqnarray*}
  \ola{c}.\Delta.s & \sdef & \glb.\Delta \notin \{-\infty, \infty\} \land
  c.(s.(\glb.\Delta))\\
  \ora{c}.\Delta.s & \sdef & \lub.\Delta \notin \{-\infty, \infty\} \land
  c.(s.(\lub.\Delta))
\end{eqnarray*}
It is often useful to specify that a state predicate holds on a point
interval. Hence we define
\begin{eqnarray*}
  \ceil{c}.\Delta.s & \sdef & \exists t : Time \st \Delta =
  \{t\} \land c.(s.t)
\end{eqnarray*}

Two useful operators for evaluating state predicates over an interval
are $\Always c$ and $\Eventually c$, which ensure that $c$ holds in
\emph{all} and \emph{some} state of the given stream within the given
interval, respectively. Thus, for an interval $\Delta$ and stream $s$,
we define:
$$
\begin{array}{rcl}
  (\Always c).\Delta.s & \sdef &  \all t : \Delta \st c.(s.t)
  \\
  (\Eventually c).\Delta.s & \sdef & \exists t : \Delta 
  \st c.(s.t)
\end{array}
$$

As demonstrated using \refex{ex:ABA} below, operator $\Eventually$ may
be used to give a straightforward formalisation of the ABA problem.
\begin{example}[The ABA problem]
  \label{ex:ABA}
  To effectively implement CAS-based operations, where ${\tt CAS(ae,
    \alpha, \beta)}$ is executed by a process $p$, operations are
  often structured so that the value at address $ae$ is stored as
  value of variable $\alpha$, some processing is performed on $\alpha$
  and the result stored in a variable $\beta$. Then a ${\tt CAS(ae,
    \alpha, \beta)}$ is executed to attempt updating the value at $ae$
  to $\beta$. If the CAS fails, then the environment must have
  modified the value at $ae$ since it was last read as $\alpha$.
  However, it is possible for a CAS to succeed (when it should have
  failed) if the environment exhibits so-called ABA-like behaviour
  \cite{CDG05}, where the value at $ae$ changes from say $\alpha$ to
  $\beta$ and then back to $\alpha$. ABA-like behaviour of an
  expression $e$ is formalised using the following interval predicate:
  \begin{eqnarray*}
    ABA.e &  \sdef  & 
    \exists k_1, k_2 \in Val \st 
    \begin{array}[c]{@{}l}
       k_1 \neq k_2 \land 
       (\Eventually( e = k_1) \ch
      \Eventually ( e = k_2) \ch \Eventually
      ( e = k_1))
    \end{array}
  \end{eqnarray*}
  Hence $(ABA.e).\Delta$ holds iff $\Delta$ can be partitioned into
  three adjoining intervals $\Delta_1$, $\Delta_2$ and $\Delta_3$ such
  that the value at $e$ is $k_1$ sometime within $\Delta_1$, then
  changes to $k_2$ sometime within $\Delta_2$ and back to $k_1$
  sometime within $\Delta_3$. Note that $ABA.e$ allows the value of
  $e$ to change several times when changing from $k_1$ to $k_2$ then
  to $k_1$. \hfill ${}_\clubsuit$
\end{example}

We say a location $va$ is \emph{stable} at time $t$ in stream $s$
(denoted $\mathsf{stable\_at}.va.t.s$) iff the value of $va$ in $s$ at
time $t$ does not change from its value at time $t-1$,
i.e., 
\begin{eqnarray*}
  \mathsf{stable\_at}.va.t.s & \sdef & s.t.va = s.(t-1).va
\end{eqnarray*}
Variable $va$ is stable over an interval $\Delta$ in a stream $s$
(denoted $\mathsf{stable}.va.\Delta.s$) iff the value of $va$ is stable
at each time within $\Delta$. A set of locations $VA$ is \emph{stable}
in $\Delta$ (denoted $\mathsf{stable}.VA.\Delta$) iff each variable in
$VA$ is stable in $\Delta$. Thus, we define:
\begin{eqnarray*}
  \mathsf{stable}.va.\Delta.s & \sdef & 
  \all t : \Delta \st \mathsf{stable\_at}.va.t.s
  \\
  \mathsf{stable}.VA.\Delta & \sdef &  \all va : VA \st \mathsf{stable}.va.\Delta
\end{eqnarray*}
Note that every location is stable in an empty interval and the empty
set of locations is stable in any interval, i.e., both
$(\stable.VA).\emptyset$ and $\stable.\emptyset.\Delta$ hold trivially.

\subsection{Read/write permissions and interference}
\label{sec:readwr-perm-interf}
As we shall see in \refsec{sec:behaviour}, the behaviour a process
executing a command is formalised by an interval predicate, and the
behaviour of a parallel execution over an interval is given by the
conjunction of these behaviours over the same interval. Because the
state-spaces of the two processes are potentially overlapping, there
is a possibility that a process writing to a variable conflicts with a
read or write to the same variable by another process. To ensure that
such conflicts do not take place, we follow Boyland's idea of mapping
variables to a {\em fractional permission} \cite{Boy03}, which is {\em
  rational} number between $0$ and $1$.  A process has write-only
access to a variable $v$ if its permission to access $v$ is $1$, has
read-only access to $v$ if its permission to access $v$ is above $0$
but below $1$, and has no access to $v$ if its permission to access
$v$ is $0$. Note that we restrict access so that a process may not
have both read and write permission to a variable. Because a
permission is a rational number, read access to a variable may be
split arbitrarily (including infinitely) among the processes of the
system. However, at most one process may have write permission to a
variable in any given state. Note that the precise value of the read
permission is not important, i.e., there is no notion of priority
among processes based on the values of their read
permissions. % We ensure non-conflicting access

We assume that every state contains a \emph{permission} variable $\Pi$
whose value in state $\sigma \in State_V$ is a function of type
\[ 
V \fun Proc \fun \{n : \rat | 0 \leq n \leq 1\}
\]
Note that it is possible for permissions to be distributed differently
within states $\sigma$, $\sigma'$ even if the values of the standard
variables in $\sigma$ and $\sigma'$ are identical, i.e., it is
possible to get $\sigma.\Pi \neq \sigma.\Pi'$ even if $(\{\Pi\} \ndres
\sigma) = (\{\Pi\} \ndres \sigma')$ holds, where `$\ndres$' denotes
the domain
anti-restriction.  % We formalise write, read and no permissions as
\begin{definition}[Permission]
  \label{def:read-write-perm}
  A process $p \in Proc$ has {\em write-permission} to variable $va$ in
  state $\sigma$ iff $\sigma.\Pi.va.p = 1$, has {\em read-permission}
  to $va$ in $\sigma$ iff $0 < \sigma.\Pi.va.p < 1$, and has {\em
    no-permission} to access $va$ in $\sigma$ iff $\sigma.\Pi.va.p = 0$
  holds.
\end{definition}

We introduce the following shorthands, which define state predicates
for a process $p$ to have read-only and write-only permissions to a
variable $va$, and to be denied permission to access $va$.
$$
\begin{array}[t]{rclrclrcl}
  \mcR.va.p &\sdef& 0 < \Pi.va.p < 1   \qquad\quad 
  \mcW.va.p &\sdef& \Pi.va.p = 1  \qquad\quad \mcD.va.p &\sdef& \Pi.va.p = 0
\end{array}
$$
In the context of a stream $s$, for any time $t \in \integer$, process
$p$ may only write to and read from $va$ in the transition step from
$s.(t-1)$ to $s.t$ if $\mcW.va.p$ and $\mcR.va.p$,
respectively. % Otherwise, i.e., if $\mcD_p.v.(s.t)$ holds, process $p$
Thus, $\mcW.va.p$ does not grant process $p$ permission to write to
$va$ in the transition from $s.t$ to $s.(t+1)$ (and similarly
$\mcR.va.p$).

It is straightforward to use fractional permissions to characterise
interference within a set of processes. For $P \subseteq Proc$ and $VA
\subseteq Var \cup Addr$, we define
\begin{eqnarray*}
  \mcI.VA.P& \sdef &  \exists va: VA, p : Proc \bs P \st
  \mcW.va.p
\end{eqnarray*}
which states that there may be interference on locations in $VA$ from
a process different from those in $P$. We use $\mcI.VA.p$ to denote
$\mcI.VA.\{p\}$ for a singleton set $\{p\}$ and $\mcI.va.P$ to denote
$\mcI.\{va\}.P$ for a singleton set $\{va\}$. This characterisation of
interference is particularly useful in this paper where we use rely
conditions (see \refsec{sec:rely-conditions}) to formalise the
behaviour of the environment. For example, to state that environment
of $p$ does not modify a variable $v$, we simply ensure that the rely
condition implies $\Always \neg \mcI.v.p$.

Such notions are particularly useful because we aim to develop
rely/guarantee-style reasoning, where we use rely conditions to
characterise the behaviour of the environment. To state an assumption
that there is no interference on $va$ during the execution of a
command that uses $va$, one may introduce $\Always \neg \mcI.va.p$ as
a rely condition to the command (see \refex{ex:using-nonint}).

We define some conditions on streams using fractional permissions that
formalise our underlying assumptions on access permissions. 
\begin{description}
\item[\textbf{HC1}] If no process has write access to $va \in Var \cup
  Addr$ within an interval, then the value of $va$ does not change
  within the interval, i.e.,
  \begin{eqnarray*}
    \Always (\all p : Proc \st \neg \mcW.va.p) & \imp & \mathsf{stable}.va
    % \all va: VA, s: Stream_VA, t: \integer \st 
    % (\all p: Proc \st \neg \mcW_p.va.t)
    % & \imp &  (s.(t-1)).va = (s.t).va 
  \end{eqnarray*}
\item[\textbf{HC2}] The sum of the permissions of the processes on any
  location $va$ is at most $1$, i.e.,
  \begin{eqnarray*}
    \Always ((\displaystyle\Sigma_{p
      \in Proc} \Pi.va.p) & \leq &  1)
  \end{eqnarray*}

\end{description}
For the rest of this paper, we implicitly assume that these conditions
hold. Alternatively, one could make the conditions explicit by adding
them to the rely conditions of the programs under consideration.

Using these healthiness conditions, we obtain a number of
relationships between the values of a variable and the permissions
that a process has to access the variable. For example, we may prove
that if a process has read permission to a variable, then no process
has write permission to the variable. Furthermore, if over an interval
no process has write permission to a variable, then the variable must
be stable over the interval.
\begin{lemma} Both of the following hold for any location $va \in Var
  \cup Addr$
  \begin{eqnarray}
    \label{eq:67}
    \Always (
    (\exists p : Proc \st \mcR.va.p)
    & 
    \imp 
    &  (\all p : Proc \st \neg \mcW.va.p)) % (s.(t-1)).va = (s.t).va 
    \\
    \label{eq:60}
    \Always (\all x: Proc \st \neg \mcW.va.p)  &
    \entails & \mathsf{stable}.va
  \end{eqnarray}
\end{lemma}
One may also define additional properties. For instance, a set of
locations $VA$ may be declared to be local within a set of set of
processes $P$ using the following interval predicate.
\begin{eqnarray*}
  Local.VA.P & \sdef & 
  \all va : VA \st \Always \left(\begin{array}[c]{@{}l@{}}
    (\all q : Proc \bs P \st \mcD.va.q) \land \\
    ((\all p: P \st \neg \mcW.v.p)  \imp (\all p: P \st \mcR.v.p))
  \end{array}\right)
\end{eqnarray*}
Hence, if $Local.VA.P$, then no process outside of $P$ has permission
to accesss the locations in $VA$ and furthermore, if no process in $P$
has write permission to a location $va$ in $VA$, then all processes in $P$
automatically obtain read permission to $va$.

It is important to be able to determine the set of all variables and
addresses that must be accessed in order to evaluate an expression in
a given state. The set of variables that are accessed is state
dependent because an expression may involve pointers and the address
that the pointer points to may vary with the state. Hence, we define a
function $accessed$, which returns the set of locations (i.e.,
variables and addresses) accessed. Below, we assume that $k \in Val$,
$v \in Var$, $a \in Addr$, $f \in Field$, $t \in Time$, $e_1$, and
$e_2$ are expressions and $s$ is a stream.
$$
\begin{array}[t]{l@{\qquad\quad}l}
  \begin{array}[t]{rcl}
  accessed.k.\sigma & \sdef &  \{\} 
  \\ 
  accessed.v.\sigma & \sdef &  \{v\} 
  \\
  accessed.(\derefe ae).\sigma & \sdef &  \{eval.ae.\sigma\} \cup accessed.ae.\sigma
  % \\
  % accessed.a.\sigma & \sdef &  \{a\} 
  % \\
  % accessed.(\& a).\sigma & \sdef &  \{\} 
\end{array}
& 
\begin{array}[t]{rcl}
  accessed.(ae \cdota f).\sigma & \sdef & accessed.ae.\sigma
  % \\
  % accessed.(e \mapsto f).\sigma & \sdef & \{(e \cdota f).\sigma\}  \cup
  % accessed.e.\sigma
  \\
  accessed.(\bop e).\sigma & \sdef & accessed.e.\sigma
  \\
  accessed.(e_1 \bop e_2).\sigma & \sdef &
  \begin{array}[t]{@{}l@{}}
    accessed.e_1.\sigma
    \cup {}\\
    accessed.e_2.\sigma
  \end{array}
\end{array}
\end{array}
$$

We define the following predicate, which states that process $p$ has
permission to read each of the locations required to evaluate
expression $e$ in state $\sigma$.
\begin{eqnarray*}
  ReadAllLocs.e.p.\sigma & \sdef & \all va : accessed.e.\sigma \st \mcR.va.p
\end{eqnarray*}
The following interval predicates state that process $p$ writes to a
location of $e$ within interval $\Delta$, that there is no
interference on $e$ within $\Delta$, and that all other processes are
denied access to the locations in $e$.
\begin{eqnarray*}
  (WriteSomeLoc.e.p).\Delta.s & \sdef & \Eventually (\exists va \st va
  \in accessed.e \land \mcW.va.p)
  \\
  (IntFree.e.p).\Delta.s & \sdef & \Always (\all va  \st va \in
  accessed.e \imp \neg \mcI.va.p)
  \\
  (OnlyAccessedBy.e.p).\Delta.s & \sdef & \Always (\all va \st \all q : Proc \bs \{p\}\st 
  va \in accessed.e \imp \mcD.va.q)
\end{eqnarray*}

\section{Interval-based semantics of parallel programs}
\label{sec:interv-based-semant}

In \refsec{sec:behaviour}, we present our interval-based semantics of
the programming model, and in \refsec{sec:enforced-conditions} we
present the concept of enforced properties, which enable a behaviour
of a command to be constrained. We present a theory for refining
program behaviour in \refsec{sec:behaviour-refinement}.

\subsection{Semantics of commands}
\label{sec:behaviour}

We use the following interval predicates to formalise the semantics of
the commands in \refdef{def:synt}, where $p$ is a process, $va \in Var
\cup Addr$ is a location and $Z \subseteq Var \cup Addr$ is a set of
locations, $e$ is an expression, $k$ a constant and $c$ is a state
predicate.
\begin{eqnarray*}
  \idle_{p,Z} & \sdef & \all va : Z \st \Always \neg \mcW.va.p
  \\ 
  \Eval_{p,Z}(e,k) & \sdef & \Eventually (e = k \land
  ReadAllLocs.e.p) \land  \idle_{p,Z} \\
  % \Sometime_{p,Z}.c & \sdef & \Eval_{p,Z}(c, true) \\
  \Update_{p,Z}(va,k) & \sdef &
  \left\{\begin{array}[c]{@{}l@{\qquad}l@{}}
    \idle_{p, Z \bs\{va\}}
    \land \neg \Empty \land \Always(va = k \land \mcW_p.va) &
    \textrm{if $va \in Var$} \\
    \idle_{p, Z \bs\{va\}}
    \land \neg \Empty \land \Always((\deref va) = k \land \mcW_p.va) &
    \textrm{if $va \in Addr$} 
  \end{array} \right.
  % \\
  % \mathsf{record}_{p}(va, e) & \sdef & \exists k \st k = \ora{e} \ch
  % \Update_{p,\emptyset} (v, k)
\end{eqnarray*}
Hence, $\idle_{p,Z}$ states that process $p$ does not write to any of
the locations in $Z$ within the given interval. Interval predicate
$\Eval_{p,Z}(e,k)$ models the evaluation of expression $e$ to a value
$k$ by process $p$ in context $Z$, where $\Eval_{p,Z}(e,k).\Delta.s$
holds iff there is a state $s.t$ (for $t \in \Delta$) such that the
value of $e$ in state $s.t$ is $k$ and $p$ can read each of the
locations needed to evaluate $e$ in $s.t$. Furthermore, $p$ does not
write to any of the variables of the context $Z$ within $\Delta$.
Interval predicate $\Update_{p,Z}(va,k)$ models the modification of
location $va$ to value $k$ by process $p$ executing in a context $Z$,
where $\mathsf{update}_{p,Z}(va,k).\Delta.s$ holds iff within $s$,
throughout $\Delta$, the value of $va$ is $k$, $p$ has write
permission to $va$ and does not have write permission to any other
variable in $Z$. Additionally, $\neg \Empty$ holds to ensure that $va$
is actually updated --- this is necessary because $\Always c$
trivially holds for an empty
interval. % If $va$ is a variable the store is updated and

We obtain the following lemma, which relates write permissions
to expression evaluation and stability of a variable.
\begin{lemma}
  Suppose $V, Z \subseteq Var$, $p \in Proc$, $e$ is an expression and
  $k \in Val$ is a constant. Then each of the following hold:
  \begin{eqnarray}
    \label{eq:widen-1}
    \begin{array}[c]{@{}l@{}}
      \idle_{p,Z}% (e,k)
      \ch \Eval_{p,Z}(e,k) 
    \end{array}
    & \entails &  \Eval_{p,Z}(e,k)
    \\
    \label{eq:widen-2}
    \begin{array}[c]{@{}l@{}}
      \Eval_{p,Z}(e,k) \ch \idle_{p,Z}
      \end{array}
    & \entails &  \Eval_{p,Z}(e,k)
  \end{eqnarray}
\end{lemma}

\begin{figure}[t]
  \centering
      $
  \begin{array}{@{}rcl@{}}
    \begin{array}[t]{rcl}
      beh_{p,Z}.\Chaos & \sdef & \True
      \smallskip\\
      beh_{p,Z}.\Idle & \sdef & \idle_{p,Z}
      \smallskip\\
      beh_{p,Z}.{[}c{]} & \sdef & \Eval_{p,Z}.(c, true)
      \smallskip\\
      beh_{p,Z}.(\Init{c}\st {C}) & \sdef & \prev.\ora{c} \land beh_{p,Z}.C
    \end{array}
    &\qquad & 
  \begin{array}[t]{rcl}
    beh_{P,Z}.(C_1 \ch C_2) & \sdef &  beh_{P,Z}.C_1 \ch
    beh_{P,Z}.C_2
    \smallskip \\
    beh_{P,Z}.(C_1 \sqcap C_2) & \sdef &  beh_{P,Z}.C_1 \lor
    beh_{P,Z}.C_2 \smallskip \\
      beh_{P,Z}.C^\omega & \sdef & (beh_{P,Z}.C)^\omega 
      \smallskip \\
    % \smallskip \\
    % beh_{P,Z}.(\Rely{r}{C}) &  \sdef &  r \imp beh_{P,Z}.C 
    % \smallskip\\
    % beh_{P,Z}.(\Enf{g}{C}) & \sdef &  g \land beh_{P,Z}.C 
    % \smallskip\\
    beh_{p, Z}.(l : C) & \sdef & \Always(pc_p = l) \land beh_{p,Z}.C
  \end{array}
  \medskip\\
  \multicolumn{3}{l}{
    \begin{array}[t]{rcl}
      beh_{p,Z}.(v \asgn e) & \sdef &
      \exists k \st \Eval_{p,Z} (e,k) 
      \ch \Update_{p,Z}(v,k)
      \smallskip \\
      beh_{p,Z}.(ae \hasgn e) & \sdef &
      \exists k, a\st \Eval_{p,Z} (ae, a) \land \Eval_{p,Z} (e,k)
      \ch  \Update_{p,Z}(a,k)
      \smallskip \\
      beh_{P,Z}.(\Par_{p:P}\ C_p) & \sdef & 
      \left\{\begin{array}[c]{@{}l@{\quad }l}
          \True & \textrm{if $P = \emptyset$} \smallskip \\
          beh_{p,Z}.C_p & \textrm{if $P = \{p\}$} \smallskip  \\
          \exists P_1, 
          P_2 \st 
          \begin{array}[t]{@{}l@{}}
            (P_1 \cup P_2 = P) \land  
            (P_1 \cap P_2 = \emptyset) \land \\
            beh_{P_1,Z}.((\Par_{p:P_1} C_p) \ch \Idle)
            \land 
            beh_{P_2,Z}. ((\Par_{p:P_2} C_p) \ch \Idle) 
          \end{array}
          & \textrm{otherwise} 
        \end{array}\right. 
      \smallskip\\
      beh_{P, Z}.\Context{Y}{C} & \sdef &
      \begin{array}[t]{@{}l@{}}
        Local.P.Y
        \land 
        (Z \cap Y = \emptyset)   \land 
        beh_{P,Z \cup Y}.C
      \end{array}
    \end{array}
  }
  \end{array}$

    \figrule
  \caption{Formalisation of behaviour function}
  \label{fig:beh-def}
\end{figure}

\begin{definition}[Behaviour]
  \label{def:beh}  
  The \emph{behaviour} of a command $C$ given by the abstract syntax
  in \refdef{def:synt} executed by a non-empty set of processes $P$ in
  a context $Z \subseteq Var \cup Addr$ is given by interval predicate
  $beh_{P,Z}.C$, which is defined inductively in
  \reffig{fig:beh-def}. % The \emph{behaviour} of $C$ executed by
  % processes in $P$ is given by $beh_P.C \sdef beh_{P,\emptyset}.C$. We use
  % $beh_p.C$ to denote $beh_{\{p\}}.C$ for a process $p$.
\end{definition}
Within $beh_{P, Z}$ the context $Z$ defines the set of locations that
the processes in $P$ may or may not modify. For example, command
$\Idle$ executed by process $p$ in a context $Z$ should not write to
the locations in $Z$. Similarly, an assignment $v \asgn e$ should not
write to locations in $Z\bs \{v\}$.

In the description below, we assume that the executing process is $p$
and the command under consideration occurs within the scope of a
context $Z$. The behaviour of command $\Idle$ states that $p$ does not
write to any of the locations in $Z$. The behaviour of $[c]$ states
that $c$ evaluates to true in some state within the interval, that
process $p$ has permission to read the locations of $c$ in this state
and, and that $p$ does not write to any of the locations in $Z$. For
example, within $ToCAS(p)$ (\reffig{fig:treibermodel}), $[ptr.t_p \neq
null]$ holds iff there is a state in which $p$ can read $t_p$ and the
value of $ptr.t_p$ is non-null.

Execution of a variable assignment $v \asgn e$ consists of two parts,
where the value of $e$ is evaluated to $k$, and then the value of $v$
is updated to $k$. The executing process must have permission to read
the locations of $e$ within during the evaluation, and must have
permission to write to $v$ during the update. Note that because we
assume true concurrency, the evaluation is non-deterministic (in the
sense of \cite{DDH12,HBDJ11}). For example, consider the assignment
$t_p \asgn \derefe Top$ within $ToCAS(p)$
(\reffig{fig:treibermodel}). It is possible for $\derefe Top$ to
change multiple times within the interval in which $\derefe Top$ is
evaluated due to the execution of other processes. The value returned
by the evaluation of $\derefe Top$ depends on the time at which the
value at $Top$ is read. Because $t_p$ is local to $p$, the update to
$t_p$ may always be executed.  Assignment $ae \hasgn e$ is similar to
a variable assignment, but the command must additionally must evaluate
the address-valued expression $ae$ to determine the address to be
updated.

The behaviours of sequential composition, non-deterministic assignment
and iteration are modelled in a straightforward manner using chop,
disjunction and iteration of interval predicates. The parallel
composition is chaotic if the given set of processes is empty and
behaves as a single process if the set of processes is
singleton. Otherwise the given set processes is partitioned into two
disjoint subsets and the behaviours of both subsets occur in the given
interval. Note that the two or more parallel processes may access the
same shared variables --- the manner in which these variables are
accessed is controlled using fractional permissions. We allow $\Idle$
to be executed within the parallel composition to allow asynchronous
termination \cite{DDH12}. 

The behaviour of $\Context{Y}{C}$ is the behaviour of $C$ in an
extended context $Y$ no other process is given permission to access
locations in $Y$ during the execution of $C$.  The behaviour of a
labelled command assumes the existence of an auxiliary program counter
variable $pc_p$ local to each process $p$. Execution of $l: C$ by
process $p$ guarantees that $pc_p$ has value $l$ throughout the
interval of execution. Unlike \cite{deRoever01,DSW11,DGLM04,Don09}
where labels strictly correspond to the atomic portions of the
programs under consideration, we only use labels to determine
auxiliary information and the same label may correspond to a number of
atomic steps.

\subsection{Enforced conditions}
\label{sec:enforced-conditions}

We introduce a construct for defining an enforced condition, which
restricts the behaviour of a command so that the property being
enforced is guaranteed to hold \cite{Don09}.
\begin{definition}[Enforced condition]
  For interval predicate $d$ and command $C$, we let $\enf{d}\st{C}$
  denote a command with an \emph{enforced condition} $d$, where
  \begin{eqnarray*}
    beh_{P,Z}.(\enf{d}\st{C})  & \sdef &    d \land beh_{P,Z}.C
  \end{eqnarray*}
\end{definition}
Hence, $\enf{d}\st{C}$ executes as $C$ and in addition guarantees that
$d$ holds \cite{Don09,DH12}. Note that if $\neg d$ holds, then
$\enf{d}\st{C}$ has no behaviours.  Enforced conditions may be used to
state properties of an implementation that may not be easily
expressible as a command. We present three examples relevant to the
Treiber stack to illustrate the use of enforced properties. Although
each of the examples only uses enforced properties on fractional
permissions, the theory of enforced properties is more general as it
allows any interval predicate to be enforced
\cite{DH12,DH12MPC,DDH12}.
\begin{example}[Compare-and-swap]
  \label{ex:CAS-beh}
  A ${\tt CAS(ae, \alpha, \beta)}$ is atomic on ${\tt ae}$, i.e., the value at
  address ${\tt ae}$ is guaranteed not to be modified by the
  environment over the interval in which ${\tt CAS(ae, \alpha, \beta)}$
  is executed. Enforced properties, state predicate evaluation and
  fractional permissions are used together to formalise the behaviour
  of a ${\tt CAS(ae, \alpha, \beta)}$ executed by a process $p$, as
  follows:
\begin{eqnarray*}
  CASOK_p(ae, \alpha, \beta) &  \sdef & \enf{OnlyAccessedBy.(\derefe ae).p
  }\st{\Guard{\derefe ae = \alpha} 
    \ch ae \hasgn \beta}  
  \\
  CASFail_p(ae, \alpha) &  \sdef & 
  \Guard{\derefe ae \neq \alpha}
  \\
  CAS_p(ae, \alpha, \beta) & \sdef &  CASOK_p(ae, \alpha, \beta) \sqcap CASFail_p(ae,\alpha)
\end{eqnarray*}
Note that due to the enforced properties $OnlyAccessedBy.ae.p$ within
$CASOK_p$, the environment is denied access to the locations needed to
evaluate $ae$. Although the CAS may take a number of atomic steps to
execute, access to the test and set of $ae$ occurs without any process
accessing the locations in $ae$. A CAS that performs the test and set
in a single transition may be considered to be an implementation of
this specification.  

Because the CAS may take multiple steps, it is possible for the values
of $\alpha$ and $\beta$ to change. To achieve more deterministic
behaviour over the interval, (in the sense of \cite{HBDJ11}), $\alpha$
and $\beta$ are typically local variables of the executing process,
and hence, by {\bf HC4}, $\alpha$ and $\beta$ cannot be modified by
the environment of $p$. A CAS that performs a test and set in a single
transition typically places further restrictions on the structure of
$ae$ to ensure implementability.  \hfill ${}_\clubsuit$
\end{example}

\begin{example}[New nodes]
  \label{ex:newnode}
  We assume that the set of free addresses is given by $FAddr$. To
  ensure that two different processes are not assigned the same free
  address, we must ensure that there is no interference to $FAddr$
  while $p$ accesses $FAddr$. This may be achieved via an enforced
  condition on $FAddr$.
  \begin{eqnarray*} 
    newNode(p, v) & \sdef &
    \begin{array}[t]{@{}l@{}}
      \enf
      OnlyAccessedBy.FAddr.p \st \\
      \quad \bigsqcap_{fn \in FAddr} [fn+ offset.nxt \in
      FAddr] \ch v \asgn fn \ch FAddr \asgn
      FAddr \bs \{fn, fn+offset.nxt\} 
    \end{array}
  \end{eqnarray*}
  Hence, $newNode(p, v)$ ensures that all processes different from $q$
  are denied access to the current set of free nodes
  $FAddr$. Furthermore, it non-deterministically chooses a free node
  from $fn$ such that $fn+1$ is a free node, then assigns $fn$ to $v$,
  and removes both $fn$ and $fn+ offset.nxt$ (i.e., $fn$ and $fn+1$)
  from the set of available free locations.  \hfill ${}_\clubsuit$
\end{example}

\begin{example}
  To further illustrate the use of permissions and enforced
  properties, consider the program in \reffig{fig:abs}, which extends
  the abstract program in \reffig{fig:abs-stack-1} by allowing the
  operations to be executed concurrently by the processes in $P
  \subseteq Proc$. However, due to the enforced permissions, each
  process is guaranteed to update the stack without
  interference. Lynch \cite{Lyn96} refers to such a program as a
  \emph{canonical specification} of the concurrent stack.  Note that
  some idling must be allowed both before and after the main operation
  to enable interleaving to take place. 
  Without this idling, if $p \neq q$, we have: 
  \begin{derivation}
    \step{beh_{p, \{S\}}.APush(p,x) \land beh_{q, \{S\}}.APush(q,y)}

    \trans{\entails}{definitions}
    
    \step{\Eventually \mcR.S.p \land  OnlyAccessedBy.S.q}

    \trans{\equiv}{definitions}
    
    \step{false }
  \end{derivation}

  \begin{figure}[t]
  \centering
  $\begin{array}[t]{@{}r@{~~}c@{~~}l@{}}
    APush(p,x) & \sdef &   
    \enf OnlyAccessedBy.S.p \st SPush(x) \smallskip \\
    AEmpty(p,arv_p) & \sdef &   
    \enf OnlyAccessedBy.S.p \st 
      SEmpty(arv_p) \\
    ADoPop(p,arv_p) & \sdef &   
    \enf OnlyAccessedBy.S.p \st 
      SPop(arv_p) \\
    APop(p,arv_p) & \sdef &   
    AEmpty(p, arv_p) \sqcap APop(p, arv_p) \smallskip \\
    APP(p) & \sdef & \Idle \ch \left( \left(\textstyle\bigsqcap_{x : Val}\
      APush(p,x)\right) \sqcap  APop(p, arv_p)\right) \ch \Idle \smallskip \\ 
    AS(P) & \sdef &
    \Context{S}{\begin{array}[c]{@{}l@{}}
        \Init\ S = \aang{} \st  
        \Par_{p: P}\ 
        \Context{arv_p}{APP(p)^\omega}
      \end{array}}
  \end{array}
  $
  \figrule

  \caption{A canonical stack specification}
\label{fig:abs}
\end{figure}
 \hfill ${}_\clubsuit$
\end{example}

\subsection{Behaviour refinement}
\label{sec:behaviour-refinement}
We prove correctness of the concurrent data structure by proving
\emph{refinement} between the concurrent program and an abstract
specification and a concrete representation (e.g., in
\refsec{sec:line-via-forw} we show $LS(P)$ is a data refinement of
$AS(P)$). The sets of locations of the abstract program may differ
from those of the concrete program. Thus, we define a refinement
relation between commands parametrised by the sets of abstract and
concrete locations and the processes executing the command.

We first consider behaviour refinement --- a simple form a refinement
in which the concrete and abstract state spaced do not need to be
linked. 
\begin{definition}[Behaviour refinement]
  Suppose $A$ and $C$ are commands in contexts (sets of locations) $Y$
  and $Z$, respectively. We say $A$ is \emph{behaviour refined} by $C$
  with respect to a set of processes $P$ (denoted $A \sref_P^{Y,Z} C$)
  iff $beh_{P,Z}.C \entails beh_{P,Y}.A$ holds.
\end{definition}
We write $A \sref_p^{Y,Z} C$ for $A \sref_{\{p\}}^{Y,Z} C$ (i.e., the
set of processes is the singleton set $p$), write for $A \sref_P^Z C$
for $A \sref_P^{Z,Z} C$ (i.e., the concrete and abstract contexts are
identical) and write $A \sref_P C$ for $A \sref_P^{\emptyset} C$
(i.e., the abstract and concrete contexts are empty).  Because
refinement is defined by universal implication between behaviours, the
following monotonicity results may be proved in a straightforward
manner using monotonicity of the corresponding interval predicate
operators.

\begin{figure}[t]
  \centering
  \begin{minipage}[t]{0.42\columnwidth}
    \scalebox{0.7}{\input{intfree.pspdftex}}
    \caption{Conflicting execution -- interference freedom with two writes}
    \label{fig:intfree}
  \end{minipage}
  \qquad \qquad 
  \begin{minipage}[t]{0.46\columnwidth}
    \scalebox{0.7}{\input{intfree2.pspdftex}}
    \caption{Interference freedom with a read and write}
    \label{fig:intfree-nw}
  \end{minipage}

  \figrule
\end{figure}
\begin{example}
  \label{ex:APtoBP}
  The set of variables of the abstract program is given by: 
  \begin{eqnarray*}
    M & \sdef & S \cup \{arv_p \mid p: P\}
  \end{eqnarray*}
  We perform a behaviour refinement, where we show that the
  $OnlyAccessedBy$ permission in \reffig{fig:abs} may be weakened to
  $IntFree$ as given in program $BS(P)$ in \reffig{fig:iabs}, i.e., we
  prove $AS(P) \srefeq_P^{M} BS(P)$. Condition $BS(P) \sref_P^{M}
  AS(P)$ is trivial because $OnlyAccessedBy.S.p \imp
  Intfree.S.p$. Condition $BS(P) \sref_P^{M} AS(P)$. holds because for
  any processes $p, q \in P$ such that $p \neq q$ and intervals
  $\Delta_p$, $\Delta_q$ such that $\Delta_p \cap \Delta_q \neq
  \emptyset$ (i.e., $\Delta_p$ and $\Delta_q$ overlap), we have
  \begin{eqnarray*}
    beh_{p,S}.(BDoPush(p,arv_p) \sqcap BDoPop(p,arv_p)).\Delta_p &
    \entails & 
    \neg beh_{q,S}.(BDoPush(q,arv_q) \sqcap BDoPop(q,arv_q)).\Delta_q
  \end{eqnarray*}
  A visualisation of this is given in \reffig{fig:intfree}, where the
  write of process $p$ conflicts with the $IntFree.S.q$ condition in
  process $q$. Note that $AEmpty(p, arv_p) = BEmpty(p, arv_p)$, i.e.,
  one cannot replace the enforced property $OnlyAccessedBy.S.p$ by
  $IntFree.S.p$ because for example
  $$beh_{p,S}.(\enf IntFree.S.p \st SEmpty(p,arv_p)).\Delta_p \land beh_{q,
    S}.BDoPop(q,arv_q).\Delta_q
  $$
  is an allowable behaviour. This may be visualised as shown in
  \reffig{fig:intfree-nw}, where condition $IntFree.S.q$ in process
  $q$ does not conflict with the execution of process $q$. Replacement of $AS(P)$ by $BS(P)$
  simplifies the rest of the proof because the conditions that the
  implementation needs to ensure are weaker.  
  \begin{figure}[t]
  \centering
  $\begin{array}[t]{@{}r@{~~}c@{~~}l@{}}
    BPush(p,x) & \sdef &   
    \enf IntFree.S.p \st SPush(x) \smallskip \\
    BEmpty(p,arv_p) & \sdef &   
    \enf OnlyAccessedBy.S.p \st 
      SEmpty(arv_p) \\
    BDoPop(p,arv_p) & \sdef &   
    \enf IntFree.S.p \st 
      SPop(arv_p) \\
    BPop(p,arv_p) & \sdef &   
    BEmpty(p, arv_p) \sqcap BPop(p, arv_p) \smallskip \\
    BPP(p) & \sdef & \Idle \ch \left( \left(\textstyle\bigsqcap_{x : Val}\
      BPush(p,x)\right) \sqcap  BPop(p, arv_p)\right) \ch \Idle \smallskip \\ 
    BS(P) & \sdef &
    \Context{S}{\begin{array}[c]{@{}l@{}}
        \Init\ S = \aang{} \st  
        \Par_{p: P}\ 
        \Context{arv_p}{BPP(p)^\omega}
      \end{array}}
  \end{array}
  $
  \figrule

  \caption{An equivalent specification to \reffig{fig:abs}}
\label{fig:iabs}
\end{figure}
 \hfill ${}_\clubsuit$
\end{example}

\begin{lemma}
  \label{lem:command-ref}
  Suppose $P$ is a non-empty set of processes, $Y, Z \subseteq Var
  \cup Addr$, and $A$, $A_1$, $A_2$, $C$, $C_1$ and $C_2$ are commands
  such that $A \sref_P^{Y,Z} C$, $A_1 \sref_P^{Y,Z} C_1$ and $A_2
  \sref_P^{Y,Z} C_2$ hold, and $h$ and $g$ are interval
  predicates. Each of the following holds provided that $b' \imp b$,
  $\all k : Val \st \Eventually (e' = k) \entails \Eventually (e = k)$
  and $h \entails g$.
  \begin{eqnarray}
    \Guard{b} & \sref_P^{Y,Z} & \Guard{b'} \\
    \label{eq:33}
    v \asgn e & \sref_P^{Y,Z} & v \asgn e' \\
    \label{eq:41}
    A_1 \ch A_2 & \sref_P^{Y,Z} & C_1 \ch C_2 \\
    A_1 \sqcap A_2 & \sref_P^{Y,Z} & C_1 \sqcap C_2 \\
    A^\omega & \sref_P^{Y,Z} & C^\omega \\
    \label{eq:59}
    \Init g \st A & \sref_P^{Y,Z} & \Init h \st C
  \end{eqnarray}
\end{lemma}

Using an interval-based semantics to formalise a program's behaviour
allows one to obtain the following useful results which allows one to
split a command executed over a large interval into subintervals, and
combine a command executed over a number of adjoining intervals into
the same command over the larger interval.
\begin{lemma} 
  \label{lem:cmd-joins}
  For any non-empty set of processes $P$ and a command $C$, if
  $beh_P.C$ joins, then both of the following hold.
  \begin{eqnarray}
    C & \sref_P^{Y,Z} & C \ch C \\
    C & \sref_P^{Y,Z} & C \ch C^\omega 
  \end{eqnarray}
\end{lemma}
\begin{lemma} 
  \label{lem:cmd-splits}
  For any non-empty set of processes $P$ and a command $C$, if
  $beh_P.C$ splits, then.
  \begin{eqnarray}
    C \ch C & \sref_P^{Y,Z} & C
    % C^\omega & \sref_P^{Y,Z} & C 
  \end{eqnarray}
\end{lemma}
Note that $C^\omega \sref_P^{Y,Z} C$ holds trivially by the definition
of ${}^\omega$. The following lemma states that a guard evaluation is
equivalent to a program that performs some finite length idling, some
non-empty idling in which $c$ is guaranteed to hold, followed by some
more idling.
\begin{lemma}
  \label{lem:grd-to-chop}
  $[c] \srefeq_p^{Z} (\enf \Fin \st \Idle) \ch (\enf (\Always c \land
  \neg \Empty \land ReadAllLocs.c.p) \st \Idle) \ch \Idle$
\end{lemma}
Note that the behaviours of both $\enf \Fin \st \Idle$ and $\Idle$
hold in any empty interval, and hence each of the properties below
follow from \reflem{lem:grd-to-chop}:
\begin{eqnarray}
  [c] & \sref_p^{Z} & (\enf (\Always
  c \land \neg \Empty \land ReadAllLocs.c.p) \st \Idle) \ch \Idle
  \\
  {[}c] & \sref_p^{Z} &  (\enf \Fin \st \Idle) \ch (\enf\ (\Always c \land
  \neg \Empty \land ReadAllLocs.c.p) \st \Idle)
  \label{eq:72}
  \\
  {[}c] & \sref_p^{Z} &  \enf (\Always c \land
  \neg \Empty \land ReadAllLocs.c.p) \st \Idle
\end{eqnarray}
For example, \reflem{lem:command-ref} and \reflem{lem:grd-to-chop} may
be used to prove the following:
\begin{eqnarray}
  \label{eq:69}
  [b] & \sref_p^{Y,Z} &  [b]\ch [c]
\end{eqnarray}

The lemma below allows refinement within a wider context. 
\begin{lemma}
  \label{lem:fr-ref}
  If $A \sref^{W\cup Y, X \cup Z }_P C$, $Y \subseteq Z$, $W \cap Y
  = \emptyset = X \cap Z$ and $W \subseteq X$ then $\Context{W}{A}
  \sref_P^{Y,Z} \Context{X}{C}$.
\end{lemma}
\begin{proof}
  \begin{derivation}
    \step{beh_{P,Z}.\Context{X}{C}}
    \trans{\equiv}{behaviour definition}
    \step{Local.Z.p \land Z \cap X = \emptyset \land
      % fresh.X \land
      beh_{P,Z \cup X}.C}
    
    \trans{\entails}{assumptions}
    
    \step{Local.Y.p \land W \cap Y = \emptyset % \land
      % fresh.W
      \land beh_{P,Y \cup W}.A}
    
    \trans{\equiv}{behaviour definition}
    
    \step{beh_{P,Y}.\Context{W}{A} \hfill \qedhere}
    
  \end{derivation}
\end{proof}

We also obtain properties for refinement of enforced conditions.
The lemma below states that a refinement may be performed by
introducing a new enforced condition or by strengthening an existing
enforced condition \cite{Don09,DH09,DH12}.
\begin{lemma}
  \label{lem:enf-prop-intro-ref}
  If $C$ is a command, $g$ and $h$ are interval predicates, $P$ is a
  set of processes and $Y$, $Z$ are sets of locations, then both of the
  following hold.
  \begin{eqnarray}
    \label{eq:34}
    C & \sref_P^{Y,Z} & \enf{g}\st{C}
    \\
    \label{eq:10}
    h \entails g & \imp &  \enf{g}\st{C} \sref_P^{Y,Z} \enf{h}\st{C} 
  \end{eqnarray}
\end{lemma}
The lemma below allows decomposition within sequential choice and
iteration, provided that the interval predicate under consideration
joins.
\begin{lemma}
  \label{lem:enf-sc-ref}
  If $C$, $C_1$ and $C_2$ are commands, $g$ is an interval predicate
  that joins, $Y$, $Z$ are sets of locations and $P$ is a set of
  processes, then
  \begin{eqnarray*}
    (\enf{g}\st{C_1 \ch C_2}) & \sref_P^{Y,Z} &  (\enf{g}\st{C_1}) \ch
    (\enf{g}\st{C_2}) \\
    (\enf{g}\st{C^\omega}) & \sref_P^{Y,Z} & (\enf{g}\st{C})^\omega
  \end{eqnarray*}
\end{lemma}
The following lemma allows decomposition of behaviours with enforced
properties over parallel composition, sequential composition and
non-deterministic choice. Note that decomposition over sequential
composition requires that the interval predicate under consideration
splits. 
\begin{lemma}
  \label{lem:decompenf}
  If $A \sdef \|_{p:P} A_p$, $C \sdef \|_{p:P} C_p$, $A_1$, $A_2$,
  $C_1$, and $C_2$ are commands, $Y, Z \subseteq Var \cup Addr$ and
  $g$ and $h$ are interval predicates such that $h$ splits, then both
  of the following hold.
  \begin{eqnarray}
    \label{eq:75}
    (\all p : P \st A_p \sref_p^{Y,Z} (\enf g \st C_p))\ \  & \imp &\
    \  A \sref_P^{Y,Z} (\enf g \st C)
    \\
    \label{eq:76}
    (A_1 \sref_P^{Y,Z} (\enf{h}\st{C_1})) \land 
    (A_2 \sref_P^{Y,Z} (\enf{h}\st{C_2}))\ \  &\imp &\ \   (A_1
    \ch A_2 \sref_P^{Y,Z} (\enf{h}\st{C_1 \ch C_2}))
    \\
    \label{eq:51}
    (A_1 \sref_P^{Y,Z} (\enf{g}\st{C_1})) \land 
    (A_2 \sref_P^{Y,Z} (\enf{g}\st{C_2}))\ \  &\imp &\ \   (A_1
    \sqcap A_2 \sref_P^{Y,Z} (\enf{g}\st{C_1 \sqcap C_2}))
  \end{eqnarray}

\end{lemma}

\section{A coarse-grained linearisable abstraction}
\label{sec:coarse-grain-line}

In this section, we develop the coarse-grained abstraction of the
Trieber stack (\refsec{sec:abstraction-mcl}) and develop
interval-based data refinement rules
(\refsec{sec:data-refinement}). In \refsec{sec:line-via-forw}, we
prove that the coarse-grained abstraction is linearisable with respect
to the canonical stack specification from \reffig{fig:abs} via data
refinement and in \refsec{sec:remark:-import-prov}, we discuss the
importance of proving linearisability of coarse-grained abstractions.

\subsection{The abstraction $LS(P)$}
\label{sec:abstraction-mcl}

\begin{figure}[t]

  \centering
  
  $\begin{array}[t]{@{}r@{~~}c@{~~}l@{}}
    EnvSt(p) & \sdef & \enf 
    \begin{array}[t]{@{}l@{}}
      \begin{array}[t]{@{}l@{}}
      \neg WriteSomeLoc.(SAddr \cup \{Top% ,n_p \cdota key
      \}).p % \land \\
      % \neg
      % ReadSomeLoc.(SAddr \bs \{ptr.(\derefe Top) \cdota key, ptr.(\derefe Top) \cdota nxt\})
    \end{array}
    \end{array}
    \st 
    \True
    \\
    \\
    LSetup(p, x) & \sdef &
    % n_p \nasgn FAddr \ch FAddr \asgn FAddr \bs \{n_p\} 
    newNode(n_p) \ch (n_p\cdota key) \hasgn x 
    \smallskip \\ 
    LDoPush(p) & \sdef &     
    \begin{array}[t]{@{}l@{}}
      \enf 
      IntFree.(SAddr \cup \{Top, n_p \cdota key, n_p \cdota
      nxt\}).p \st \qquad \qquad \\
      \hfill (n_p\cdota nxt) \hasgn ptr.(\derefe Top)
      \ch  Top \hasgn (n_p, ctr.(\derefe Top) + 1)
    \end{array}
    \smallskip \\ 
    LPush(p,x) & \sdef &
    LSetup(p, x) \ch EnvSt(p) \ch LDoPush(p) 
    \\ 
    \\
    LEmpty(p, rv_p) & \sdef & \Guard{ptr.(\derefe Top) = null} \ch rv_p \asgn
    Empty \smallskip \\ 
    LDoPop(p, rv_p) & \sdef &
    \enf IntFree.(SAddr \cup \{Top\}).p \st 
    \begin{array}[t]{@{}l@{}}
        (\Guard{ptr.(\derefe Top) \neq null} \ch rv_p \asgn ptr. (\derefe Top) \mapsto key \ch {}\\
         Top \hasgn
        (ptr. (\derefe Top) \mapsto nxt, ctr. (\derefe Top)+1))
    \end{array}
    \smallskip \\
    LPop(p, rv_p) & \sdef & EnvSt(p) \ch (LEmpty(p, rv_p) \sqcap
    LDoPop(p, rv_p)) 
    \\
    \\
    LPP(p) & \sdef & \Idle \ch \left(\bigsqcap_{x:Val}
      LPush(p,x)\right) \sqcap 
    LPop(p,arv_p) \ch \Idle\smallskip \\
    LS(P) & \sdef &
    \Context{Top, FAddr }{
      \begin{array}[c]{@{}l@{}}
        \Init TInit \st % \rely FreeAddr
        % \st
        \Par_{p: P} 
        \Context{n_p, rv_p}{LPP(p)^\omega}
    \end{array}}
  \end{array}$
  \figrule
  \caption{A coarse-grained abstraction $LS(P)$ of the Treiber Stack}
\label{fig:lin-abs}
\end{figure}

The coarse-grained abstraction (see \reffig{fig:lin-abs}) uses $Top$
to obtain a pointer to the top of the stack. In addition, each process
uses local variables $n_p$ and $rv_p$ for the new node to be inserted
(by the push) and the value to be returned (by the pop),
respectively. The stack addresses $SAddr$ are defined to be the set of
addresses that are reachable from $Top$. Hence, for a state $\sigma$,
we define the following:
\begin{eqnarray*}
  iter_0.\sigma & \sdef & ptr.(\sigma.Top)
  \\
  iter_{n+1}.\sigma & \sdef &\left\{
    \begin{array}[c]{@{}l@{\qquad}l}
      null & \textrm{if $eval.(iter_n.\sigma \mapsto nxt).\sigma = null$}
      \\
      eval.(iter_n.\sigma \mapsto nxt).\sigma & {\rm otherwise}
    \end{array}\right.
  \\
  SAddr.\sigma & \sdef & \{node \in Addr \mid \exists n : \nat \st
  node \in \{eval.(iter_n\cdota key).\sigma, eval.(iter_n\cdota nxt).\sigma\}  \} 
\end{eqnarray*}
Hence, $SAddr.\sigma$ denotes the addresses that are reachable from
$Top$. 

Within program $LS(p)$ in \reffig{fig:lin-abs}, each process $p$
iteratively chooses executes $LPP(p)$, which at each iteration
performs some idling, then for a non-deterministically chosen value
$x$ executes $LPush(p,x)$ or $LPop(p, rv_p)$, and then performs some
more idling. Within $LPush(p,x)$, command $LSetup(p, x)$ initialises
the push and $EnvSt(p)$ allows some execution that does not modify the
stack, and $LDoPush(p)$ performs the actual push. The coarse-grained
pop is modelled by $LPop(p)$, which allows some initial
(non-interfering) idling, then either behaves as $LEmpty(p, rv_p)$, which
models a coarse-grained pop on an empty stack, or $LDoPop(p, rv_p)$, which
models a coarse-grained pop that removes the top element of a
non-empty stack.

Like many CAS-based implementations, Treiber's stack may retry its
main operation a number of times before succeeding. An abstraction of
this behaviour is modelled by the command $EnvSt(p)$, which guarantees
that none of the addresses corresponding to the stack have been
modified.

For the rest of this paper, we assume that the set of variables of the
coarse-grained abstraction is given by $L$, which is defined as follows.
\begin{eqnarray*}
  L & \sdef & \{Top\} \cup FAddr \cup \bigcup_{p:P}\{n_p, rv_p\}
\end{eqnarray*}

\begin{example}
  One can prove that the addresses in $SAddr$ (i.e., those
  corresponding to the stack) are not modified by $LS(P)$. In
  particular one may prove
  \begin{eqnarray}
    \label{eq:56}
    beh_{P, L}.LS(P) & \entails & \all p : P \st  \neg WriteSomeLoc.SAddr.p
  \end{eqnarray}
\end{example}

\subsection{Data refinement}
\label{sec:data-refinement}

It is possible to prove behaviour refinement between the
coarse-grained program $LS(P)$ and the fine-grained implementation
$TS(P)$. However, it is not immediately obvious that $LS(P)$ is
linearisable. It turns out that a linearisability proof of $LS(P)$ is
non-trivial, however, the proof is simplified because large portion of
the code are executed atomically.

As shown by Doherty et al \cite{CDG05,DGLM04,Doherty03} and again by
Derrick et al \cite{SWD12,DSW11TOPLAS,DSW11}, a sound method for
proving linearisability is to verify data refinement from an abstract
representation of the data structure being implemented to the concrete
program. Using the framework of input/output automata \cite{Lynch89},
Doherty \cite{Doherty03} constructs a so-called \emph{canonical
  automata} in which each operation executes by invoking the
operation, then executes an (internal) atomic step (corresponding a
step of the data type being implemented) and then returns to an idle
state. The argument made is that every trace of the canonical automata
is a linearisable \cite{Doherty03,CDG05,DGLM04} and hence any
refinement of this automata must also be linearisable. However,
because this claim is not formally verified, Derrick et al present an
extension that allows links data refinement to the Herlihy and Wing's
original definition \cite{SWD12,DSW11,DSW11TOPLAS}. In this paper, we
apply this extended method to an interval-based setting that allows
true concurrency.

An abstract program $\mcA$ is simulated by concrete program $\mcC$
with respect to a simulation predicate $sim$ if the initialisation of
$\mcC$ together with $sim$ implies the initialisation of $\mcA$ and any
behaviour of $\mcC$ over an interval in which $sim$ holds implies the
behaviour of $\mcA$, and additionally $sim$ holds throughout that
interval. % For types $X_1$, $X_2$, ordered pairs $(Y_1, Y_2), (Z_1,
For streams $s_1$ and $s_2$, we  define
\begin{eqnarray*}
s_1 \Cup s_2 & \sdef & \lambda t : Time \st s_1.t \cup s_2.t
\end{eqnarray*}
If the state spaces corresponding to $s_1$ and $s_2$ are disjoint,
then for each $t \in Time$, $(s_1 \Cup s_2).t$ is a state and hence
$s_1 \Cup s_2$ is a stream.
\begin{definition}[Data refinement]
  \label{def:forward-sim}
  Suppose $P \subseteq Proc$, $Y, Z \subseteq Var \cup Addr$ are
  locations such that $Y \cap Z = \emptyset$, $AInit \in StatePred_Y$, and
  $CInit\in StatePred_Z$. We say $\mcA \sdef \Context{Y}{\Init AInit
    \st A}$ is \emph{data refined} by $\mcC \sdef \Context{Z}{\Init
    CInit \st C}$ with respect to a \emph{simulation predicate} $sim
  \in StatePred_{Y \cup Z}$ iff both:
  \begin{eqnarray}
    & & \all \sigma_a : State_Y, \sigma_c : State_Z \st CInit.\sigma_a
    \land sim.(\sigma_a \cup \sigma_c)
    \imp AInit.\sigma_a
    \label{eq:refinit}
    \\
    & & \begin{array}[c]{@{}l@{}}
      \all s_c : Stream_Z, \Delta :
      Interval, \sigma : State_Y \st \\
      \qquad sim.(\sigma \cup s_c.(\glb.\Delta - 1))
      \land beh_{P, Z}.C.\Delta.s_c \imp {} \\
      \qquad \qquad \exists s_a :  Stream_Y
      \st (\sigma = s_a.(\glb.\Delta - 1)) \land \Always sim.\Delta.(s_a
      \Cup s_c) \land beh_{P,Y}.A.\Delta.s_a
    \end{array}
    \label{eq:ref}
  \end{eqnarray}
\end{definition}
This establishes a data refinement \cite{deRoever-Englhardt98} between
the abstract program $A$ and concrete program $C$ in
interval-based setting, where the programs execute in parallel in a
truly concurrent manner. As in traditional data refinement, our
definition relies on a refinement relation $sim$ which links the
concrete and abstract states. Condition \refeq{eq:refinit} ensures that
every initialisation of the concrete program that is related to an
abstract state via $sim$ must imply an initialisation of the abstract
program. By \refeq{eq:ref}, for every concrete stream $s_c$, interval
$\Delta$ and abstract state $\sigma$, provided that both 
\begin{enumerate}
\item $sim$ holds between $\sigma$ and the state of $s_c$ just before
  $\Delta$ and
\item the concrete program executes in $s_c$ over $\Delta$
\end{enumerate}
then there exists an abstract stream $s_a$ such that
\begin{enumerate}
\item the state of $s_a$ that immediately precedes $\Delta$ is $\sigma$,
\item $sim$ holds in the combined stream $s_a \Cup s_a$ throughout $\Delta$, and
\item the abstract program executes in $s_a$ over $\Delta$. 
\end{enumerate}

Proving condition \refeq{eq:ref} directly is difficult because it does
not decompose. However, a predicate of the form $p \imp (\exists x \st
q \land r)$ holds if both $p \imp \exists x \st q$ and $\all x \st p
\land q \imp r$ hold. Hence, \refeq{eq:ref} can be proved by showing
that both of the following hold.
\begin{eqnarray}
  \begin{array}[t]{@{}l@{}}
    \all s_c : Stream_Z, \Delta :
    Interval, \sigma : State_Y \st \\
    \qquad sim.(\sigma \cup s_c.(\glb.\Delta - 1))
    \land beh_{P, Z}.C.\Delta.s_c \imp {} \\
    \qquad \qquad \exists s_a :  Stream_Y
    \st (\sigma = s_a.(\glb.\Delta - 1)) \land \Always sim.\Delta.(s_a 
    \Cup s_c) 
  \end{array}
  \label{eq:ref1}\\
  \begin{array}[t]{@{}l@{}}
    \all s_c : Stream_Z, \Delta :
    Interval, s_a : Stream_Y \st \\
    \qquad \Always sim.\Delta.(s_a \Cup s_c)
    \land beh_{P, Z}.C.\Delta.s_c \imp beh_{P,Y}.A.\Delta.s_a
  \end{array}
  \label{eq:ref2}
\end{eqnarray}
By \refeq{eq:ref1}, the simulation condition $sim$ must be such that
for any execution of the concrete program that executes from a state
that satisfies $sim$, there must exist an abstract stream such that
$sim$ holds throughout the interval.  To simplify representation of
intervals of the form in \refeq{eq:ref1}, we introduce the following
notation.
\begin{eqnarray*}
  (g \linkto_Y c).\Delta.s & \sdef & \begin{array}[t]{@{}l@{}}
    \all \sigma : State_Y \st 
    \begin{array}[t]{@{}l@{}}
      c.(\sigma \cup s.(\glb.\Delta - 1)) \land g.\Delta.s \imp
      \qquad \\ \exists s_y : Stream_Y \st \sigma =
      s_y.(\glb.\Delta - 1) \land \Always c.\Delta.(s_y \Cup s)
    \end{array}
  \end{array}
  \\
  g \dref_{Y, Z} c & \sdef & \begin{array}[t]{@{}l@{}}
      \all s_z : Stream_Z, \Delta : Interval% , \sigma : State_Y
      \st 
      (g \linkto_Y c).\Delta.s_z
    %   \\
    %   \qquad 
    %   c.(\sigma \cup s_z.(\glb.\Delta - 1)) \land g.\Delta.s_z \imp
    %   \qquad \\ \qquad \exists s_y : Stream_Y \st \sigma =
    % s_y.(\glb.\Delta - 1) \land \Always c.\Delta.(s_y \Cup s_z)
    \end{array}
\end{eqnarray*}

By \refeq{eq:ref2} for any concrete and abstract streams and interval,
if the concrete program executes in the interval and forward
simulation holds throughout the interval, then it must be possible to
execute the abstract program in the interval and stream. Because $Y
\cap Z = \emptyset$, we may further simplify \refeq{eq:ref2} to
$\Always sim \land beh_{P, Z}.C \entails beh_{P,Y}.A$, which may be
written using behaviour refinement and enforced properties as
\begin{eqnarray}
  \label{eq:18}
   A & \sref_P^{Y,Z} & \enf\Always sim \st C
\end{eqnarray}
We prove this as follows: 
\begin{derivation}
  \step{\refeq{eq:ref2}}
    % \begin{array}[t]{@{}l@{}}
    %   \all s_c : Stream_Z, \Delta :
    %   Interval, s_a : Stream_Y \st\\
    %   \qquad \Always sim.\Delta.(s_a \Cup s_c)
    %   \land beh_{P, Z}.C.\Delta.s_c \imp beh_{P,Y}.A.\Delta.s_a
    % \end{array}}

  \trans{=}{$Y \cap Z = \emptyset$}

  \step{
    \begin{array}[t]{@{}l@{}}
      \all s_c : Stream_Z, \Delta :
      Interval, s_a : Stream_Y \st \\
      \qquad \Always sim.\Delta.(s_a \Cup s_c)
      \land beh_{P, Z}.C.\Delta. (s_a \Cup s_c) \imp beh_{P,Y}.A.\Delta. (s_a \Cup s_c)
    \end{array}}

  \trans{=}{logic}

  \step{
    \begin{array}[t]{@{}l@{}}
      \all s : Stream_{Y \cup Z}, \Delta :
      Interval \st \Always sim.\Delta.s
      \land beh_{P, Z}.C.\Delta.s \imp beh_{P,Y}.A.\Delta.s
    \end{array}}

  \trans{=}{definition of `$\entails$' and `$\sref$'}

  \step{\refeq{eq:18}}

  % \trans{=}{definition}

  % \step{\enf \Always sim \st A \sref_P^{Y,Z} C }

\end{derivation}

Condition \refeq{eq:ref1} can also be decomposed. First, we prove the
following lemma that allows one to decompose a proof of a simulation
predicate over chop, which in turn enables decomposition of sequential
composition. Similar proof techniques for relational frameworks are
well studied \cite{deRoever-Englhardt98}.
\begin{lemma}
  \label{lem:simulation}
  Suppose $p \in Proc$, $Y, Z \subseteq Var \cup Addr$ such that $Y
  \cap Z = \emptyset$, $g_1, g_2 \in IntvPred_Z$ and $ c \in
  StatePred_{Y \cup Z}$. Then $((g_1 \ch g_2) \dref_{Y, Z}  c)$
  % \begin{eqnarray*}
  %   % \begin{array}[b]{@{}l@{}}
  %   %   \all s_z : Stream_Z, \Delta : Interval, \sigma : State_Y \st \\
  %   %   \qquad c.(\sigma \cup s_z.(\glb.\Delta - 1)) \land (g_1 \ch
  %   %   g_2).\Delta.s_z
  %   % \end{array}
  %   % & \imp & \exists s_y : Stream_Y \st \sigma =
  %   % s_y.(\glb.\Delta - 1) \land \Always c.\Delta.(s_y \Cup s_z)
  % \end{eqnarray*}
  holds if both $(g_1 \mathbin{\dref_{Y, Z}}  c)$ and $(g_2
  \mathbin{\dref_{Y, Z}} c)$ hold. 
\end{lemma}
\begin{proof}
  For an arbitrarily chosen $\sigma \in State_Y$, $\Delta \in
  Interval$ and $s_z \in Stream_Z$, we have the following
  calculation. 
  \begin{derivation}
    \step{
      \begin{array}[t]{@{}l@{}}
        c.(\sigma \cup s_z.(\glb.\Delta - 1)) \land (g_1 \ch
        g_2).\Delta.s_z 
       %  \imp \\ 
       % \exists s_y : Stream_Y \st \sigma = s_y.(\glb.\Delta - 1)
       %  \land \Always c.\Delta.(s_y \Cup s_z)
      \end{array}
    }

    \trans{=}{definition of `;'}

    \step{
      \begin{array}[t]{@{}l@{}}
        c.(\sigma \cup s_z.(\glb.\Delta - 1)) \land  (\exists \Delta_1, \Delta_2: Interval \st (\Delta_1 \cup \Delta_2 =
        \Delta)
        \land (\Delta_1 \adjoins \Delta_2) \land g_1.\Delta_1 \land 
        g_2.\Delta_2.s_z)
      % \imp \\ 
      % \exists s_y : Stream_Y \st \sigma = s_y.(\glb.\Delta - 1)
      % \land \Always c.\Delta.(s_y \Cup s_z)
    \end{array}
  }

    \trans{\imp}{logic and $\glb.\Delta_1 = \glb.\Delta$}

    \step{
      \exists \Delta_1, \Delta_2: Interval \st 
      \begin{array}[t]{@{}l@{}}
      (\Delta_1 \cup \Delta_2 =
      \Delta)
      \land (\Delta_1 \adjoins \Delta_2)  \land 
      c.(\sigma \cup s_z.(\glb.\Delta - 1)) \land g_1.\Delta_1.s_z \land 
      g_2.\Delta_2.s_z) 
      % \imp \\
      % \exists s_y : Stream_Y \st \sigma = s_y.(\glb.\Delta - 1)
      % \land \Always c.\Delta.(s_y \Cup s_z)
    \end{array}
  }

    % \trans{\imp}{$\Always c$ joins, $\glb.\Delta_1 = \glb.\Delta$}

    % \step{
    %   \exists \Delta_1, \Delta_2: Interval \st
    %   \begin{array}[t]{@{}l@{}}
    %     \Delta_1 \cup \Delta_2 = \Delta 
    %     \land \Delta_1 \adjoins \Delta_2 \land 
    %     c.(\sigma \cup s_z.(\glb.\Delta_1 - 1)) \land
    %     g_1.\Delta_1.s_z \land 
    %     g_2.\Delta_2.s_z 
    %     % \imp \\ 
    %     % \exists s_y : Stream_Y \st \sigma =
    %     % s_y.(\glb.\Delta_1 - 1) 
    %     % \land \Always c.\Delta_1.(s_y \Cup s_z) \land \Always
    %     % c.\Delta_2.(s_y \Cup s_z) 
    %   \end{array}
    % }

    \trans{\imp}{use $(g_1 \dref_{Y,Z} c)$ % \refeq{eq:8}
      and logic}

    \step{
      \exists \Delta_1, \Delta_2: Interval, s_y : Stream_Y \st  
      \begin{array}[t]{@{}l@{}}
        (\Delta_1 \cup \Delta_2 = \Delta )
        \land (\Delta_1 \adjoins \Delta_2) \land \\
        \sigma =
        s_y.(\glb.\Delta_1 - 1) 
        \land \Always c.\Delta_1.(s_y \Cup s_z) \land 
        g_2.\Delta_2.s_z
      \end{array}
    }

    \trans{=}{use $\Always c.\Delta_1.(s_y \Cup s_z)$, 
      $\Delta_1 \adjoins \Delta_2$ }

    \step{
      \exists \Delta_1, \Delta_2: Interval, s_y : Stream_Y \st  
      \begin{array}[t]{@{}l@{}}
        (\Delta_1 \cup \Delta_2 = \Delta )
            \land (\Delta_1 \adjoins \Delta_2) \land \\
            \sigma =
            s_y.(\glb.\Delta_1 - 1) 
            \land \Always c.\Delta_1.(s_y \Cup s_z) \land \\
            (\exists \sigma' : State_Y \st c.(\sigma' \cup
            s_z.(\glb.\Delta_2 - 1))) \land 
            g_2.\Delta_2.s_z
      \end{array}
    }

    \trans{\imp}{logic and $(g_2 \dref_{Y,Z} c)$}%(\ref{eq:11})}

    \step{
      \exists \Delta_1, \Delta_2: Interval, s_y, s_y' : Stream_Y,\sigma' : State_Y \st 
      \begin{array}[t]{@{}l@{}}
            (\Delta_1 \cup \Delta_2 = \Delta )
            \land (\Delta_1 \adjoins \Delta_2) \land \\
            \sigma =
            s_y.(\glb.\Delta_1 - 1) 
            \land \Always c.\Delta_1.(s_y \Cup s_z) \land \\
            \sigma' = s_y'.(\glb.\Delta_2 - 1) 
            \land \Always c.\Delta_2.(s_y' \Cup s_z) 
      \end{array}
    }

    \trans{\imp}{can pick $s_y'' \in Stream_Y$ such that $\all t : \Delta_1 \st s_y''.t =
      s_y.t$ and $\all t : \Delta_2 \st s_y''.t =
      s_y'.t$,}
    \trans{}{then rename $s_y''$ to $s_y$}
    
    \step{
      \exists \Delta_1, \Delta_2: Interval, s_y : Stream_Y,\sigma' : State_Y \st 
      \begin{array}[t]{@{}l@{}}
        (\Delta_1 \cup \Delta_2 = \Delta)
        \land (\Delta_1 \adjoins \Delta_2) \land \\
        \sigma =
        s_y.(\glb.\Delta_1 - 1) 
        \land \Always c.\Delta_1.(s_y \Cup s_z) \land \\
        \sigma' = s_y.(\glb.\Delta_2 - 1) 
        \land \Always c.\Delta_2.(s_y \Cup s_z) 
      \end{array}
    }

    \trans{\imp}{logic, $\Always c$ joins, and $\glb.\Delta_1 = \glb.\Delta$}
    
    \step{
      \exists s_y : Stream_Y \st 
      \begin{array}[t]{@{}l@{}}
            \sigma =
            s_y.(\glb.\Delta - 1) 
            \land \Always c.\Delta.(s_y \Cup s_z)
      \end{array}
    \hfill \qedhere}

  \end{derivation}
\end{proof}

We use the following lemma to further decompose proof obligation
\refeq{eq:ref1}. 
\begin{lemma}
  \label{lem:link-split}
  If $Y, Z \subseteq Var$ such that $Y \cap Z = \emptyset$, $g \in
  IntvPred_Z$, $s_z \in Stream_Z$ and $\Delta \in Interval$.  Then $
  (g \linkto_{Y} c).\Delta.s_z $ holds
  % \begin{eqnarray}
  %     \label{eq:13}
  %   % c.(\sigma \cup s_z.(\glb.\Delta - 1))
  %   % \land beh_{P, Z}.C.\Delta.s_z & \imp & \exists s_y :  Stream_Y
  %   % \st (\sigma = s_y.(\glb.\Delta - 1)) \land \Always c.\Delta.(s_y 
  %   % \Cup s_z) 
  % \end{eqnarray}
  holds if there exists a $w \in StatePred_Z$ such that:
  \begin{eqnarray}
    \label{eq:46}
    g.\Delta.s_z \imp \Box(\Always  w \linkto_{Y}
    c).\Delta.s_z \land   \Box(\Always \neg w \linkto_{Y}
    c).\Delta.s_z
  \end{eqnarray}
\end{lemma}
\begin{proof}
  For an arbitrarily chosen $\sigma \in State_Y$, $\Delta \in
  Interval$ and $s_z \in Stream_Z$ and prefix $\Delta'$ of $\Delta$,
  we have that either $\Always w.\Delta'.s_z$ or $\Always \neg
  w.\Delta'.s_z$ holds. By \refeq{eq:46}, we have that and a similar calculation to
  the proof of \reflem{lem:simulation}. \hfill \qedhere
  % \begin{derivation}
  %   \step{
  %     \begin{array}[t]{@{}l@{}}
  %       c.(\sigma \cup s_z.(\glb.\Delta - 1)) \land g.\Delta.s_z 
  %      %  \imp \\ 
  %      % \exists s_y : Stream_Y \st \sigma = s_y.(\glb.\Delta - 1)
  %      %  \land \Always c.\Delta.(s_y \Cup s_z)
  %     \end{array}
  %   }
  %   \trans{=}{}
    
  %   \step{c.(\sigma \cup s_z.(\glb.\Delta - 1)) \land g.\Delta.s_z
  %     \land (\exists \delta : part.\Delta \st \all i : \dom.\delta
  %     \st \Always w.(\delta.i).s_z \lor \Always \neg w.(\delta.i).s_z)}
  % \end{derivation}
\end{proof}
\begin{figure}[t]
  \centering 

  $\begin{array}[t]{@{}r@{~~}c@{~~}l@{}} 
    \mathsf{start\_record}_p (v, e) & \sdef &
    \exists k \st \prev.\ora{e = k} \land \ola{v = k \land \mcW.v.p}
    \\
    \mathsf{end\_record}_p (v, e) & \sdef &
    \exists k \st \ora{e = k} \ch \ceil{(v = k) \land \mcW.v.p}
    \\
    ARecord(Inv, Res) & \sdef & \mathsf{end\_record}_p(HA, HA \cat \aang{Inv,Res})
    \\
    LRecord(Inv, Res) & \sdef &
    \begin{array}[t]{@{}l@{}}
      \mathsf{start\_record}_{p}(HL, HL
      \cat \aang{Inv}) \land 
      \mathsf{end\_record}_{p}(HL, HL \cat \aang{Res})
    \end{array}
    \medskip \\
    HAPush(p,x) & \sdef &
    \begin{array}[c]{@{}ll@{}}
      \enf ARecord(push_p^I(x),push_p^R)
      \st BPush(p, x) 
    \end{array}
    \smallskip 
    \\
    HAEmpty(p, arv_p) & \sdef &
    \begin{array}[c]{@{}ll@{}}
      \enf ARecord(pop_p^I,pop_p^R(arv_p))
      \st BEmpty(p, arv_p) 
    \end{array}
    \\
    HADoPop(p, arv_p) & \sdef &
    \begin{array}[c]{@{}ll@{}}
      \enf ARecord(pop_p^I,pop_p^R(arv_p))
      \st BPop(p, arv_p) 
    \end{array}
    \\
    HAPop(p, arv_p) & \sdef &
    HAEmpty(p, arv_p) \sqcap HADoPop(p, arv_p)
    \smallskip \\
    HAPP(p, arv_p) & \sdef & \Idle \ch \left( \left(\textstyle\bigsqcap_{x : Val}\
      HAPush(p,x)\right) \sqcap  HAPop(p, arv_p)\right) \ch \Idle
    \medskip \\
    HAS(P) & \sdef & \Context{S}{\begin{array}[c]{@{}l@{}} \Init S = \aang{} \st
        \Par_{p: P}\ 
        HAPP(p)^\omega
      \end{array}}
    \\
    \\
    HLPush(p,x) & \sdef &
    \begin{array}[t]{@{}l@{}}
      \enf
      \begin{array}[t]{@{}l@{}}
        LRecord(push_p^I(x),push_p^R)\st LPush(p,x)
      \end{array}
      % \qquad \qquad \\
      % \hfill  lh_1: LSetup(p, x) \ch lh_2: EnvSt(p) \ch lh_3:
      %   LDoPush(p) 
    \end{array}
    \\ 
    HLEmpty(p,rv_p) & \sdef &
    \begin{array}[t]{@{}l@{}}
      \enf
      LRecord(pop_p^I, pop_p^R(rv_p)) \st LEmpty(p, rv_p)
      % \st \\
      % lp_1: EnvSt(p) \ch (lp_2: LEmpty(p, rv_p) \sqcap lp_3: LDoPop(p, rv_p))
    \end{array}
    \\ 
    HLDoPop(p,rv_p) & \sdef &
    \begin{array}[t]{@{}l@{}}
      \enf
      LRecord(pop_p^I, pop_p^R(rv_p)) \st LDoPop(p, rv_p)
    \end{array}
    \\
    \\
    HLPP(p) & \sdef & \Idle \ch \left(\left(\bigsqcap_{x:Val}
      HLPush(p,x)\right) \sqcap 
    HLPop(p, rv_p)\right) \ch \Idle \smallskip \\
    HLS(P) & \sdef &
    \Context{Top, FAddr}{\begin{array}[c]{@{}l@{}}
        \Init 
            TInit
          \st  \Par_{p: P} 
      \Context{n_p, rv_p}{HLPP(p)^\omega}
    \end{array}}
  \end{array}$

  \figrule
  \caption{Programs $BS(P)$ and $LS(P)$ extended with history, status
    and labels}
  \label{fig:abs-hsl}
\end{figure}

\subsection{Linearisability of $LS(P)$ via simulation}
\label{sec:line-via-forw}

Derrick et al show that a proof of linearisability can be reduced to a
proof of data refinement by encoding the definition of linearisability
from \refsec{sec:an-altern-defin} within the simulation relation and
extending the abstract and concrete programs with histories of
invocations and responses \cite{DSW11TOPLAS,DSW11,SWD12}. Programs
$AS(P)$ and $LS(P)$ extended with histories are respectively given by
programs $HAS(P)$ and $HLS(P)$ in \reffig{fig:abs-hsl}. In particular,
the canonical program produces a sequential history by recording an
invocation immediately followed by the matching response in history
$HA$ at the end of the executions of both $HAPush(p, x)$ and $HAPop(p,
arv_p)$. On the other hand, the coarse-grained atomic program $LS(P)$
is extended to $HLS(P)$ so that invocations and responses of
$HLPush(p, x)$ and $HLPop(p, arv_p)$ in history $HL$. Note that
invocations and responses of operations of $HLS(P)$ may not be sequential, i.e.,
other processes may be arbitrarily interleaved between any matching
pair of events.

In addition, one must also establish a relationship between concrete
stack (represented as a linked list) and the abstract stack
(represented as a sequence of values). Hence, for an concrete state
$\sigma \in State_L$, we define:
\begin{eqnarray*}
  Stack.\sigma    & \sdef & \klet len = \frac{size.SAddr.\sigma}{2} \kin
  \lambda n : 0..len-1 \st (iter_n.Top \mapsto key).\sigma
\end{eqnarray*}
\begin{example}
  \begin{figure}[t]
    \centering
    \scalebox{0.7}{\input{stack.pspdftex}}
    \caption{A state corresponding to abstract stack $\aang{aa, bb,
        cc}$, where $ptr.(\derefe Top) = X$ and $ctr.(\derefe Top) = 14$}
    \label{fig:ex-abs-stack}
  \end{figure}
    Consider the abstract stack corresponding to the state $\sigma$
    depicted in \reffig{fig:ex-abs-stack}. We have
  $$\begin{array}[t]{l@{\qquad\qquad}l}
    iter_0.\sigma
    \begin{array}[t]{@{}l@{}}
      = ptr.(\sigma.Top) \\
      = X
    \end{array}
    & 
    iter_1.\sigma
    \begin{array}[t]{@{}l@{}}
      = eval.(iter_0.\sigma \mapsto nxt).\sigma \\
      = eval.(\deref (X + offset.nxt)).\sigma \\
      = eval.(\deref (X + 1)).\sigma \\
      = \sigma.(X + 1) \\
      = Y
    \end{array}
  \end{array}
  $$
  Similarly, $iter_2.\sigma = Z$ and $iter_3.\sigma
  = null$. Hence, $Stack.\sigma = \aang{aa, bb, cc}$
  \hfill ${}_\clubsuit$
\end{example}

\begin{theorem}
  \label{thm:drefmain}
  $HAS(P)$ is data refined by $HLS(P)$ with respect to 
  \begin{eqnarray}
    simTS & \sdef &  (S
    = Stack) \land (\all p : P \st \mcW.S.p = \mcW.Top.p) \land
    linearisable(HL, HA) 
    % simTS & \sdef & 
    % (S = Stack) \land \label{eq:18} \\
    % % & & \left(
    % %    \all p : P, k_1, k_2 \st
    % %    \begin{array}[c]{@{}l@{}}
    % %      k_1 = eval.(ptr.(\derefe
    % %      Top) \cdota key) \land k_2 = eval.(ptr.(\derefe
    % %      Top) \cdota nxt) \imp \\
    % %       \mcR.S.p = (\mcR.Top.p \lor \mcR.k_1.p \lor \mcR.k_2.p)
    % %   \end{array}
    % % \right) \land \label{eq:40}\\
    % & & (\all p : P \st \mcW.S.p = \mcW.Top.p) \land \label{eq:44}\\
    % & & linearisable(HL, HA) \label{eq:45} 
  \end{eqnarray}
\end{theorem}
Condition \refeq{eq:18} establishes a relationship between concrete
stack (which is a linked list) and its abstract representation (which
is a sequence), % Condition \refeq{eq:40} states that for every process
ensures that for any process $p$, the process has permission to write
to $S$ iff it has permission to write to location $Top$ and that that
the concrete history $HC$ is linearisable with respect to the abstract
history $HA$.

The linearisation of push and a non-empty pop correspond to the
intervals in which $Top$ is modified, and the linearisation point of
the empty pop corresponds to an interval in which $ptr.(\derefe Top) =
null$ holds. Hence, we use \refeq{eq:72} to split the $HLEmpty(p,
rv_p)$ operation into the before, during and after cases of the
linearisation point as follows.
\begin{eqnarray*}
  HLEmptyPre(p, rv_p) & \sdef & \enf{\mathsf{start\_record}_{p}(HL, HL
    \cat \aang{pop_p^I}) \land \Fin}\st{\Idle} \\
  HLEmptyLin(p, rv_p) & \sdef & \enf{\Always (ptr.(\derefe Top) = null)
    \land \neg \Empty \land \mcR.Top.p}\st{\Idle}
  \\
  HLEmptyPost(p, rv_p) & \sdef & \enf{\mathsf{end\_record}_{p}(HL, HL
    \cat \aang{pop_p^I(rv_p)})}\st{rv_p \asgn Empty}
\end{eqnarray*}

To prove this, we first show the more straightforward property that
for any process $p$ and interval $\Delta$ in which $p$ has write
permission to $Top$ throughout $\Delta$, $p$ maintains $simTS$
throughout $\Delta$. Furthermore, $simTS$ is also maintained if no
process writes to $Top$.
\begin{lemma}
  % For the programs $HLS(P)$ given in \reffig{fig:abs-hsl}
  Both of the following hold: 
  \begin{eqnarray}
    \label{eq:16}
    % \all p: P \st 
    beh_{P,L}.HLS(P)% .\Delta.s_c
    & \entails & \Box (\Always (\exists p: P \st \mcW.Top.p)
    \linkto_{M} simTS)% ).\Delta.s_c  
    \\
    \label{eq:54}
    beh_{P,L}.HLS(P) & \entails &  \Box(\Always (\all p : P \st \neg \mcW.Top.p) 
    \linkto_{M} simTS)
  \end{eqnarray}
\end{lemma}
\begin{proof}[\ref{eq:16}]
  This holds because there is always a corresponding abstract stack
  after an update to $Top$ and furthermore intervals in which $Top$ is
  modified can be treated as linearisation intervals by appending
  corresponding invocation and response events to the abstract
  history.
\end{proof}
\begin{proof}[\ref{eq:54}]
  The only way to modify $simTS$ is to write to $HL$ and the only
  operations that modify $HL$ without writing to $Top$ are the
  invocations of each operation and the return operation
  $HLEmptyPost(p, rv_p)$. Every invocation is trivial because we may
  treat them as non-pending operations, and $HLEmptyPost(p, rv_p)$ is
  satisfied by treating operations that have executed $HLEmptyLin(p,
  rv_p)$ as a pending invocation.
  \hfill \qedhere
\end{proof}

We now return to the proof of the theorem that establishes data
refinement between $HLS(P)$ and $HAS(P)$.
\begin{proof}[\refthm{thm:drefmain}]
  The proof of the initialisation condition \refeq{eq:refinit} is
  trivial. The main condition \refeq{eq:ref} is proved by first
  splitting the proof into conditions \refeq{eq:ref1} and
  \refeq{eq:ref2}. To prove \refeq{eq:ref1}, i.e., 
  \[
  beh_{P,L}.HLS(P) \dref_{M, L} simTS
  \]
  we apply
  \reflem{lem:link-split} with $w$ instantiated to $\exists p : P \st
  \mcW.Top.p$ and use \refeq{eq:16} and \refeq{eq:54}. We must now
  prove \refeq{eq:ref2}, which we have shown above holds by
  \refeq{eq:18}, i.e.,
  \begin{eqnarray*}
    HAS(P) & \sref_P^{M,L} & \enf{\Always simTS} \st HLS(P)
  \end{eqnarray*}
  Using \reflem{lem:decompenf} (i.e., decomposition of parallel
  composition) and \reflem{lem:command-ref} (i.e., monotonicity of
  ${}^\omega$, this may be proved by showing that for any $p \in P$.
  \begin{eqnarray*}
    HAPP(p) & \sref_p^{M,L} & \enf{\Always simTS} \st HLPP(p)
  \end{eqnarray*}
  Because $\Idle \srefeq_P^{Y,Z} \Idle \ch \Idle$, we may further
  decompose this using \reflem{lem:decompenf}, where the non-idle
  cases are given below, and the rest of the $HAPP(p)$ are refinements
  of $\Idle$
  \begin{eqnarray*}
    HAPush(p, x) & \sref_p^{M,L} & \enf
    \Always simTS \st LDoPush(p) \\
    HAEmpty(p, arv_p) & \sref_p^{M,L} & \enf
    \Always simTS \st HLEmptyLin(p, rv_p) \\
    HADoPop(p, arv_p) & \sref_p^{M,L} & \enf \Always simTS \st HLDoPop(p, rv_p) 
  \end{eqnarray*}
  Each of these proofs is straightforward due to assumption $\Always
  simTS$. \hfill\qedhere
\end{proof}

\subsection{Remark: The importance of proving linearisability}
\label{sec:remark:-import-prov}

It is also worth noting that some abstractions of the Treiber Stack
cannot be linearised, which reinforces the importance of proving
linearisability of a coarse-grained abstraction, as opposed to methods
such as \cite{TW11}, which only prove refinement between a concurrent
data structure and its coarse-grained abstraction without showing that
the abstraction itself is linearisable.
\begin{example}
  A pop abstraction
  $$ 
  \begin{array}[c]{@{}l@{}}
    \Guard{ptr.(\derefe Top) \neq null} \ch {} \\
    \enf{ \Always \neg \mcI.Top.p} \st {
      \begin{array}[t]{@{}l@{}}
        rv_p
        \asgn ptr.(\derefe Top) \mapsto key  \ch {} \\
        Top \hasgn (ptr.(\derefe Top) \mapsto
      nxt, ctr.(\derefe Top) + 1 )
    \end{array}
}
  \end{array}
  $$ 
  cannot be proved linearisable because it is possible for $Top$ to
  change after the $ptr.(\derefe Top) \neq null$ holds within the guard
  evaluation $\Guard{ptr.(\derefe Top) \neq null}$.  Hence, for example, when
  $Top$ is updated, the stack may already be empty. \hfill
  ${}_\clubsuit$
\end{example}
\begin{example}
  It is not necessary to strengthen $LEmpty(p)$ to
  $$
  \enf{ \Always \neg \mcI.Top.p} \st {\Guard{ptr.(\derefe Top) = null} \ch rv_p
    \asgn Empty}
  $$ 
  because the value of $Top$ is never used in the latter parts of the
  code. The guard $\Guard{ptr.(\derefe Top) = null}$ is merely used to decide
  whether or not the code should return. Although this strengthening
  does provide one with a coarse-grained program that is linearisable,
  the proof that the coarse-grained program is implemented by the
  Treiber Stack $TS(P)$ will be more difficult to achieve. 
  \hfill
  ${}_\clubsuit$
\end{example}

\section{Compositional proofs} 
\label{sec:compositional-proofs}
In this section, we describe how the proof of a command may be
decomposed into proofs of the subcomponents. In particular, we present
rely conditions in \refsec{sec:rely-conditions}, and decomposition
\refsec{sec:decomp-using-rely} over parallel composition using rely
conditions. In \refsec{sec:transformation-rules}, we present a number
of high-level transformation rules specific to CAS-based
implementations.

\subsection{Rely conditions}
\label{sec:rely-conditions}

We introduce constructs for defining a rely condition, which specifies
assumptions about the behaviour of the environment \cite{Jon83}. We
note that unlike Jones \cite{Jon83}, who assumes rely conditions are
relations, we rely conditions are interval predicates, allowing
specification of properties over an interval.  % The framework itself is
The behaviour of a command with a rely condition is given by the
behaviour of the command in an interval in which the rely condition is
assumed to hold. That is, the behaviours of the environment
\emph{overlap} \cite{DDH12} with those of the program as opposed to
\emph{interleave} \cite{Jon83,STER11} with the program.

\begin{definition}
  For an interval predicate $r$ and command $C$, we let
  $(\Rely{r}{C})$ denote a command with a \emph{rely condition} $r$,
  whose behaviour for any set of processes $P$ and set of locations $Z
  \subseteq Var \cup Addr$ is given by
  \begin{eqnarray*}
    beh_{P,Z}.(\Rely{r}{C})  &\ \  \sdef\ \  &  r \imp beh_{P,Z}.C 
  \end{eqnarray*}
\end{definition}
Hence, $(\rely r \st C)$ consists of an execution of $C$ under the
assumption that $r$ holds. Note that if $\neg r$ holds, then the
behaviour of $(\rely r \st C)$ is chaotic, i.e., any behaviour is
allowed. % One may refine a rely

This interpretation of rely/guarantee has been shown to be effective
for reasoning in a real-time setting, where conflicting updates by a
program and its environment are avoided by ensuring the environment
variables are disjoint from the program variables
\cite{DH12MPC,DH12}. In this paper, a program and its environment may
share a common set of locations, hence we use fractional permission to
ensure conflicting accesses do not occur. % Such an interpretation may

\begin{example}
  \label{ex:using-nonint}
  Suppose we want show that for an assignment $x \asgn x+1$, the final
  value of $x$ is one greater than its initial value provided that the
  environment does not modify $x$. We have
  \begin{derivation}
    \step{beh_{p,Z}.(\Rely{\Always \neg \mcI.x.p}{x \asgn x + 1})}
    
    \trans{\equiv}{definition of $beh$}
    
    \step{\Always \neg \mcI.x.p \imp \exists k \st beh_{p,Z}.[x + 1 =
      k] \ch \Update_{p,Z} (x = k)}

    \trans{\entails}{$\Always c$ splits}

    \step{\Always\neg \mcI.x.p \imp \exists k \st (\Always \neg \mcI.x.p \imp beh_{p,Z}.[x + 1 =
      k]) \ch \Update_{p,Z} (x = k)}

    \trans{\entails}{using $\Always \neg \mcI.x.p$ and by definition
      $beh_{p,Z}.[x + 1 = k] \entails \Always \neg \mcW.x.p$}
    
    \step{\Always\neg \mcI.x.p \imp \exists k \st \Always (x + 1 = k) \land \neg \Empty \ch
      \Always (x = k)\land \neg \Empty}

    \trans{\entails}{$\Always c \entails \ola{c}$}
    
    \step{\Always\neg \mcI.x.p \imp \exists k \st \ola {x + 1 = k} \land \neg \Empty \land
      \ora{x = k}}  

    \trans{\entails}{logic}
    
    \step{\Always\neg \mcI.x.p \imp \exists k \st \ola {x = k} \land \neg \Empty\land \ora{x = k
      + 1}\hfill 
      {}_\clubsuit}
  
  \end{derivation}
\end{example}

\subsection{Decomposition using rely conditions}
\label{sec:decomp-using-rely}

One may develop a number of rules for refining commands with rely
conditions. 
\begin{lemma} Each of the following holds.  \label{lem:refine-rely}
  \begin{eqnarray}
    \label{eq:17}
    \Rely{r}{C} & \!\sref_P^{Y,Z}\! & C
    \\
    \label{eq:38}
    r \entails r' %  \land \Rely{r}{A} \sref_P C
    &
    \!\imp\! &  \Rely{r}{C} 
    \sref_P^{Y,Z} \Rely{r'}{C} \\
    \label{eq:3}
    \!\!\!\!\!\!\!\!\begin{array}[c]{@{}l@{}}
      r \land beh_P.C \entails 
      \hfill beh_P.A
    \end{array}
    & \!\imp\! &
    \Rely{r}{A} 
    \sref_P^{Y,Z} \Rely{r}{C} 
    \\
    \label{eq:5}
    r \land beh_P.C \entails d
    & \!\imp\! & 
    \begin{array}[c]{@{}l@{}}
      \Rely{r}{\enf{d}\st{C}}
      \sref_P^{Y,Z} \Rely{r}{C}
    \end{array}
  \end{eqnarray}
\end{lemma}
Rule \refeq{eq:17} allows a rely condition to be removed,
\refeq{eq:38} allows a rely condition to be weakened and by
\refeq{eq:3}, the refinement holds for the rely condition $r$ on both
sides if the behaviour of $C$ implies the behaviour of $A$ under rely
condition $r$.  By \refeq{eq:5}, we may remove the enforced property
$d$ if the rely condition and behaviour of $C$ together imply $d$. Of
course, it may be the case that only the rely condition without the
program or the program without the rely condition is enough to
establish the enforced property. Both these cases are covered by
\refeq{eq:5}.

The lemma below allows one to distribute a rely condition
in and out of a sequential composition and an iterated command. The
lemma requires that the given rely condition splits.
\begin{lemma}
  \label{lem:rely-sc-ref}
  Suppose $r$ is an interval predicate, $P$ a non-empty set of
  processes and $Y$, $Z$ are sets of variables. If $r$ splits, then
  \begin{eqnarray}
    \label{eq:27}
    (\Rely{r}{C_1 \ch C_2}) & \sref_P^{Y,Z} & (\Rely{r}{C_1}) \ch
    (\Rely{r}{C_2}) 
    \\
    \label{eq:29}
    \Rely{r}{C^\omega} & \sref_P^{Y,Z} & (\Rely{r}{C})^\omega
    \\
    \label{eq:28}
    (\Rely{r}{A  \sref_P^{Y,Z} C}) & \imp & 
    (\Rely{r}{A^\omega} \sref_P^{Y,Z} C^\omega)
  \end{eqnarray}
\end{lemma}

The following theorem shows that rely and enforced conditions form a
Galois connection. Namely, command $C$ refines a command $A$ under
rely condition $r$ if and only if the command $C$ with enforced
condition $r$ refines $A$.
\begin{theorem}
  \label{thm:rely-enf-gc}
  $(\Rely{r}{A}) \sref_P^{Y,Z} C \quad  =  \quad A  \sref_P^{Y,Z} (\enf{r}\st{C})$
\end{theorem}
\begin{proof}
  \begin{derivation}
    \step{(\Rely{r}{A}) \sref_P^{Y,Z} C}
    
    \trans{=}{definitions and logic}

    \step{r \land beh_{P,Z}.C \entails beh_{P,Y}.A}

    \trans{=}{definitions}

    \step{A \sref_P^{Y,Z} (\enf{r}\st{C}) \hfill \qedhere}
  \end{derivation}
\end{proof}
To see the usefulness of this theorem, consider the lemma below that
allows refinement over non-deterministic choice in the presence of a
rely condition. Using \refthm{thm:rely-enf-gc}, and the result of
\reflem{lem:rely-nd-ref} (below), we obtain a dual result
\refeq{eq:35} below on an enforced property.
\begin{lemma} 
  \label{lem:rely-nd-ref}
  $\begin{array}[c]{@{}r@{}}
    ((\Rely{r}{A_1}) \sref_P^{Y,Z} C_1) \land 
    ((\Rely{r}{A_2}) \sref_P^{Y,Z} C_2)
    \end{array}
     \imp
    \begin{array}[c]{@{}l@{}}
      (\Rely{r}{A_1 \sqcap A_2}) \sref_P^{Y,Z} C_1 \sqcap C_2
    \end{array}$
  \end{lemma}
Using \refthm{thm:rely-enf-gc}, the results \refeq{eq:34} and
\reflem{lem:rely-nd-ref} may both immediately be converted into a
property for enforced conditions:
\begin{eqnarray}
  \label{eq:35}
  \begin{array}[c]{@{}r@{}}
    (A_1 \sref_P^{Y,Z} (\enf{r}\st{C_1})) \land \\
    (A_2 \sref_P^{Y,Z} (\enf{r}\st{C_2}))
  \end{array}
  &\  \ \imp\ \  &
  \begin{array}[c]{@{}l@{}}
    A_1 \sqcap A_2 \sref_P^{Y,Z} (\enf{r}\st{C_1 \sqcap C_2})
  \end{array}  
\end{eqnarray}

One can also use \refthm{thm:rely-enf-gc} and \reflem{lem:rely-sc-ref}
to obtain the following dual property, where we assume $g$ is an
interval predicate that splits.
\begin{eqnarray}
  \label{eq:36}
  (\enf{g}\st{C_1}) \ch
  (\enf{g}\st{C_2})  & \sref_P^{Y,Z} & ( \enf{g}\st{C_1 \ch C_2})
\end{eqnarray}
The following theorem may be used to decompose a proof that an
parallel composition of an abstract program is refined by the parallel
composition of a concrete program in the context of an overall rely
condition $r$.
\begin{theorem} 
  \label{thm:decompose} $(\Rely{r}{\Par_{p:P}A_p}) \sref_P^{Y,Z}
  (\Par_{p:P}C_p)$ holds if there exist $P_1, P_2 \subseteq
  P$ such that $P = P_1 \cup P_2$ and $P_1
  \cap P_2 = \emptyset$ and both of the following hold for some
  interval predicates $r_1$ and $r_2$.
  \begin{eqnarray}
    \label{eq:15}
    (\Rely{r \land r_1}{\Par_{p:P_1}A_p}) & \sref_{P_1}^{Y,Z} & (\Par_{p:P_1}
    C_p)
    \\
    \label{eq:25}
    (\Rely{r \land r_2}{\Par_{p:P_2}A_p}) & \sref_{P_2}^{Y,Z} & (\Par_{p:P_2}C_p)
    \\
    \label{eq:26}
    r \land beh_{P_2,Z}.(\Par_{p:P_2}C_p) & \entails &
    r_1
    \\
    r \land beh_{P_1,Z}.(\Par_{p:P_1}C_p)
    & \entails &  r_2\label{eq:14}
  \end{eqnarray}
\end{theorem}
\begin{proof}
  \begin{derivation*}
    \step{\refeq{eq:15} \land \refeq{eq:25}}

    \trans{=}{definition of $\sref_P^{Y,Z}$, logic $(a \entails (b
      \imp c)) = (a \land b \entails c)$}
    % \trans{}{lift $(r).P_i$ to $r$}

    \step{(r \land r_1 \land
      beh_{P_1,Z}.(\Par_{p:P_1}C_p) \entails beh_{P_1,Y}.(\Par_{p:P_1}
      A_p)) \land}
    \step{
      (r \land r_2 \land
      beh_{P_2,Z}.(\Par_{p:P_2}C_p) \entails beh_{P_2,Y}.(\Par_{p:P_2} A_p))}

    \trans{\imp}{logic} 

    \step{r \land r_1 \land beh_{P_1,Z}.(\Par_{p:P_1}C_p)
      \land r_2 \land beh_{P_2,Z}.(\Par_{p:P_2}C_p) \entails beh_{P_1,Y}.(\Par_{p:P_1} A_p) \land
      beh_{P_2,Y}.(\Par_{p:P_2} A_p)}

    \trans{\imp}{\refeq{eq:26} and \refeq{eq:14}}

    \step{r \land beh_{P_1,Z}.(\Par_{p:P_1}C_p) \land
      beh_{P_2,Z}.(\Par_{p:P_2}C_p) \entails beh_{P_1,Y}.(\Par_{p:P_1} A_p) \land
      beh_{P_2,Y}.(\Par_{p:P_2} A_p) }
    \step{}
  \end{derivation*}

  \noindent Hence we have the following calculation: 
  \begin{derivation}
    \step{\exists P_1, P_2  \st P_1 \cup
      P_2 = P 
      \land (P_1 \cap P_2 = \emptyset) 
      \land \refeq{eq:15} \land  \refeq{eq:25}}
    
    \trans{\imp}{calculation above and logic}

    \step{\exists P_1, P_2
       \st
         (P_1 \cup P_2 = P) \land
         (P_1 \cap P_2 = \emptyset)
         \land  
         r \land  beh_{P_1,Z}.(\Par_{p:P_1}C_p) \land
         beh_{P_2,Z}.(\Par_{p:P_2}C_p)
       \entails}
     \step{\hfill 
         beh_{P_1,Y}.(\Par_{p:P_1} A_p) \land 
         beh_{P_2,Y}.(\Par_{p:P_2} A_p)
     }

    \trans{=}{definitions}

    \step{(\Rely{r}{\Par_{p:P} A_p})
      \sref_{P} (\Par_{p:P} C_p) \hfill\qedhere}
  \end{derivation}
\end{proof}

When modelling a lock-free program \cite{CGLM06,DSW11,VHHS06}, one
assumes that each process repeatedly executes operations of the data
structure, and hence the processes of the system only differ in terms
of the process ids. For such programs, a proof of the parallel
composition may be simplified to as described by the theorem below.
\begin{theorem}
  \label{thm:parametrise}
  Suppose $p \in Proc$ and $A % \Context{W}{A(p)^\omega}
  $ and $C% \Context{X}{C(p)^\omega}
  $ are commands with input parameter $p$ such that $W \subseteq Y$
  and $X \subseteq Z$. Then $ (\Rely{r}{\Par_{p:P}A(p)}) \sref_P^{Y,Z}
  (\Par_{p:P}C(p))$ holds if the following holds for some
  interval predicate $r_1$ and some $p \in P$ where $Q \sdef P \bs
  \{p\}$.
  \begin{eqnarray}
    \label{eq:70}
    (\Rely{r \land r_1}{A(p)}  \sref_{p}^{Y,Z} 
    C(p))% \Rely{r \land r_1}{A(p)} & \sref_{p}^{Y,Z} & C(p)
    & \quad \land \quad 
    (r \land beh_{Q}.(\Par_{q:Q}C(q))  \entails 
    r_1)
  \end{eqnarray}
\end{theorem}
\begin{proof}
  Applying \refthm{thm:decompose} and choosing $P_1 = \{p\}$,
  and $P_2 = Q$, the proof of $$(\Rely{r}{\Par_{p:P} A(p)})
  \sref_P^{Y,Z} (\Par_{p:P} C(p))$$ decomposes as follows for some
  interval predicates $r_1$ and $r_2$.
  \begin{eqnarray}
    \label{eq:12-a}
    (\Rely{r \land r_1}{A(p)}  \sref_{p}^{Y,Z} 
    C(p))
    & \land & 
    (\Rely{r \land r_2}{\Par_{p:Q}A(p)}  \sref_{Q}^{Y,Z}  \Par_{p:Q}C(p))
    \\
    \label{eq:19-a}
    (r \land beh_{Q}.(\Par_{p:Q}C(p))  \entails 
    r_1)
    & \land & 
    (r \land beh_{p}.C(p) \entails r_2)
  \end{eqnarray}
  The first conjuncts of conditions \refeq{eq:12-a} and
  \refeq{eq:19-a} hold by assumption \refeq{eq:70}.  Choosing $r_2 =
  true$, the proof of the second conjunct of \refeq{eq:19-a} is trivial
  and the proof of the second conjunct of \refeq{eq:12-a} follows by
  induction on the size of the set of processes $Q$. In particular, we
  use the fact that the behaviour of $\Par_{p: \{\}} A(p)$ is $true$
  as the base case of the induction. \hfill\qedhere
\end{proof}

\subsection{Transformation rules}
\label{sec:transformation-rules}
In this section we present a number of additional refinement rules
that allow one to transform coarse-grained code into code with finer
granularity. These proofs are greatly simplified by the fact that we
consider interval-based behaviour, which includes consideration of the
possible interference from other processes. The theorems we present
are developed around the Treiber Stack example. We anticipate that
several more can be developed when considering other examples. We
further conjecture that such theorems can also be used to derive
concurrent data structures via a series of refinements.

The theorem below allows refinement of an enforced property $\Always
(e = va)$ to a command $C$ within $\rely r \st va \asgn e \ch C \ch
\Guard{e = va}$ provided that $C$ does not write to $va$ and the rely
condition $r$ ensures that the $ABA$ problem does not occur on $e$.
\begin{theorem}
  \label{thm:ABA}
  Suppose $p$ is a process, $Z \subseteq Var \cup Addr$, $v \in Z
  \cap Var$, $e$ is an expression, $C$ is a command and $r$ is an
  interval predicate. If both of the following hold
  \begin{eqnarray}
    \label{eq:66}
    beh_{p,Z}.C & \entails &  \Always \neg \mcW.v.p 
    \\
    \label{eq:65}
    r & \imp & 
    \Box \neg ABA.e \land (\all q : Proc \bs \{p\} \st \neg \mcW.v.q)
  \end{eqnarray}
  then
  \begin{eqnarray}
    \label{eq:7}
    \begin{array}[c]{@{}l@{}}
      \rely r \st 
      \begin{array}[t]{@{}l@{}}
        v \asgn e \ch 
        (\enf{\Always (e = v) }\st{C}) \ch 
        \Guard{e = v}
      \end{array}
    \end{array} & \sref_p^{Z} &
    \begin{array}[c]{@{}l@{}}
      v \asgn e \ch  
      C  \ch 
      \Guard{e = v}
    \end{array}
    % \\
    % \label{eq:21}
    % \!\!\!\!\begin{array}[c]{@{}l@{}}
    %     \rely \neg ABA.v \st  \\
    %      a \asgn v \ch (\Enf{\Always (v = a)}{C}) \ch
    %     {[}v = a{]}    
    %   \end{array}
    % & \sref_p &
    % \begin{array}[c]{@{}l@{}}
    %   a \asgn v \ch \\
    %   C \ch [v = a]
    % \end{array}
  \end{eqnarray}
\end{theorem}
\begin{proof}
  Condition \refeq{eq:7} is equivalent to
  \begin{eqnarray*}
    r \land beh_{p,Z}.(v \asgn e \ch  
    C  \ch 
    \Guard{e = v})
    & \entails & 
    beh_{p,Z}.(v \asgn e \ch 
    (\enf{\Always (e = v) }\st{C}) \ch 
    \Guard{e = v})
  \end{eqnarray*}
  By logic, the antecedent of the formula above is equivalent to
  \[
  r \land beh_{p,Z}.(v \asgn e \ch  
  (\Enf{\Always (e = v) \lor \Eventually (e \neq v)}{C})  \ch 
  \Guard{e = v})
  \]
  The $\Always(e = v)$ case is trivial. For the $\Eventually (e \neq
  v)$ case, we have 
  \begin{derivation}
    \step{r \land beh_{p,Z}.(v \asgn e \ch 
      (\enf{\Eventually (e \neq v) }\st {C}) \ch 
      \Guard{e = v})}

    % \trans{\entails}{$r \imp \Box(\Eventually \mcI_p.e \imp
    %   \Eventually(e \neq v))$}

    % \step{r \land beh_{p,Z}.(v \asgn e \ch 
    %   (\Enf{\Eventually (e \neq v) }{C}) \ch 
    %   \Guard{e = v})}

    \trans{\entails}{\refeq{eq:66} and \refeq{eq:65} together implies $\all q :
      P \st \Always \neg \mcW.v.q$}

    \step{r \land beh_{p,Z}.(v \asgn e \ch 
      (\Enf{\Eventually (e \neq v) \land \stable.v}{C}) \ch  {  \Guard{e = v}})}

    \trans{\entails}{\refeq{eq:65}, $beh_{p,Z}. \Guard{e = v}
      \entails \Always \neg \mcW.v.p$ and {\bf HC1}}

    \step{r \land beh_{p,Z}.(v \asgn e \ch 
      (\Enf{\Eventually (e \neq v) \land \stable.v}{C}) \ch 
      (\enf{\stable.v}\st{  \Guard{e = v}}))}

    \trans{\entails}{$beh_{p,Z}$ definition}

    \step{r \land \exists k \st \Eventually (e = k) \ch (\Always (v = k) \land \neg \Empty) \ch 
      ({\Eventually (e \neq v) \land \stable.v}) \ch 
      (\stable.v \land \Eventually (e = v))}

    \trans{\entails}{using stability of $v$}

    \step{r \land \exists k \st \Eventually (e = k) \ch (\Always (v = k) \land \neg \Empty) \ch 
      {\Eventually (e \neq k)} \ch 
      \Eventually (e = k)}

    \trans{\entails}{using $r \imp \neg ABA.e$}

    \step{false \hfill \qedhere }

  \end{derivation}
\end{proof}

The theorem below allows one to replace an assignment to an expression
by an equivalent assignment provided that the process under
consideration has the necessary permissions. 
\begin{theorem}[Replace assignment]
  Suppose $e$ and $e'$ are expressions, $v \in Var$, $Y \subseteq Z
  \subseteq (Var \cup Addr)$ are sets of locations and $p$ is a
  process. If $IntFree.e.p \land beh_{p,Z}.\Idle \entails
  ReadAllLocs.e.p$, then
  \label{thm:expr-eq}
  \begin{eqnarray*}
    \Rely{\Always (e = e') \land IntFree.e.p 
    }{v \asgn e} & \sref_p^{Y,Z} & v \asgn e'  
  \end{eqnarray*}
\end{theorem}
\begin{proof}
  The proof holds if $ \Always (e = e') \land IntFree.e.p
  \land beh_{p,Z}.(v \asgn e') \entails beh_{p,Y}.(v \asgn e)$
  
  \begin{derivation}
    \step{\Always (e = e') \land IntFree.e.p \land
      beh_{p,Z}.(v \asgn e')}  
    
    \trans{\equiv}{definition of $beh_{p,Z}$}
    
    \step{\Always (e = e') \land IntFree.e.p \land 
      \exists k \st \Eval_{p,Z}.(e', k) \ch \Update_{p,Z}(v,k) }
      % {\sf update\_store}(v, k)}
    
    \trans{\equiv}{definition of $\Eval$}
    
    \step{\Always (e = e') \land IntFree.e.p \land \exists
      k \st (\Eventually (e' = k \land ReadAllLocs.e'.p) \land
      beh_{p,Z}.\Idle )  \ch \Update_{p,Z}(v,k)}% {\sf 
        % update\_store}(v, k)} 
    
    \trans{\entails}{using $\Always (e = e')$ and assumption
      $IntFree.e \land beh_{p,Z}.\Idle \entails ReadAllLocs.e.p$}
    
    \step{\exists
      k \st (\Eventually (e = k \land ReadAllLocs.e.p) \land 
      beh_{p,Z}.\Idle)  \ch \Update_{p,Z}(v,k)}% {\sf 
        % update\_store}(v, k)} 
    
    \trans{\entails}{$Y \subseteq Z$, definition of $beh_{p,Y}$}
    
    \step{beh_{p,Y}.(v \asgn e) \hfill \qedhere} 
    
  \end{derivation}
\end{proof}

The following theorem allows one to split a guard evaluation so that
the variable being tested is stored locally. The theorem allows one to
perform some additional behaviour that does not affect the abstract
state.
\begin{theorem}[Introduce command]
  \label{thm:grd-intro}
  Suppose $e$ is an expression, $v$ is a variable, $ae$ is an
  address-valued expression, $p$ is a process, $Y, Z \subseteq Var$
  such that both $Y \subseteq Z$ and $v \in Y$, and $C$ is a
  command. Then both of the following hold:
  \begin{align}
    \label{eq:61}
    v \asgn e & \sref_p^{Y,Z} C \ch v \asgn e & &
    \textrm{provided $beh_{p,Z}.C \entails \idle_{p,Y}$}
    \\
    \label{eq:62}
    ae \hasgn e & \sref_p^{Y,Z} C \ch ae \hasgn e & &\textrm{provided
      $beh_{p,Z}.C \entails \idle_{p,Y} \land \neg WriteSomeLoc.ae.p$}
  \end{align}
\end{theorem}
\begin{proof} We prove \refeq{eq:61} as follows. The proof of
  \refeq{eq:62} follows the same structure.  
  \begin{derivation}
    \step{}
    \step{beh_{p,Z}.(C \ch v \asgn e)}

    \trans{\equiv}{definition of $beh_{p,Z}$}

    \step{beh_{p,Z}.C
      \ch 
          \exists k  \st
            \Eval_{p,Z}(e,k)
            \ch \Update_{p,Z}(v,k)%{} {\sf update\_store}_{p,Z}(v,k)) % \\ 
      }

    \trans{\entails}{assumption on $beh_{p,Z}.C$}

    \step{ \idle_{p,Y} \ch \exists k  \st
      \begin{array}[t]{@{}l@{}}
        \Eval_{p,Z}(e,k) \ch \Update_{p,Z}(v,k) 
        % {\sf update\_store}_{p,Z}(v,k)
      \end{array}
    }

    \trans{\entails}{logic, $Y \subseteq Z$ and $v \in Y$}
     
    \step{\exists k  \st 
        \begin{array}[t]{@{}l@{}}
          \idle_{p,Y} \ch  
          \Eval_{p,Y}(e,k) \ch \Update_{p,Y}(v,k) 
        \end{array}
      }

    \trans{\entails}{\refeq{eq:widen-1} and definition of $beh_{p,Y}$}
    
    \step{beh_{p,Y}.(v \asgn e) \hfill \qedhere}
  \end{derivation}
\end{proof}

Guard and expression evaluation occur in two different (observable)
states, and hence, for example, $[a = 42] \ch v \asgn a$ does not
guarantee that $v$ has a final value $42$ because the value of $a$ may
change after the guard evaluation.  The following theorem shows that
testing a guard can be split over multiple steps by introducing a new
local variable to the program.

\begin{theorem}[Split guard]
  \label{thm:refine-guard}
  Suppose % $va \in LVar_p \bs Y$, 
  $va, v \in Z$, $k$ is a
  constant and $beh_{p,Z}.C \imp \idle_{Y \cup
    \{va\}} % \Always \neg
  $ then
  \begin{eqnarray*}
    [v = k] & \sref_p^{Y,Z} & va \asgn v \ch \enf IntFree.va.p \st (C
    \ch [va = k])
  \end{eqnarray*}
\end{theorem}
\begin{proof} 
  We first prove the following: 
  \begin{derivation}
    \step{IntFree.va.p \land  (beh_{p,Z}.C \ch \Eval_{p,Z}(va, k))}

    \trans{\entails}{assumptions $beh_{p,Z}.C \entails \idle_{Y \cup
        \{va\}}$ and $Y \subseteq Z$}

    \step{\stable.va \land (\idle_{p,Y} \ch \Eval_{p,Y}(va, k))} 

    \trans{\entails}{\refeq{eq:widen-1}}

    \step{\stable.va \land \Eval_{p,Y}(va, k)}

    \trans{\entails}{use $\stable.va$}

    \step{\Always (va = k) \land \idle_{p,Y}}  

  \end{derivation}

  \noindent
  Then, we have the following calculation
  \begin{derivation}
    \step{beh_{p,Z}.(va \asgn v \ch \enf IntFree.va.p \st (C \ch [va = k]))}

    \trans{\entails}{expand definition}

    \step{(\exists j \st \Eval_{p,Z}(v,j) \ch \Update_{p,Z}(va,j))
      \ch (IntFree.va.p \land (beh_{p,Z}.C \ch \Eval_{p,Z}(va, k)))} 

    \trans{\entails}{calculation above}

    \step{(\exists j \st \Eval_{p,Z}(v,j) \ch \Update_{p,Z}(va,j)
      ) \ch \Always (va = k) \land \idle_{p,Y}}

    \trans{\entails}{logic, assume $j$ fresh% , then use $\stable.va$
    }

    \step{\exists j \st \Eval_{p,Z}(v,j) \ch \Update_{p,Z}(va,j)
       \ch \Always (va = k) \land \idle_{p,Y}}

    \trans{\entails}{case analysis on $j = k$, case $j \neq k$ yields
      a contradiction}

    \step{\Eval_{p,Z}(v,k) \ch % {\sf update\_store}_{p,Y}(a,k)
      \Update_{p,Z}(va,k)
      \ch  \idle_{p,Y}}   

    \trans{\entails}{$Y \subseteq Z$, definition of $\Update$ using
      assumption $va \notin Y$} 

    \step{\Eval_{p,Y}(v,k) \ch \idle_{p,Y} \ch \idle_{p,Y}}   

    \trans{\entails}{\refeq{eq:widen-2} twice and definition of $beh$}

    \step{beh_{p,Y}.[v = k] \hfill\qedhere}

  \end{derivation}
\end{proof}
A theorem such as \refthm{thm:refine-guard} is more difficult to
establish in a model that only considers pre/post states because the
guard evaluation on the left of $\sref_p^{Y,Z}$ is over a single state
and the command on the right is over multiple
states. % \refthm{thm:refine-guard} can be generalised as follows, where

The following theorem allows an assignment to be introduced provided
that the new variable is distinct from the abstract context.
\begin{theorem} 
  \label{thm:intro-asgn}
  Suppose $Y, Z \subseteq Var$ such that $Y \subseteq Z$, $v \in Z \bs
  Y$ and $e$ is an expression and $p \in Proc$. Then
  \begin{eqnarray*}
    \enf{\neg \Empty}\st{\Idle} & \sref_p^{Y,Z} & v \asgn e
  \end{eqnarray*}
\end{theorem}
\begin{proof}
  The refinement holds because $v \notin Y$, $Y \subseteq Z$ and $v
  \asgn e$ ensures $\idle_{Z \bs \{v\}}$. \hfill \qedhere
\end{proof}

We also develop a transformation theorem for executing a successful
CAS operation. % The theorem relies on the following lemma, which allows
\begin{theorem}[Introduce CAS]
  \label{thm:replace-CAS}
  Suppose $ae$ is an address-valued expression, $e$ is an expression,
  $\alpha \in Z$ and $r$ is an interval predicate such that
  \begin{eqnarray}\
    \label{eq:68}
    r &\ \   \entails\ \  &
    \Always (accessed.(\derefe ae)\subseteq Z) \land
    IntFree.\alpha \land % \land
    \\
    \label{eq:71}
    & & \Always ((\derefe ae = \alpha) \imp (\beta = ae \mapsto f)) \\
    \label{eq:79}
    & & \Always(ReadAllLocs.\beta.\alpha \imp ReadAllLocs.e)
  \end{eqnarray}
  Then
  \begin{eqnarray*}
    \Rely{r}{ae \hasgn e} &
    \sref_p^{Z} & 
    CASOK_p(ae, \alpha, \beta)
  \end{eqnarray*} 
\end{theorem}
\begin{proof} 
  Using \refthm{thm:rely-enf-gc}, we may equivalently
  prove $ae \hasgn e \sref_p^{Z} \enf{r}\st{CASOK_p(ae,
    \alpha, \beta)}$. We have the following calculation.
  \begin{derivation}
    \step{beh_{p,Z}.(\enf{r}\st {CASOK_p(ae, \alpha, \beta)})}
  
    \trans{\equiv}{expand definitions}

    \step{r \land IntFree.ae.p \land \left(
        \begin{array}[c]{@{}l@{}}
          beh_{p,Z}.[\derefe ae = \alpha]
          \ch          
          beh_{p,Z}.(ae \asgn \beta)  
        \end{array}\right)
    } 

    \trans{\equiv}{expand definitions}

    \step{r \land OnlyAccessedBy.(\derefe ae).p \land }
    \step{\left(
        \begin{array}[c]{@{}l@{}}
          \Eval_{p,Z}.(\derefe ae = \alpha, true)
          \ch          
          \exists a, k \st  
          (\Eval_{p,Z}(ae, a) \land \Eval_{p,Z}(\beta, k))\ch
          \Update_{p,Z}(a,k)% {\sf 
            % update\_heap}_{p,Z}(a, k)  
        \end{array}\right)
    } 

    \trans{\entails}{logic, expand definitions}

    \step{r \land OnlyAccessedBy.(\derefe ae).p \land }
    \step{\exists a, k \st  
        \begin{array}[c]{@{}l@{}}
          (\Eventually (\derefe ae = \alpha)
          \land \idle_{p,Z})
          \ch          
          (\Eval_{p,Z}(ae, a) \land \Eval_{p,Z}(\beta, k))\ch
          \Update_{p,Z}(a,k) % {\sf
            % update\_heap}_{p,Z}(a, k)  
        \end{array}
    } 

    \trans{\entails}{\refeq{eq:widen-1}}

    \step{r \land  OnlyAccessedBy.(\derefe ae).p \land}
    \step{\exists a, k \st  
        \begin{array}[c]{@{}l@{}}
          (\Eventually (\derefe ae = \alpha)
          \land          
          \Eval_{p,Z}(ae, a) \land \Eval_{p,Z}(\beta, k)) \ch
          \Update_{p,Z}(a,k) % {\sf
            % update\_heap}_{p,Z}(a, k)  
        \end{array}
    }

    % \trans{\equiv}{% definition of $accessed$, $\alpha \in Z$ and
    %   $\Eval_{p,Z}(ae, a) \entails beh_{p,Z}.\Idle$}

    % \step{r \land IntFree.(\derefe ae = \alpha).p \land }
    % \step{\exists a, k \st
    %     \begin{array}[c]{@{}l@{}}
    %       (\Eventually (\derefe ae = \alpha)
    %       \land 
    %       \Eval_{p,Z}(ae, a) \land \Eval_{p,Z}(\beta, k) \land
    %       \neg WriteSomeLoc.(\derefe ae = \alpha).p) \ch \Update_{p,Z}(a,k)% {\sf
    %         % update\_heap}_{p,Z}(a, k)  
    %     \end{array}
    % } 

    \trans{\entails}{\refeq{eq:68}, $\alpha \in Z$ and
      $OnlyAccessedBy.(\derefe ae).p$ implies all locations in
      $\derefe ae = \alpha $ are stable}

    % \step{r \land IntFree.(\derefe ae = \alpha).p \land }
    \step{r \land \exists a, k \st  
        \begin{array}[c]{@{}l@{}}
          (\Always (\derefe ae = \alpha)
          \land 
          \Eval_{p,Z}(ae, a) \land \Eval_{p,Z}(\beta, k))\ch
          \Update_{p,Z}(a,k)% {\sf 
            % update\_heap}_{p,Z}(a, k)  
        \end{array}
    } 

    \trans{\entails}{\refeq{eq:71}}

    \step{r \land \exists a, k \st  
      \begin{array}[c]{@{}l@{}}
        (\Always (\beta = e)  \land
        \Eval_{p,Z}(ae, a) \land \Eval_{p,Z}(\beta, k))\ch 
          \Update_{p,Z}(a,k) 
        \end{array}
    } 

    \trans{\entails}{$\Always (\beta = e)$ and \refeq{eq:79}}

    \step{r \land \exists a, k \st  
        \begin{array}[c]{@{}l@{}}
          (\Eval_{p,Z}(ae, a) \land \Eval_{p,Z}(e, k))\ch
          \Update_{p,Z}(a,k) 
        \end{array}
    } 

    \trans{\entails}{definitions}

    \step{beh_{p,Z}.(ae \hasgn e) \hfill \qedhere}
  \end{derivation}
\end{proof}

\section{Behaviour refinement the Treiber Stack}
\label{sec:verification}

In this section, we verify that the Treiber Stack (modelled by
$TS(P)$) refines the coarse-grained abstract program (modelled by
$LS(P)$), i.e., we prove: 
\begin{eqnarray}
  \label{eq:19}
  LS(P) & \sref_{P} & TS(P)  
\end{eqnarray}
Our proof strategy decomposes the parallel composition
(\refsec{sec:decomp-parall-comp}), which allows us to consider the
push and pop operations executed by a single process separately
(\refsec{sec:proof-push-operation} and
\refsec{sec:proof-pop-operation}). We derive the necessary rely
conditions for the push and pop operations as part of the proofs in
Sections \ref{sec:proof-push-operation} and
\ref{sec:proof-pop-operation}, which are then discharged in
\refsec{sec:proof-that-rest}.

\subsection{Decompose parallel composition}
\label{sec:decomp-parall-comp}

We first apply \reflem{lem:fr-ref} and then \refeq{eq:59} of
\reflem{lem:command-ref} to reduce the refinement to the following
proof obligation:
\begin{eqnarray}
  \label{eq:80}
  \|_{p:P} LP(p) & \sref_{p}^{TF} &  \|_{p:P} TP(p)
\end{eqnarray}
where 
\begin{eqnarray*}
  TF & \sdef & \{Top, FAddr\} \\
  LP(p) & \sdef &  \Context{n_p,rv_p}{LPP(p)^\omega} \\
  TP(p) & \sdef &  \Context{t_p,n_p, tn_p, rv_p}{TPP(p)^\omega}
\end{eqnarray*}
Using \refthm{thm:decompose}, we further decompose the parallel
composition in \refeq{eq:80} to obtain the following proof
obligations.
\begin{eqnarray}
  \label{eq:52}
  \rely r_1 \st LPP(p)^\omega  & \sref_{p}^{TF} & PP(p)^\omega
  \\
  \label{eq:64}
  beh_{P',TF}.(\Par_{P'} PP(p)^\omega)  &
  \entails &
  r_1  
\end{eqnarray}

Condition $r_1$ is yet to be determined, and is calculated as part of
the proof of \refeq{eq:52}. However, we require that the condition
$r_1$ that we derive splits to allow our transformation lemmas to be
applied. Recalling that the $L$ denotes the set of variables of the
coarse-grained abstraction $LS(P)$, we define
\begin{eqnarray}
  \label{eq:47}
  T & \sdef &  L \cup \{p: Proc \st tn_p, rv_p\}
\end{eqnarray}
to be set of variables of the coarse-grained abstraction $TS(P)$ and
obtain the follows. 
\begin{derivation}
  \step{\refeq{eq:52}}

  \trans{\follows}{\refeq{eq:29} decompose iteration assuming $r_1$
    splits}
  
  \step{\Rely{r_1}{LPP(p) \sref_{p}^{TF} PP(p)}}
  
  \trans{\follows}{\refeq{eq:27} of \reflem{lem:rely-sc-ref} assuming
    $r_1$ splits, then \reflem{lem:rely-nd-ref}}
  
  \step{\Rely{r_1}{LPush(p)} \sref_{p}^{L,T} Push(p) \land \hfill
    \refstepcounter{equation} \label{eq:A} (\arabic{equation})}
  \step{\Rely{r_1}{LPop(p)} \sref_{p}^{L,T} Pop(p) \hfill
    \refstepcounter{equation} \label{eq:B} (\arabic{equation})}
\end{derivation}

\subsection{Proof of push operation \refeq{eq:A}}
\label{sec:proof-push-operation}
We prove the $Push(p)$ operation as follows:
\begin{derivation}
  \step{\Rely{r_1}{LPush(p)} \sref_{p}^{L,T} Push(p)}

  \trans{\iff}{expand definitions}

  \step{\Rely{r_1}{LSetup(p,x) \ch
      EnvSt(p) \ch 
      LDoPush(p)} \sref_{p}^{L,T} \Code{Setup(p,x) \ch TryPush(p)^\omega
      \ch DoPush(p)}}
  
  \trans{\follows}{ \reflem{lem:rely-sc-ref} assuming $r_1$ splits}

  \step{\Rely{r_1}{EnvSt(p) \ch
      LDoPush(p)} \sref_{p}^{L,T} \Code{TryPush(p)^\omega \ch DoPush(p)}}
  
  \trans{\follows}{Lemmas \ref{lem:cmd-splits} and
    \ref{lem:cmd-joins}, using $EnvSt(p)$ both splits and joins }

  \step{\Rely{r_1}{EnvSt(p) \ch EnvSt(p) \ch 
      LDoPush(p)} \sref_{p}^{L,T}  \Code{TryPush(p)^\omega 
    \ch DoPush(p)}} 

  % \trans{\follows}{\reflem{lem:rely-sc-ref}, assumption $r_1$ splits} 
  % \trans{}{$\Rely{r}{A} \sref(p) A$ holds}

  % \step{\Rely{r_1}{{\sf EnvSt}^\omega \ch  {\sf finPush}} \sref_{p}
  %   {\sf TryCAS}^\omega \ch {\sf DoCAS} }
  
  % \trans{\follows}{\reflem{lem:rely-nd-ref}}

\end{derivation}
\noindent
Using \reflem{lem:rely-sc-ref} (i.e., monotonicity of `;') and the assumption that $r_1$ splits,
the final refinement above holds if both of the following hold. 
\begin{eqnarray}
  \label{eq:30}
  \Rely{r_1}{EnvSt(p)} & \sref_{p}^{L,T} & TryPush(p)^\omega
  \\
  \label{eq:20}
  \Rely{r_1}{
        EnvSt(p) \ch 
        LDoPush(p)
    } & \sref_{p}^{L,T} &
  DoPush(p)
\end{eqnarray}

\paragraph{Proof of \refeq{eq:30}.}

\begin{derivation}
  \step{\refeq{eq:30}}

  \trans{\iff}{Lemmas \ref{lem:cmd-joins} and \ref{lem:cmd-splits},
    as $EnvSt(p)$ both joins and splits}

  \step{\Rely{r_1}{EnvSt(p)^\omega} \sref_{p}^{L,T} TryPush(p)^\omega}

  \step{(\Rely{r_1}{EnvSt(p)})^\omega \sref_{p}^{L,T} TryPush(p)^\omega}

  \trans{\follows}{$\omega$ is monotonic}

  \step{\Rely{r_1}{EnvSt(p)} \sref_{p}^{L,T}
    TryPush(p)}
  
  \trans{\iff}{\refthm{thm:rely-enf-gc}}

  \step{EnvSt(p) \sref_{p}^{L,T}
    \enf{r_1}\st{TryPush(p)}}
\end{derivation}

Commands $h_3$ and $hf_5$ of $TryPush(p)$ trivially satisfy the
requirements on the write permissions within $EnvSt(p)$, and hence
satisfy $EnvSt(p)$. For command $h_4$, we must ensure that $n_p \cdota
nxt \notin SAddr$, otherwise $\neg WriteSomeLoc.SAddr$ may not
hold. Hence, we require that $r_1$ implies:
\begin{eqnarray}
  \Always(pc_p = h_4 & \imp &  (n_p \cdota nxt) \notin SAddr) 
  \label{eq:22}
\end{eqnarray}
Using assumption \refeq{eq:22}, the proof of \refeq{eq:30} is
completed.

\paragraph{Proof of \refeq{eq:20}.}

By \refthm{thm:rely-enf-gc}, we may turn any rely condition on the
left of $\sref_p^{L,T}$ into an enforced property on the
right. % Furthermore,
Hence, condition \refeq{eq:20} is equivalent to:
\begin{eqnarray}
  \label{eq:31}
  \begin{array}[c]{@{}l@{}}
      EnvSt(p) \ch 
    LDoPush(p)
  \end{array}
  & \sref_{p}^{L,T} &
  \enf{r_1}\st{DoPush(p)}
\end{eqnarray}
To prove \refeq{eq:31}, we first simplify the right hand side of
$\sref_p^{L,T}$. Because $Top$ includes a modification counter and
every update to $Top$ increments this counter, each new value of $Top$
is guaranteed to be different from all previous values and hence, the
following is trivially guaranteed:
\begin{eqnarray}
  \label{eq:24}
  \Box \neg ABA.Top
\end{eqnarray}
We write $C \sqsupseteq_P^{Y,Z} A$ for $A \sref_P^{Y,Z} C$ and perform
the following calculation.

\begin{derivation}
  \step{\enf{r_1}\st{DoPush(p)}}

  \trans{\sqsupseteq_p^{T}}{expand $DoPush(p)$, remove labels}

  \step{\enf r_1 \st t_p \asgn \derefe Top \ch n_p \cdota nxt \hasgn ptr.t_p \ch
    CASOK_p( Top, t_p,  
    (n_p, ctr.t_p + 1))}
  
  \trans{\sqsupseteq_p^{T}}{assumption \refeq{eq:24} and
    \refthm{thm:ABA}}
  
  \step{\enf r_1 \st
    \begin{array}[t]{@{}l@{}}
      t_p \asgn \derefe Top \ch {}\\
      (\Enf{\Always (\neg \mcI.Top.p \land t_p = \derefe Top)}{ 
        n_p\cdota nxt \hasgn ptr.t_p}) \ch 
      CASOK_p( Top, t_p, (n_p, ctr.t_p + 1))
    \end{array}
  }

  \trans{\sqsupseteq_p^{T}}{\refthm{thm:expr-eq} then \refeq{eq:10}
    of \reflem{lem:enf-prop-intro-ref}}

  \step{\enf r_1 \st
    \begin{array}[t]{@{}l@{}}
      t_p \asgn \derefe Top \ch  {} \\
      (\enf{\Always \neg \mcI.Top.p}\st{ 
        n_p\cdota nxt \hasgn ptr.(\derefe Top)}) \ch  
      CASOK_p( Top, t_p, (n_p, ctr.t_p + 1))
    \end{array}
  }

  \trans{\sqsupseteq_p^{T}}{assumption $r_1$ splits, 
    then \refeq{eq:36} and \reflem{lem:enf-prop-intro-ref} to remove
    enforced property}

  \step{(\enf r_1 \st t_p \asgn \derefe Top) \ch \hfill
    \refstepcounter{equation} \label{eq:A1} (\arabic{equation})}
  \step{(
      \begin{array}[c]{@{}l@{}}
        \enf r_1 \st 
      \begin{array}[t]{@{}l@{}}
         (
          \begin{array}[c]{@{}l@{}}
            \enf \Always \neg \mcI.Top.p \st  n_p\cdota nxt \hasgn ptr.(\derefe Top)
          \end{array}
        ) \ch 
        CASOK_p( Top, t_p, (n_p, ctr.t_p + 1))
      \end{array}
    \end{array}
  ) \hfill \refstepcounter{equation} \label{eq:A2} (\arabic{equation})} 
\end{derivation}

Hence, by \refeq{eq:41} of \reflem{lem:command-ref}, the proof of
\refeq{eq:20} holds if we prove both of the following:
\begin{eqnarray}
  \label{eq:42}
  EnvSt(p) & \sref_p^{L,T} & 
  \refeq{eq:A1}
  \\
  \label{eq:43}
  LDoPush(p)
  & \sref_p^{L,T} & \refeq{eq:A2}
\end{eqnarray}
The proof of \refeq{eq:42} is trivial. For \refeq{eq:43}, we
strengthen $r_1$ so that it implies that there is no interference on
the stack nodes if there is no interference on $Top$ and that there is
no interference on addresses $n_p \cdota key$ and $n_p \cdota nxt$,
i.e., we require that $r_1$ satisfies: 
\begin{eqnarray}
  \label{eq:81}
  \Box (\Always \neg \mcI.Top.p \imp  IntFree.SAddr.p) \land 
  IntFree.\{n_p \cdota key, n_p \cdota nxt\}.p
\end{eqnarray}
Thus, we have the following calculation. 

\begin{derivation}
  \step{\refeq{eq:A2}} 
  
  \trans{\srefeq_p^{T}}{expand definitions}

  \step{\enf r_1 \st
      \begin{array}[t]{@{}l@{}}
        (
        \enf \Always \neg \mcI.Top.p \st  n_p\cdota nxt \hasgn ptr.(\derefe Top)
        ) \ch {} \\
        (\enf{\Always \neg \mcI.Top.p }\st{[\derefe Top = t_p]
          \ch Top \asgn (n_p, ctr.t_p + 1)})
      \end{array}} 

    \trans{\sqsupseteq_p^{T}}{$\Always c$ joins for any state
      predicate $c$ and \reflem{lem:enf-sc-ref}}
  
  \step{\Enf{r_1 \land \Always \neg \mcI.Top.p}{
      n_p\cdota nxt \hasgn ptr.(\derefe Top)} \ch 
    [\derefe Top = t_p]\ch Top \hasgn (n_p, ctr.t_p + 1)
  } 
  
  \trans{\sqsupseteq_p^{T}}{\refeq{eq:62} of
    \refthm{thm:grd-intro} and assumption \refeq{eq:81}}
  
  \step{\enf IntFree.(SAddr \cup \{Top, n_p \cdota key, n_p \cdota
    nxt\}).p \st (n_p\cdota nxt) \hasgn ptr.(\derefe Top) \ch Top
    \hasgn (n_p, ctr.(\derefe Top) + 1) }

  \trans{\srefeq_p^{T}}{definition of $LDoPush(p)$}

  \step{LDoPush(p)}
\end{derivation}

\subsection{Proof of pop operation \refeq{eq:B}}
\label{sec:proof-pop-operation}
We may decompose this operation as follows:

\begin{derivation}
  \step{\Rely{r_1}{LPop(p) \sref_{p}^{L,T} Pop(p)}}

  \trans{\iff}{expand definitions}
  
  \step{\rely r_1 \st EnvSt(p) \ch (LEmpty(p) \sqcap LDoPop(p))
    \sref_{p}^{L,T}  TryPop(p)^\omega \ch (Empty(p)
    \sqcap DoPop(p))}
  
  \trans{\follows}{distribute `;', $EnvSt(p)$ splits}
  
  \step{\rely r_1 \st EnvSt(p) \sref_{p}^{L,T} TryPop(p)^\omega \land
    \hfill \refstepcounter{equation} \label{eq:B1}
    (\arabic{equation})}

  \step{\rely r_1 \st LEmpty(p) \sref_{p}^{L,T} Empty(p) \land \hfill
        \refstepcounter{equation} \label{eq:B2} (\arabic{equation}) }

      \step{\rely r_1 \st EnvSt(p) \ch LDoPop(p) \sref_{p}^{L,T} DoPop(p)
        \hfill \refstepcounter{equation} \label{eq:B3}
        (\arabic{equation})}
\end{derivation}

\paragraph{Proof of \refeq{eq:B1}.}
This property holds in a similar manner to the proof of
\refeq{eq:30}. In particular, the proof holds because $TryPop(p)$ does
not modify any location within $\{Top, n_p\cdota val\} \cup SAddr$.

\paragraph{Proof of \refeq{eq:B2}.}
This property is trivial using monotonicity properties and
\refthm{thm:refine-guard}.

\paragraph{Proof of \refeq{eq:B3}.}

We strengthen $r_1$ so that it implies
\begin{eqnarray}
  \Box(\Always (\derefe Top = t_p \land pc_p = lt_7)
  &\ \    \imp\ \  &   
  \Always (tn_p = ptr.(\derefe Top)\mapsto nxt) \land
  IntFree.SAddr.p) \label{eq:9}   
\end{eqnarray}
By \refeq{eq:9}, if the global top value $Top$ matches the local copy
$t_p$, then the global next value $Top.nxt$ must be the same as the
local copy $tn_p$, and that there is no interference on the locations
within $SAddr$. Using this condition, we prove \refeq{eq:B3} as
follows.
\begin{derivation}
  \step{\enf{r_1}\st {ToCAS(p) \ch lt_7: CASOK_p( Top, t_p, (tn_p, ctr.t_p +
      1))}}

  \trans{\sqsupseteq_p^{T}}{expandin definition of $ToCAS(p)$}

  \step{\enf r_1 \st 
    \begin{array}[t]{@{}l@{}}
      t_p \asgn \derefe Top \ch 
      [ptr.t_p \neq null] \ch tn_p \asgn ptr.t_p \mapsto nxt \ch rv_p \asgn 
      ptr.t_p \mapsto key \ch  {} \\
      lt_7: CASOK_p( Top, t_p, (tn_p, ctr.t_p + 1))
    \end{array}
  }

  \trans{\sqsupseteq_p^{T}}{\refthm{thm:ABA} and \refeq{eq:24}}
  
  \step{\enf r_1  \st
    \begin{array}[t]{@{}l@{}}
      t_p \asgn \derefe Top \ch {} \\
      \left(\begin{array}[c]{@{}l@{}}
          \enf (\Always \neg \mcI.Top.p
          \land \derefe Top = t_p) \st
          \begin{array}[t]{@{}l@{}}
          {[}ptr.t_p \neq null{]} \ch {} \\
          tn_p \asgn ptr.t_p \mapsto nxt \ch {}\\
          rv_p \asgn
          ptr. t_p \mapsto key
        \end{array}
        \end{array}\right)
      \ch {}\\
      lt_7: CASOK_p( Top, t_p, (tn_p, ctr.t_p + 1))
    \end{array}
  }

  \trans{\sqsupseteq_p^{T}}{use $\Always (\derefe Top = t_p)$, then
    \reflem{lem:enf-prop-intro-ref}}
  \trans{}{$beh_{p,T}.CASOK_p( Top, \alpha, \beta) \entails
    \neg \mcW.t_p.p$} 

  \step{\enf r_1  \st
    \begin{array}[t]{@{}l@{}}
      t_p \asgn \derefe Top \ch {} \\
      \left(\begin{array}[c]{@{}l@{}}
          \enf \Always \neg \mcI.Top.p
          \st
          \begin{array}[t]{@{}l@{}}
            {[}ptr.(\derefe Top) \neq null{]} \ch {} \\
            tn_p \asgn ptr.(\derefe Top) \mapsto nxt \ch {} \\
            rv_p \asgn
            ptr.(\derefe Top) \mapsto key
          \end{array}
        \end{array}\right)
      \ch {}
      \\
      lt_7: CASOK_p( Top, t_p, (tn_p, ctr.(\derefe Top) + 1))
    \end{array}
  }

  \trans{\sqsupseteq_p^{T}}{\refeq{eq:36} using $r_1$
    splits, then weaken enforced property}

  \step{
    \begin{array}[t]{@{}l@{}}
      t_p \asgn \derefe Top \ch {} \\
      \left(\begin{array}[c]{@{}l@{}}
        \enf \Always \neg \mcI.Top.p \st
        \begin{array}[t]{@{}l@{}}
          [ptr.(\derefe Top) \neq null] \ch {} \\
          tn_p \asgn ptr.(\derefe Top) \mapsto nxt \ch {} \\
          rv_p \asgn
          ptr.(\derefe Top) \mapsto key
        \end{array}
      \end{array}\right)
      \ch {}
      \\
      \enf r_1  \st lt_7: CASOK_p( Top, t_p, (tn_p, ctr.(\derefe Top) + 1))
    \end{array}
  }

  \trans{\sqsupseteq_p^{T}}{\refthm{thm:replace-CAS} using
    \refeq{eq:9}, $CASOK_p(ae, \alpha, \beta) \entails IntFree.ae.p$}

  \step{
    \begin{array}[t]{@{}l@{}}
      t_p \asgn \derefe Top \ch {} \\
      \enf \Always \neg \mcI.Top.p \st 
      \begin{array}[t]{@{}l@{}}
        [ptr.(\derefe Top) \neq null] \ch 
        tn_p \asgn ptr.(\derefe Top)\mapsto nxt \ch rv_p \asgn
        ptr.(\derefe Top)\mapsto key
        \ch  {}\\
         Top \hasgn (ptr.(\derefe Top) \mapsto nxt, ctr.(\derefe Top) + 1)
      \end{array}
    \end{array}
  }
\end{derivation}

Using monotonicity of `;', the proof of $\refeq{eq:B3}$ reduces to the
following proof obligations.
\begin{eqnarray}
  \label{eq:6}
  EnvSt_p & \sref_p^{L,T} & t_p \asgn \derefe Top 
  \\
  \label{eq:12}
  {[}ptr.(\derefe Top) \neq null{]} & \sref_p^{L,T} & {[}ptr.(\derefe Top) \neq null{]} \ch tn_p \asgn
  ptr.(\derefe Top) \mapsto nxt
  % \\
  % rv_p \asgn
  % Top.key \ch Top \asgn Top.nxt  & \sref_p^{L,T} &
  % rv_p \asgn
  % Top.key \ch Top \asgn Top.nxt
\end{eqnarray}
The proof of \refeq{eq:6} is trivial because $t_p \notin L$. Property
\refeq{eq:12} holds as follows: 
\begin{derivation}
  \step{[ptr.(\derefe Top) \neq null]} 
  
  \trans{\sref_P^{L}}{definitions}

  \step{[ptr.(\derefe Top) \neq null] \ch (\Enf{\neg \Empty}{\Idle})}
  
  \trans{\sref_P^{L,T}}{\refthm{thm:intro-asgn} because $tn_p \notin
    L$}

  \step{[ptr.(\derefe Top) \neq null] \ch tn_p \asgn ptr.(\derefe Top) \mapsto nxt}

\end{derivation}

\OMIT{
Then we have
\begin{derivation}
  \step{beh_p.(lt_7: (\Enf{\Always \neg \mcI.Top.p }{[\derefe Top = t_p] \ch Top \asgn tn_p})) }
  
  \trans{\entails}{}

  \step{\Always \neg \mcI.Top.p \land }

  \step{\left(
    \begin{array}[c]{@{}l@{}}
\Eventually (\derefe Top = t_p \land \Always \neg \mcW_p.Top)
    \ch {} \\
    \exists k \st (\Eventually (k = tn_p) \land \Always \neg \mcW_p.Top)\ch \Always
    (Top = k)
  \end{array}
\right)}

  \trans{\entails}{using $\Always \neg \mcI.Top.p$}

  \step{
    \begin{array}[c]{@{}l@{}}
      \Always (\derefe Top = t_p)
      \ch 
      \exists k \st (\Eventually (k = tn_p) \land \Always (\derefe Top = t_p))\ch \Always
      (Top = k)
    \end{array}
  }

  \trans{\entails}{using rely condition \refeq{eq:9}}

  \step{
    \begin{array}[c]{@{}l@{}}
      true
      \ch 
      \exists k \st (\Eventually (k = tn_p) \land \Always (tn_p = Top.nxt))\ch \Always
      (Top = k)
    \end{array}}

  \trans{\entails}{logic}

  \step{
    \begin{array}[c]{@{}l@{}}
      true
      \ch 
      \exists k \st \Eventually (k = Top.nxt)\ch \Always
      (Top = k)
    \end{array}}

  \trans{\entails}{logic and $\Eventually c$ widens}

  \step{
    \begin{array}[c]{@{}l@{}}
      \exists k \st \Eventually (k = Top.nxt)\ch \Always
      (Top = k)
    \end{array}}

  \trans{\equiv}{definition}

  \step{beh_p.(Top \asgn Top.nxt)}

\end{derivation}

\begin{brijesh}
  This calculation above can probably be generalised to a lemma. Then
  the proof just becomes application of the lemma. 
\end{brijesh}
}

\subsection{Proof of \refeq{eq:64}}
\label{sec:proof-that-rest}

The rely condition $r_1$ is required to imply \refeq{eq:22},
\refeq{eq:24} and \refeq{eq:9}. We define the weakest possible
condition and obtain 
\begin{eqnarray*}
  r_1 & \sdef & \refeq{eq:22} \land
  \refeq{eq:24} \land \refeq{eq:9}
\end{eqnarray*}

To prove \refeq{eq:22}, we show that the condition below holds:
\begin{eqnarray*}
  \all q : P \bs \{p\} \st
  \begin{array}[t]{@{}l@{}}
    \Always (pc_q = ht_5 \imp n_q \neq n_p) \land
    \Always (pc_q = lt_7 \imp tn_q \neq n_p)
  \end{array}
\end{eqnarray*}
which ensures that process $q$ can never insert $n_p$ into the
queue. The proof of the formula above relies on the fact that $SAddr
\cap FAddr = \emptyset$ is an invariant of $TS(P)$. Invariance of $SAddr
\cap FAddr = \emptyset$ is straightforward to verify.

To show that processes $q \neq p$ satisfy \refeq{eq:9}, we must
consider commands executed by process $q$ that either make the
antecedent true or falsify the consequent of \refeq{eq:9}. The counter
for $Top$ is only incremented and hence process $q \neq p$ cannot make
the antecedent of \refeq{eq:9} true. Furthermore, the command in
process $q$ that falsifies the consequent (i.e., $CASOK_q( Top, t_q,
(tn_q, ctr.t_q + 1))$) also falsifies the antecedent.

\section{Conclusions and related work}
\label{sec:concl-relat-work}

Methods for verifying linearisability have received a large amount of
attention in the last few years. Herlihy and Wing's original paper use
possibilities and Owicki/Gries-style \cite{OG76} proof outlines, which
defines the set of possible abstract data structures that corresponds
to each point of interleaving. As we have already mentioned, Doherty
et al \cite{CDG05,DGLM04} use a simulation-based method using
input/output automata, Vafeiadis et al use a framework that combines
separation logic and rely/guarantee reasoning \cite{Vaf07,VHHS06} and
Derrick et al have developed refinement-based methods
\cite{DSW07,DSW11TOPLAS,DSW11}. O'Hearn et al develop a method using a
so-called hindsight lemma \cite{ORVYY10} and Jonsson presents a method
that uses refinement calculus \cite{Jon12}. A number of tool-based
methods have also been developed, but these often place restrictions
on the final implementation. For instance, Amit et al present static
analysis techniques \cite{ARRSY07}, Burckhardt et al \cite{BDMT10}
develop a tool for checking whether or not an algorithm is
\emph{deterministically linearisable} (so that future behaviour need
not be considered) and Vafeiadis has developed a tool that can be used
to verify linearisability for such deterministically linearisable
algorithms \cite{Vaf10}.  Verification of linearisability using
coarse-grained abstraction has been proposed by Turon and Wand, but
they do not show that the abstraction itself is linearisable
\cite{TW11}. Elmas et al \cite{EQSST10}, and separately Groves
\cite{Gro08} use a reduction-based method, but unlike our approach,
these methods are not compositional.

Despite this large set of results, due to the complexity of such
concurrent data structures a satisfactory scalable solution to
verification remains an open problem. The approach proposed by this
paper is to split a verification into two phases --- the first reduces
the size (and hence complexity) of the problem by showing (via a
series of small refinements) that the atomicity of an implementation
can made more coarse, leaving one with a simpler program that can be
verified to be linearisable. Note we have presented the verification
in a different order, i.e., shown linearisability of the abstraction
first. 

This paper presents a compositional interval-based method of verifying
linearisability that does not require one to identify the
linearisation points within the concrete code. Instead, we prove that
the concrete code implements a coarse-grained abstraction. Due to this
coarse granularity, the linearisation points are easier to identify
and the proof itself is simpler. Rely/guarantee-style rules together
with splits/joins properties are used to develop transformation
theorems, which are in turn used to decompose proof obligations. By
using an interval-based framework together with fractional permissions
we are able to model true concurrency between parallel processes.
This also enables reasoning at a finer level of atomicity than is
often allowed because we allow reasoning at the level of variable and
memory accesses during expression evaluation.

As B{\"a}umler et al point out, reasoning over interval allows one to
determine the future behaviour of a program, which in turn allows one
to sometimes avoid backwards reasoning \cite{BSTR11}, e.g., the
Michael and Scott queue \cite{Michael96}. Our experiments indicate
that interval-based reasoning via coarse-grained abstraction also
simplifies proofs of Heller et al's coarse grained lazy set algorithm
\cite{HHLMSS07}, which is known to have linearisation points outside
the operations being verified \cite{VHHS06,CGLM06,DSW11}. In the
terminology of Burckhardt et al, this corresponds to a
non-deterministically linearisable program, and hence lies outside the
scope of the tools in \cite{ARRSY07,BDMT10,Vaf10}. We believe that
interval-based reasoning allows generality beyond pre/post state
reasoning, and that such generalisations are necessary for taming the
increasing concurrency in everyday applications.

The methods we have presented have not yet been mechanised and this
remains the next obvious extension to this work. We conjecture that
the refinement-based framework will also be useful for a derivation,
which we aim to explore as part of future work. Such work would draw
on, for instance, the derivational approach proposed by Vechev and
Yahav \cite{VY08}.

%%% arxiv
% \begin{acknowledgements}
%   Brijesh Dongol and John Derrick are sponsored by EPSRC Grant
%   EP/J003727/1. We thank Lindsay Groves and Ian J. Hayes for their
%   helpful comments on an earlier draft. This paper has benefited from
%   the input of anonymous reviewers.
% \end{acknowledgements}
\paragraph{Acknowledgements.}
  Brijesh Dongol and John Derrick are sponsored by EPSRC Grant
  EP/J003727/1. We thank Lindsay Groves and Ian J. Hayes for their
  helpful comments on an earlier draft. This paper has benefited from
  the input of anonymous reviewers.

\bibliographystyle{plain}
\bibliography{thesis,jreferences}

\end{document}